\documentclass{article}

\usepackage[a4paper]{geometry}
\usepackage[english]{babel}
\usepackage{authblk}

\usepackage[dvipsnames]{xcolor}
\usepackage{amsmath}
\usepackage{amsthm}
\usepackage[font=small, labelfont=bf, labelsep=space]{caption}

\usepackage{aliascnt}

\theoremstyle{plain}
\newtheorem{theorem}{Theorem}

\newaliascnt{lemma}{theorem}
\newtheorem{lemma}[lemma]{Lemma}
\aliascntresetthe{lemma}

\newaliascnt{corollary}{theorem}

\aliascntresetthe{corollary}

\newaliascnt{proposition}{theorem}
\newtheorem{proposition}[proposition]{Proposition}
\aliascntresetthe{proposition}

\newaliascnt{exercise}{theorem}

\aliascntresetthe{exercise}

\newaliascnt{definition}{theorem}
\newtheorem{definition}[definition]{Definition}
\aliascntresetthe{definition}

\newaliascnt{conjecture}{theorem}

\aliascntresetthe{conjecture}

\newaliascnt{observation}{theorem}

\aliascntresetthe{observation}

\theoremstyle{definition}
\newaliascnt{example}{theorem}
\newtheorem{example}[example]{Example}
\aliascntresetthe{example}

\theoremstyle{remark}
\newaliascnt{note}{theorem}

\aliascntresetthe{note}
\newtheorem*{note*}{Note}

\newaliascnt{remark}{theorem}

\aliascntresetthe{remark}
\newtheorem*{remark*}{Remark}

\usepackage[%
  colorlinks = true,
  citecolor  = RoyalBlue,
  linkcolor  = RoyalBlue,
  urlcolor   = RoyalBlue,
  unicode,
]{hyperref}
\usepackage[nameinlink]{cleveref}

\crefname{lstlisting}{listing}{listings}
\Crefname{lstlisting}{Listing}{Listings}

\usepackage{listings}
\usepackage[T1]{fontenc}
\usepackage{url}
\usepackage[hyphenbreaks]{breakurl}
\usepackage[utf8]{inputenc}

\usepackage{graphicx} %
\usepackage{xparse}
\usepackage{xstring}
\usepackage{xspace}

\usepackage{bookmark}
\usepackage{metalogo}
\usepackage{etoolbox}

\usepackage[belowskip=0.4\baselineskip]{caption}
\usepackage{subcaption}

\usepackage{appendix}

\usepackage{algorithm}
\usepackage[noend]{algorithmic}

\usepackage{proof} %
\usepackage{amsmath} %
\usepackage{amssymb}
\usepackage{etoolbox}
\usepackage{mathtools}
\usepackage{stmaryrd}
\usepackage{mathpartir}

\usepackage{tikz}
\usetikzlibrary{positioning, quotes}
\usetikzlibrary{calc,intersections,through,backgrounds,matrix,patterns}
\usetikzlibrary{arrows, positioning, fit, shapes, backgrounds, patterns,
shapes.misc,shapes.geometric,decorations.markings}
\usetikzlibrary{tikzmark}

\newcommand{\ie}{i.e.\xspace}
\newcommand{\eg}{e.g.\xspace}
\newcommand{\cf}{cf.\xspace}

\newcommand{\etal}{et\,al.\xspace}

\makeatletter
\@ifundefined{st}{%
  
}{}
\@ifundefined{nd}{%
  
}{}
\@ifundefined{rd}{%
  
}{}
\makeatother

\newcommand{\mref}[2]{\hyperref[#2]{#1~\ref*{#2}}}
\definecolor{ultralightgray}{HTML}{eeeeee}

\newtoggle{REPORT}
\newtoggle{DEBUG}
\toggletrue{REPORT}
\toggletrue{DEBUG}

\definecolor{reportcolor}{HTML}{188781}
\definecolor{papercolor}{HTML}{536872}

\newtoggle{ANONYMOUS}
\togglefalse{ANONYMOUS}

\newtoggle{COMMENTS}
\toggletrue{COMMENTS}

\newcommand{\sts}{S2S\xspace}

\newcommand{\sccqpattern}{P}
\newcommand{\sccqtemplate}{H}

\newcommand{\vconcept}[1]{V_{#1}}

\newcommand{\rdftypea}{\texttt{a}}
\newcommand{\construct}{\texttt{CONSTRUCT}\xspace}
\newcommand{\select}{\texttt{SELECT}\xspace}

\newcommand{\var}{\operatorname{var}}

\newcommand{\prefixiri}[2]{\mathtt{#1{:}#2}}
\newcommand{\rdftype}{\prefixiri{rdf}{type}}
\newcommand{\metaprefixiri}[1]{\prefixiri{m}{#1}}
\newcommand{\metn}{\metaprefixiri{etn}}
\newcommand{\mnte}{\metaprefixiri{nte}}
\newcommand{\mnode}{\metaprefixiri{node}}
\newcommand{\medge}{\metaprefixiri{edge}}

\newcommand{\allshapes}{\mathcal{S}}
\newcommand{\shapesin}{\mathcal{S}_{\mathrm{in}}}
\newcommand{\shapesout}{\mathcal{S}_{\text{out}}}

\newcommand{\shaclvalid}[2]{\operatorname{valid}(#1,#2)}
\newcommand{\sparqleval}[2]{\llbracket #1 \rrbracket_{#2}}

\newcommand{\DLogics}{\mathcal{ALCHOI}}

\newcommand{\atomC}[2]{#1\,{:}\,#2}

\newcommand{\atomP}[3]{\atomC{(#1,#2)}{#3}}

\newcommand{\atomGenN}[1]{\atomC{#1}{\mnode}}
\newcommand{\atomGenE}[1]{\atomC{#1}{\medge}}
\newcommand{\atomPin}[2]{\atomC{(#1,#2)}{\mnte}}
\newcommand{\atomPout}[2]{\atomC{(#1,#2)}{\metn}}

\newcommand{\cvar}[1]{\overline{#1}^{\ConceptNames}}
\newcommand{\rvar}[1]{\overline{#1}^{\RoleNames}}
\newcommand{\ivar}[1]{\overline{#1}^{\IndividualNames}}

\newcommand{\atomGenp}[1]{\atomC{(#1,\ivar{#1})}{\rvar{#1}}}
\newcommand{\atomGenc}[1]{\atomC{#1}{\cvar{#1}}}
\newcommand{\atomCNot}[2]{\cvar{#1} \neq #2}
\newcommand{\atomPNot}[2]{\rvar{#1} \neq #2}

\newcommand{\ABox}{\mathcal{A}}
\newcommand{\TBox}{\mathcal{T}}
\newcommand{\Int}{{\mathcal{I}}}
\newcommand{\KB}{\mathcal{K}}

\newcommand{\ConceptNames}{\mathbf{C}}
\newcommand{\IndividualNames}{\mathbf{I}}
\newcommand{\RoleNames}{\mathbf{R}}
\newcommand{\Variables}{\mathbf{V}}

\newcommand{\concept}[1]{\texttt{#1}}
\newcommand{\prop}[1]{\texttt{#1}}

\newcommand{\voc}{\operatorname{voc}}
\newcommand{\graphout}{G_{\mathrm{out}}}
\newcommand{\graphin}{G_{\mathrm{in}}}

\newcommand{\Concept}[1]{\operatorname{C}_{#1}}

\newcommand{\vcg}{\operatorname{vcg}}

\newcommand{\graphmed}{G_{\mathrm{med}}}

\newcommand{\graphext}{G_{\mathrm{ext}}}
\newcommand{\graphvar}{G_{\Variables}}

\newenvironment{proofsketch}
{%
  \begin{proof}}
  {\end{proof}}

\newcommand{\propgraph}{PG\xspace}

\newcommand{\labelname}[1]{\texttt{#1}} %

\newcommand{\idname}[1]{\textbf{\texttt{#1}}}
\newcommand{\varname}[1]{\texttt{#1}}

\newcommand{\propname}[1]{\texttt{#1}} %
\newcommand{\stringvalue}[1]{\texttt{"#1"}} %

\newcommand{\gcoreEdge}[3]{(#1,#2){:}#3}
\newcommand{\gcoreNode}[1]{#1}

\newcommand{\gcoreEdgeS}[4]{(#1,#2){:}(#3, #4)}
\newcommand{\gcoreNodeS}[2]{(#1, #2)}

\newcommand{\gcoreEdgeF}[4]{(#1,#2){:}(#3, #4)}
\newcommand{\gcoreNodeF}[2]{(#1, #2)}

\newcommand{\whereLabelRelative}[1]{\ {:}#1}
\newcommand{\wherePropertyRelative}[1]{\ .#1}
\newcommand{\wherePropertyEqualsRelative}[2]{\ .#1 = #2}

\newcommand{\shapeassignment}{\ensuremath\Sigma\xspace}

\NewDocumentCommand{\iversonleft}{}%
{\ensuremath [\,}

\NewDocumentCommand{\iversonright}{}%
{\ensuremath \,]}

\NewDocumentCommand{\iverson}{m}%
{\ensuremath \iversonleft #1 \iversonright}

\NewDocumentCommand{\iversonerr}{O{(\phi,x)} m}%
{\ensuremath [\,#2\,]_{#1}}

\newcommand{\names}[1]{%
  \ifthenelse{\isempty{#1}}%
  {\ensuremath\mathrm{Names}\xspace}%
  {\ensuremath\mathrm{Names}(#1)}%
}%

\NewDocumentCommand{\evalpath}{O{\shapeassignment} O{n} O{G} m}%
{\ensuremath\llbracket #4 \rrbracket^{#1, #2, #3}}%

\NewDocumentCommand{\evalany}{O{\shapeassignment} O{x} O{G} m}%
{\ensuremath\llbracket #4 \rrbracket^{#1, #2, #3}}%

\NewDocumentCommand{\evaln}{O{\shapeassignment} O{n} O{G} m}%
{\ensuremath\llbracket #4 \rrbracket^{#1, #2, #3}}%

\NewDocumentCommand{\evalsn}{O{G} O{S} m}%
{\ensuremath\llbracket #3 \rrbracket^{#1, #2}}%

\NewDocumentCommand{\evalse}{O{G} O{S} m}%
{\ensuremath\llbracket #3 \rrbracket^{#1, #2}}%

\NewDocumentCommand{\evale}{O{\shapeassignment} O{e} O{G} m}%
{\ensuremath\llbracket #4 \rrbracket^{#1, #2, #3}}%

\NewDocumentCommand{\evalp}{O{\shapeassignment} O{(x,k)} O{G} m}%
{\ensuremath\llbracket #4 \rrbracket^{#1, #2, #3}}%

\NewDocumentCommand{\evalv}{O{\shapeassignment} O{v} O{G} m}%
{\ensuremath\llbracket #4 \rrbracket^{#1, #2, #3}}%

\NewDocumentCommand{\evalq}{O{G} m}%
{\ensuremath\llbracket #2 \rrbracket_{#1}}%

\NewDocumentCommand{\evali}{m}%
{\ensuremath\llbracket #1 \rrbracket_{\textrm{init}}}%

\NewDocumentCommand{\sevaln}{O{n} O{G} m}%
{\ensuremath\llbracket #3 \rrbracket_{#1, #2}}%

\NewDocumentCommand{\sevale}{O{e} O{G} m}%
{\ensuremath\llbracket #3 \rrbracket_{#1, #2}}%

\NewDocumentCommand{\nodeshape}{O{s_N} O{\phi_N} O{q_N}}%
{\ensuremath _N \langle #1,#2,#3 \rangle}

\NewDocumentCommand{\nodeshapeanon}{O{q_N} O{\phi_N}}%
{\ensuremath _N \langle #1,#2 \rangle}

\NewDocumentCommand{\nodeshapel}{O{s_N} O{\phi_N} O{q_N}}%
{\begin{align*}
_N \langle &#1,\\
           &#2,\\
           &#3 \rangle
\end{align*}}

\newcommand{\existsleft}[1]{\exists^{\leftarrow} #1}
\newcommand{\existsright}[1]{\exists^{\rightarrow} #1}

\newcommand\utimes{\mathbin{\ooalign{$\cup$\cr%
   \hfil\raise0.42ex\hbox{$\scriptscriptstyle\times$}\hfil\cr}}}
\newcommand\bigutimes{\mathop{\ooalign{$\bigcup$\cr%
   \hfil\raise0.36ex\hbox{$\scriptscriptstyle\boldsymbol{\times}$}\hfil\cr}}}

\NewDocumentCommand{\edgeshape}{O{s_E} O{\phi_E} O{q_E}}%
{\ensuremath _E \langle #1,#2,#3 \rangle}

\NewDocumentCommand{\edgeshapeanon}{O{q_E} O{\phi_E}}%
{\ensuremath _E \langle #1,#2 \rangle}

\NewDocumentCommand{\propertyshape}{O{s_P} O{\phi_P} O{q_P}}%
{\ensuremath _P \langle #1,#2,#3 \rangle \ }

\NewDocumentCommand{\csrule}{m m m}%
{$\rightarrow_{#1}$ & #2 & #3\\}

\definecolor{originalworld}{RGB}{86,180,233}
\definecolor{mirrorworld}{RGB}{213,94,0}
\definecolor{integrationhighlight}{RGB}{213,0,119}

\newcommand{\coremappingname}[2]{{_{#1 \mapsto #2}}}
\newcommand{\coremappingnameannotated}[3]{{_{#1 \mapsto #2}^{\mathit{#3}}}}

\newcommand{\coreevalt}[4]{\llbracket #3 \rrbracket_{#4}^{#2\text{-}#1}}
\newcommand{\coreevall}[3]{\llbracket #2 \rrbracket_{#3}^{#1}}

\newcommand{\cpart}[2]{\concept{V}_{\varname{#1}}^{\concept{#2}}}

\newcommand{\gcorename}{G\babelhyphen{nobreak}CORE\xspace}

\newcommand{\pgnames}{\mathit{pg}\xspace}

\newcommand{\pgexa}[1]{G_{#1}^{\pgnames}}

\newcommand{\rdfnames}{\mathit{rdf}\xspace}

\newcommand{\rdfexa}[1]{G_{#1}^{\rdfnames}}

\newcommand{\pref}[2]{\texttt{#1:#2}}
\newcommand{\triple}[3]{(#1,\ \allowbreak #2,\ \allowbreak #3)}
\newcommand{\triplei}[3]{(\concept{#1},\allowbreak\prop{#2},\allowbreak\concept{#3})}

\newcommand{\pgtordfname}{\coremappingname{\pgnames}{\rdfnames}}
\newcommand{\pgtordfnameannotated}[1]{\coremappingnameannotated{\pgnames}{\rdfnames}{#1}}
\newcommand{\pgtordf}[1]{\pgtordfname(#1)}
\newcommand{\pgtordfn}[1]{\pgtordfnameannotated{node}(#1)}
\newcommand{\pgtordfe}[1]{\pgtordfnameannotated{edge}(#1)}
\newcommand{\rdftopgname}{\pgtordfname^{-1}}
\newcommand{\rdftopg}[1]{\rdftopgname(#1)}

\newcommand{\toiriname}{\widehat{\_}}
\newcommand{\toiri}[1]{\widehat{#1}}

\newcommand{\simpleprogs}{sProGS\xspace}
\newcommand{\simpleprogss}{\mathit{pro}\xspace}

\newcommand{\progsexa}[1]{s_{#1}^{\simpleprogss}}
\newcommand{\progsexac}[2]{#1_{#2}^{\simpleprogss}}

\NewDocumentCommand{\progsevaln}{O{n} O{G} m}%
{\coreevalt{n}{\simpleprogss}{#3}{#1,#2}}%

\NewDocumentCommand{\progsevale}{O{e} O{G} m}%
{\coreevalt{e}{\simpleprogss}{#3}{#1,#2}}%

\NewDocumentCommand{\progsevalq}{O{G} m}%
{\coreevalt{q}{\simpleprogss}{#2}{#1}}%

\newcommand{\shaclname}{SHACL\xspace}
\newcommand{\shaclnames}{\mathit{shacl}\xspace}

\newcommand{\shaclexa}[1]{s_{#1}^{\shaclnames}}

\newcommand{\sprogstoshaclname}{\coremappingname{\simpleprogss}{\shaclnames}}
\newcommand{\sprogstoshaclnameannotated}[1]{\coremappingnameannotated{\simpleprogss}{\shaclnames}{#1}}
\newcommand{\sprogstoshacl}[1]{\sprogstoshaclname(#1)}
\newcommand{\sprogstoshaclnc}[1]{\sprogstoshaclnameannotated{nconstr}(#1)}
\newcommand{\sprogstoshaclnt}[1]{\sprogstoshaclnameannotated{ntarget}(#1)}
\newcommand{\sprogstoshaclec}[1]{\sprogstoshaclnameannotated{econstr}(#1)}
\newcommand{\sprogstoshaclet}[1]{\sprogstoshaclnameannotated{etarget}(#1)}

\newcommand{\eccqname}{ECCQ\xspace}
\newcommand{\eccqnames}{\mathit{eccq}\xspace}
\newcommand{\eccquery}[3]{\{ #1 \} \gets (\{ #2 \}, \{ #3\})}
\newcommand{\eccqpattern}{P}
\newcommand{\eccqtemplate}{H}
\newcommand{\eccqfilter}{F}
\newcommand{\eccqformal}{\eccqtemplate \gets (\eccqpattern, \eccqfilter)}

\newcommand{\eccqexa}[1]{q_{#1}^{\eccqnames}}

\newcommand{\eccqeval}[2]{\coreevall{\eccqnames}{#1}{#2}}
\newcommand{\eccqpatterneval}[2]{\coreevalt{p}{\eccqnames}{#1}{#2}}
\newcommand{\eccqtripleeval}[2]{\coreevalt{a}{\eccqnames}{#1}{#2}}

\newcommand{\gcname}{SGCQ\xspace}
\newcommand{\gcnames}{\mathit{sgcq}\xspace}

\newcommand{\gcformal}{\gctemplate \Leftarrow \gcpattern}
\newcommand{\gcformalname}{\gctemplate \Leftarrow \gcpattern}
\newcommand{\gcpattern}{\Gamma}
\newcommand{\gctemplate}{\mathcal{B}}

\newcommand{\gctemplatecomponent}{\beta}
\newcommand{\gcpatterncomponent}{\gamma}

\newcommand{\gcwhere}{\xi}
\newcommand{\gcset}{\oplus}
\newcommand{\gcremove}{\ominus}
\newcommand{\gcsr}{S}

\newcommand{\gcexa}[1]{q_{#1}^{\gcnames}}

\newcommand{\gceval}[2]{\coreevall{\gcnames}{#1}{#2}}
\newcommand{\gcceval}[2]{\coreevalt{t}{\gcnames}{#1}{#2}}
\newcommand{\gcweval}[2]{\coreevalt{e}{\gcnames}{#1}{#2}}
\newcommand{\gcmeval}[2]{\coreevalt{p}{\gcnames}{#1}{#2}}
\newcommand{\gcseval}[2]{\coreevalt{s}{\gcnames}{#1}{#2}}

\newcommand{\iqltemplate}{I_{T}}
\newcommand{\iqlpattern}{I_{P}}
\newcommand{\iqlformal}{\iqltemplate \leftarrow \iqlpattern}
\newcommand{\iqlname}{IQL\xspace}
\newcommand{\iqlnames}{\mathit{iql}\xspace}

\newcommand{\sgcoretoiqlname}{\coremappingname{\gcnames}{\iqlnames}}
\newcommand{\sgcoretoiqlnameannotated}[1]{\coremappingnameannotated{\gcnames}{\iqlnames}{#1}}
\newcommand{\sgcoretoiql}[1]{\sgcoretoiqlname(#1)}
\newcommand{\sgcoretoiqlc}[1]{\sgcoretoiqlnameannotated{template}(#1)}
\newcommand{\sgcoretoiqlm}[1]{\sgcoretoiqlnameannotated{pattern}(#1)}

\newcommand{\iqltoeccqname}{\coremappingname{\iqlnames}{\eccqnames}}
\newcommand{\iqltoeccqnameannotated}[1]{\coremappingnameannotated{\iqlnames}{\eccqnames}{#1}}
\newcommand{\iqltoeccq}[1]{\iqltoeccqname(#1)}
\newcommand{\iqltoeccqc}[1]{\iqltoeccqnameannotated{template}(#1)}
\newcommand{\iqltoeccqm}[1]{\iqltoeccqnameannotated{pattern}(#1)}

\newcommand{\sgcoretoeccqname}{\coremappingname{\gcnames}{\eccqnames}}

\newcommand{\sgcoretoeccq}[1]{\sgcoretoeccqname(#1)}

\usepackage{stmaryrd}
\usepackage{trimclip}
\makeatletter
\DeclareRobustCommand{\shortmaparrow}{%
  \mathrel{\mathpalette\short@to\relax}%
}
\newcommand{\short@to}[2]{%
  \mkern2mu
  \hspace{-0.08cm}\clipbox{{.5\width} 0 0 0}{$\m@th#1\vphantom{+}{\rightarrow}$}\hspace{-0.05cm}%
  }

\newcommand{\eA}{\concept{A}} %
\newcommand{\eB}{\concept{B}} %
\newcommand{\eE}{\concept{E}} %

\newcommand{\ea}{\prop{a}} %
\newcommand{\eb}{\prop{b}} %

\newcommand{\eAi}{\namein{\eA}}
\newcommand{\eBi}{\namein{\eB}}
\newcommand{\eEi}{\namein{\eE}}

\newcommand{\eAo}{\nameout{\eA}}
\newcommand{\eBo}{\nameout{\eB}}
\newcommand{\eEo}{\nameout{\eE}}

\newcommand{\exvb}{\concept{V}_\varname{x}}
\newcommand{\eyvb}{\concept{V}_\varname{y}}

\newcommand{\eevb}{\concept{V}_\varname{e}}

\newcommand{\exv}{\namevar{\exvb}}
\newcommand{\eyv}{\namevar{\eyvb}}

\lstdefinelanguage{SHACL}{
    keywords = {NodeShape,targetClass,property,path,class,minCount,targetSubjectsOf,a,xsd,sh,string}
}

\lstdefinelanguage{SPARQL}{
    keywords = {CONSTRUCT,WHERE,a,FILTER}
}

\lstdefinelanguage{GRAPHQL}{
    keywords = {person}
}

\lstdefinelanguage{GREMLIN}{
    keywords = {hasLabel,as,outE,inV,select}
}

\lstdefinelanguage{PSEUDO}{
    keywords = {CONSTRUCT,MATCH,WHERE,for,in,print}
}

\lstdefinelanguage{GCORE}{
    keywords = {CONSTRUCT,WHERE,FROM,MATCH,AND,SET,REMOVE,RETURN}
}

\lstdefinestyle{progs}{
  showstringspaces=false,
  basicstyle=\footnotesize\ttfamily,
  keywordstyle=\bfseries\color{black},
  commentstyle=\itshape\color{white!40!black},
  stringstyle=\color{white!40!black},
  morecomment=[l]{//},
  morestring=[b]",
  morekeywords={NODE,EDGE},
}

\lstdefinestyle{asp}{
  showstringspaces=false,
  basicstyle=\footnotesize\ttfamily,
  keywordstyle=\bfseries\color{black},
  commentstyle=\itshape\color{white!40!black},
  stringstyle=\color{white!40!black},
  morecomment=[l]{//},
  morestring=[b]",
  morekeywords={edge,label,property,constraint,path,nodeshape,greaterEq,assignN,assignE,satisfiesN,satisfiesE,edgeshape,not,targetN,targetE,node,min},
}

\tikzset{new node/.style={circle,minimum size=1.0em,inner sep=0.2em}}
\tikzset{corenode/.style={circle,minimum size=1.0em,inner sep=0.2em}}
\tikzset{outernode/.style={inner sep=0.3em}}

\SetSymbolFont{stmry}{bold}{U}{stmry}{m}{n}

\definecolor{colorgin}{HTML}{f58231}
\definecolor{colorgmed}{HTML}{800000}
\definecolor{colorgout}{HTML}{e6194B}
\definecolor{colorgvar}{HTML}{4363d8}

\newcommand{\namein}[1]{\textcolor{colorgin}{#1}}

\newcommand{\nameout}[1]{\textcolor{colorgout}{\ddot{#1}}}

\newcommand{\namevar}[1]{\textcolor{colorgvar}{#1}}

\usepackage{booktabs}

\tikzset{
  pglabel/.style={
     text centered,
     anchor=center,
     fill=white,
     text opacity=1,
     inner sep=0.2ex,
     font=\sffamily\footnotesize
  },
  arrout/.style={
    draw=black,
    text=black,
    arrows={-stealth},
    thick,
	pos=0.45
  },
}

\newcommand{\fullproofsref}{Full proofs are available in the appendix.}

\title{Transforming Shape Schemas with Composable Property-Graph Queries (Extended Version)}

\author[1]{Philipp Seifer}
\author[2]{Daniel Hernández}
\author[1]{Ralf Lämmel}
\author[2,3]{Steffen Staab}

\affil[1]{The Software Languages Team\\University of Koblenz, Germany\\\newline\texttt{\{pseifer, laemmel\}@uni-koblenz.de}}
\affil[2]{Institute for Artificial Intelligence\\University of Stuttgart, Germany\\\texttt{\{daniel.hernandez,steffen.staab\}@ki.uni-stuttgart.de}}
\affil[3]{University of Southampton, UK}

\date{\today}

\begin{document}

\maketitle
\thispagestyle{empty}

\begin{abstract}
  Property graphs may be constrained by schemas that inform both query engines and human users
  about the shape of valid data, enforcing a contract between data provider and consumer. 
  Composable property-graph queries transform input graphs into output graphs. 
  Then, the question arises of which schema can be expected after one (or several) transformation steps.
  We investigate how schema constraints can be inferred 
  given an input schema and a transforming query.
  Specifically, we propose a reasoning procedure that, 
  given an input schema in ProGS and a query in \gcorename infers an output schema.
  Since graph updates will happen frequently, our inference procedure does not rely on graph instances, 
  such that the computed output schema applies to all graphs originating from any input graph 
  complying with the input schema.
  Related work has addressed this problem for SPARQL \construct queries, 
  encoding it in Description Logics (DLs) so that the output schema is entailed 
  by axioms inferred from input schema and queries.
  Property graphs and their queries, however, complicate the matter, 
  as property graphs feature label and property annotations as well as first-class edges.
  Thus, reification has to be used in one way or another, though 
  available DLs lack the means to encode such features directly.
  We approach this novel challenge via a family of mappings for 
  i) property graphs reified in RDF, aligned with 
  ii) a mapping from ProGS to SHACL and 
  iii) a mapping from \gcorename to SPARQL \construct queries.
  In this manner, schema inference for property graphs becomes manageable, as we break apart the problem
  through the extra mapping layer and utilize efficient DL reasoners.
  We develop the metatheory regarding the soundness of inferred schema constraints 
  and the semantic equivalence of mapped schemas and queries.
\end{abstract}

\newpage

\section{Introduction}

Composable property-graph queries transform graphs into graphs.
In this paper, we are interested in schema inference for the results of such composable queries. 
For RDF graphs, the W3C proposes a composable query language, namely SPARQL~\cite{sparql} 
with its \construct~\cite{DBLP:conf/icdt/KostylevRU15} clause.
For \emph{property graphs} (PGs), which we aim at in this paper, 
the recent GQL~\cite{ISO39075} standard is not composable (as yet~\cite{gqlshort}).
Accordingly, we leverage \gcorename~\cite{DBLP:conf/sigmod/AnglesABBFGLPPS18}, an 
earlier \emph{composable} query language for property graphs that inspired GQL.

Property-graph schemas are used to constrain PGs; the constraints can be enforced by validation. 
There are schema languages, sometimes also called \emph{shape} languages, for PGs, such as the 
\emph{Property Graph Shapes Language} (ProGS)~\cite{DBLP:conf/semweb/SeiferLS21} 
or PG-Schema~\cite{DBLP:journals/pacmmod/AnglesBD0GHLLMM23}.
As far as schema constraints of the inference approach in this paper are concerned, 
both ProGS and PG-Schema are viable candidates. 
Pragmatically, we commit to ProGS, as it is more directly aligned with the established schema language for RDF, 
\ie SHACL~\cite{bibshacl}, which is readily covered by the foundation we rely on.
In particular, the correspondence between SHACL and description logics (DL) is well 
understood~\cite{DBLP:conf/www/Seifer0LS24,DBLP:conf/lpnmr/BogaertsJB22}. 
Thus, we can leverage efficient reasoner implementations for inference.
We follow the terminology of ProGS and SHACL and refer to schemas as \emph{shapes}.

We refer to shapes validating all possible result graphs of a query as \emph{output shapes}.
Similarly, we use the term \emph{input shapes} to refer to the set of shapes over
the input graphs of a query.
The inference of output shapes is a challenging aspect of the composable query notion:
input shapes do not carry over directly to the output but do indeed impact it.
Additionally, to obtain sound shapes for all possible query results,
inference must rely on the \emph{input shapes} and the \emph{query} as inputs
rather than individual result graphs.
Access to these statically inferred shapes facilitates numerous use cases,
even when considering compositions of multiple queries:
they enable tools such as type systems for language-embedded queries
(compare~\cite{DBLP:conf/semweb/LeinbergerSSLS19}) among others
(\eg,~\cite{DBLP:conf/semweb/ThapaG23,DBLP:conf/mtsr/PacielloTBVS22})
that rely on shapes.
For query pipelines that are executed repeatedly, such shapes need
only be inferred once, which is particularly useful in cases where some level
of manual review or selection of shapes is desired.
Finally, even in data transformation or integration use cases, where
result graphs are stored in databases, they complement shapes inferred from graph
instances and can provide insights to developers without query execution.

\begin{figure}[ht]
   \centering
   \begin{tikzpicture}[
    every label/.style={inner sep=0.0em},
    x=3.0cm,
    y=2.0cm
]
    \newcommand{\mdim}{0.3cm}

    \node[label={90:{\color{originalworld}\texttt{ProGS}}}] (a) at (0,1) {$\shapesin$};
    \node[label={90:{\color{originalworld}\texttt{PG}}}] (b) at (1,1)  {\color{gray} $G_{i}^{1}$\ldots$G_{i}^{n}$};
    \node[label={90:{\color{originalworld}\texttt{G-CORE}}}] (c) at (2,1)  {$q$};
    \node[label={90:{\color{originalworld}\texttt{PG}}}] (d) at (3,1) {\color{gray} $G_{o}^{1}$\ldots$G_{o}^{n}$};
    \node[label={90:{\color{originalworld}\texttt{ProGS}}}] (e) at (4,1) {$\shapesout$};

    \path[->]
        (b) edge [] node [below] {} (c)
        (c) edge [] node [below] {} (d);

    \path[-]
        (b) edge [draw=gray, double] node [below] {} (a)
        (d) edge [draw=gray, double] node [below] {} (e);

    \node (am) [below = \mdim of a] {};

    \node[] (ki) at (2.0,0.5)  {$k$};

    \path[->]
        (a) edge [draw=originalworld, line width=0.3mm] node [right] {} (ki)
        (c) edge [draw=originalworld, line width=0.3mm] node [right] {} (ki)
        (ki) edge [draw=originalworld, line width=0.3mm] node [right] {} (e);

\end{tikzpicture}%
\caption{Schematic overview of the problem: Input shapes $\shapesin$ validate a (possibly infinite) 
         set of input graphs $G_{i}^{1}$\ldots$G_{i}^{n}$ which are valid inputs for the query $q$.
         The query produces, as its result, a set of corresponding result graphs 
         $G_{o}^{1}$\ldots$G_{o}^{n}$.
         Just as $\shapesin$ validates all input graphs, $\shapesout$ validates 
         all output graphs. 
         The set of $\shapesout$ is constructed from some formalization $k$, 
         which is derived from $\shapesin$ and $q$.}
\label{trfm:fig:overviewone}
\end{figure}
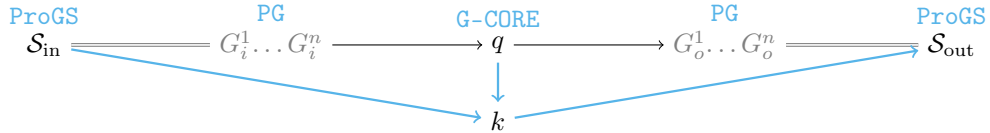

\Cref{trfm:fig:overviewone} shows the input (shapes $\shapesin$ and query $q$) 
and the output (shapes $\shapesout$) of shape inference, 
with the respective languages annotated in {\color{originalworld}blue}. 
For these languages, we will later define the relevant subsets we consider in this work.
The figure also hints in {\color{gray}gray} at the set of input graphs 
(conforming to shapes $\shapesin$, indicated by {\color{gray}gray} double lines) 
and corresponding output graphs (conforming to $\shapesout$), 
with black arrows representing input and output of query execution.
These graphs are included in the figure only for the sake of clarity; 
our approach does not depend on any concrete graphs.
Our approach is rendered using {\color{originalworld}blue} arrows: 
we encode in $k$ the information that allows us to deduce which shapes are in $\shapesout$.

To introduce $k$, we establish a common language that can express the constraints
of the input shapes $\shapesin$ and how inputs and outputs relate
to each other over the query $q$.
Given such a language, we then define an algorithm for constructing $k$ 
from both $\shapesin$ and the query $q$
and for inferring the new shapes in $\shapesout$ from $k$.

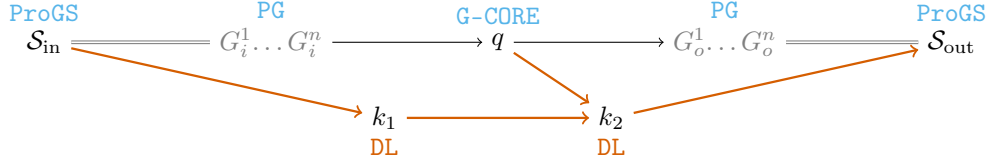
\begin{figure}[ht]
   \centering
   \begin{tikzpicture}[
    every label/.style={inner sep=0.0em},
    x=3.0cm,
    y=2.0cm
]
    \newcommand{\mdim}{0.3cm}

    \node[label={90:{\color{originalworld}\texttt{ProGS}}}] (a) at (0,1) {$\shapesin$};
    \node[label={90:{\color{originalworld}\texttt{PG}}}] (b) at (1,1)  {\color{gray} $G_{i}^{1}$\ldots$G_{i}^{n}$};
    \node[label={90:{\color{originalworld}\texttt{G-CORE}}}] (c) at (2,1)  {$q$};
    \node[label={90:{\color{originalworld}\texttt{PG}}}] (d) at (3,1) {\color{gray} $G_{o}^{1}$\ldots$G_{o}^{n}$};
    \node[label={90:{\color{originalworld}\texttt{ProGS}}}] (e) at (4,1) {$\shapesout$};

    \path[->]
        (b) edge [] node [below] {} (c)
        (c) edge [] node [below] {} (d);

    \path[-]
        (b) edge [draw=gray, double] node [below] {} (a)
        (d) edge [draw=gray, double] node [below] {} (e);

    \node (am) [below = \mdim of a] {};

    \node[label={-90:{\color{mirrorworld}\texttt{DL}}}] (ki) at (1.5,0.5)  {$k_1$};
    \node[label={-90:{\color{mirrorworld}\texttt{DL}}}] (ko) at (2.5,0.5)  {$k_2$};

    \path[->]
        (ki) edge [draw=mirrorworld, line width=0.3mm] node [right] {} (ko)
        (a) edge [draw=mirrorworld, line width=0.3mm] node [right] {} (ki)
        (ko) edge [draw=mirrorworld, line width=0.3mm] node [right] {} (e)
        (c) edge [draw=mirrorworld, line width=0.3mm] node [right] {} (ko);

\end{tikzpicture}%
\caption{Schematic overview of a possible solution where constraints are encoded in description logics.}
\label{trfm:fig:overviewtwo}
\end{figure}

One possibility for encoding these constraints is \emph{description logics} (DL).
\Cref{trfm:fig:overviewtwo} indicates where DL is used with {\color{mirrorworld}red} arrows:
we first encode the input shapes $\shapesin$ as a set of DL axioms 
and then combine them with a second set of axioms inferred from the query $q$.
This approach benefits from relying on an existing formalism and,
in practical applications, from established and efficient reasoners
for implementing the required decision procedures.
In prior work~\cite{DBLP:conf/www/Seifer0LS24}, we also relied on a DL to solve
a similar inference problem for RDF graphs and SPARQL \construct queries.
However, there is a drawback: Common DLs do not support properties with key-value annotations.
Thus, we would have to either define a new variant\,---\,for which no efficient reasoners 
exist\,---\,or we have to use some form of reification to encode properties in an existing DL.

Since the first option is not desirable, we choose the second.
However, instead of encoding reification directly in DL, we introduce an intermediate RDF layer.
This provides the following benefits:
foremost, it clearly separates the \emph{reification} from the \emph{DL encoding and inference} steps
of our solution, making each part easier to formalize, prove, and implement.
Secondly, we can build on the aforementioned prior work for the encoding and inference step;
we still need to modify and extend the algorithm.
Finally, a mapping between the involved shape and query languages has not been formally explored,
which we achieve in this work.

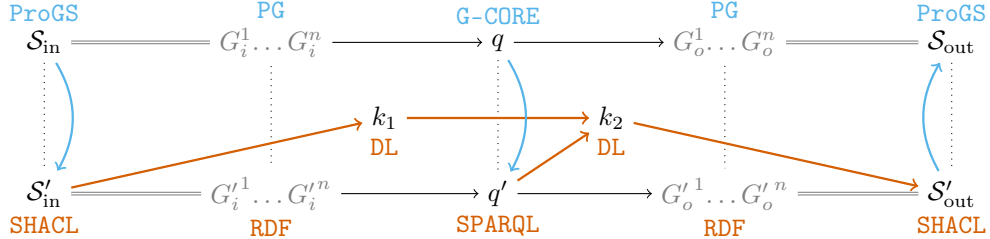
\begin{figure}[ht]
   \centering
   \begin{tikzpicture}[
    every label/.style={inner sep=0.0em},
    x=3.0cm,
    y=2.0cm
]
    \newcommand{\mdim}{0.3cm}

    \node[label={90:{\color{originalworld}\texttt{ProGS}}}] (a) at (0,1) {$\shapesin$};
    \node[label={90:{\color{originalworld}\texttt{PG}}}] (b) at (1,1)  {\color{gray} $G_{i}^{1}$\ldots$G_{i}^{n}$};
    \node[label={90:{\color{originalworld}\texttt{G-CORE}}}] (c) at (2,1)  {$q$};
    \node[label={90:{\color{originalworld}\texttt{PG}}}] (d) at (3,1) {\color{gray} $G_{o}^{1}$\ldots$G_{o}^{n}$};
    \node[label={90:{\color{originalworld}\texttt{ProGS}}}] (e) at (4,1) {$\shapesout$};

    \path[->]
        (b) edge [] node [below] {} (c)
        (c) edge [] node [below] {} (d);

    \path[-]
        (b) edge [draw=gray, double] node [below] {} (a)
        (d) edge [draw=gray, double] node [below] {} (e);

    \node[label={-90:{\color{mirrorworld}\texttt{SHACL}}}] (f) at (0,0) {$\shapesin'$};
    \node[label={-90:{\color{mirrorworld}\texttt{RDF}}}] (g) at (1,0) {\color{gray} ${G_i'}^{1}$\ldots${G_i'}^{n}$};
    \node[label={-90:{\color{mirrorworld}\texttt{SPARQL}}}] (h) at (2,0) {$q'$};
    \node[label={-90:{\color{mirrorworld}\texttt{RDF}}}] (i) at (3,0) {\color{gray} ${G_o'}^{1}$\ldots${G_o'}^{n}$};
    \node[label={-90:{\color{mirrorworld}\texttt{SHACL}}}] (j) at (4,0) {$\shapesout'$};

    \path[->]
        (g) edge [] node [above] {} (h)
        (h) edge [] node [above] {} (i);

    \path[-]
        (g) edge [draw=gray, double] node [above] {} (f)
        (i) edge [draw=gray, double] node [above] {} (j);

    \path[-]
        (a) edge [dotted] node [right] {} (f)
        (b) edge [dotted] node [right] {} (g)
        (c) edge [dotted] node [right] {} (h)
        (d) edge [dotted] node [right] {} (i)
        (e) edge [dotted] node [right] {} (j);

    \node (am) [below = \mdim of a] {};

    \node[label={-90:{\color{mirrorworld}\texttt{DL}}}] (ki) at (1.5,0.5)  {$k_1$};
    \node[label={-90:{\color{mirrorworld}\texttt{DL}}}] (ko) at (2.5,0.5)  {$k_2$};

    \path[->]
        (ki) edge [draw=mirrorworld, line width=0.3mm] node [right] {} (ko)
        (a) edge [draw=originalworld, bend left, line width=0.3mm] node [right] {} (f)
        (c) edge [draw=originalworld, bend left, line width=0.3mm] node [right] {} (h)
        (f) edge [draw=mirrorworld, line width=0.3mm] node [right] {} (ki)
        (ko) edge [draw=mirrorworld, line width=0.3mm] node [right] {} (j)
        (j) edge [draw=originalworld, bend left, line width=0.3mm] node [right] {} (e)
        (h) edge [draw=mirrorworld, line width=0.3mm] node [right] {} (ko);

\end{tikzpicture}%
\caption{Schematic overview of the involved languages, mappings, and algorithms.
         The upper PG layer is mapped to the bottom RDF layer, 
         where individual mappings are indicated using dotted lines.
         Inference and reasoning exclusively operate on RDF.
         The {\color{originalworld}blue} arrows show how mappings are utilized.}
\label{trfm:fig:overview}
\end{figure}

\Cref{trfm:fig:overview} shows the full solution we propose.
Utilizing reification (indicated with dotted lines) at all levels allows us to step down from a 
PG-based language family ({\color{originalworld}upper half} of \Cref{trfm:fig:overview}) 
to an RDF-based one ({\color{mirrorworld}lower half} of \Cref{trfm:fig:overview}).
We use RDF for representing the data by reifying first-class edges with identities
and properties in PG as ordinary RDF triples.
We use SHACL as the shape language, accounting for this reification.
This allows us to rely on DL for the formalization and implementation of shape inference,
since round tripping between DL axioms and SHACL shapes is possible
(\cf~\cite{DBLP:conf/www/Seifer0LS24,DBLP:conf/lpnmr/BogaertsJB22}).
Finally, we use SPARQL \construct as the composable query language
to which we map.
Here, the definition of a sound mapping is more complex, since in addition to the reification, we must
also account for the differences in query language semantics:
where labels and properties are copied by default in \gcorename and explicitly removed, SPARQL relies
on explicit triple patterns and filters.
The data reification, shape-to-DL mapping, and shape inference approaches have been 
studied~\cite{DBLP:conf/grades/KhayatbashiFH22,DBLP:journals/pacmmod/AnglesBD0GHLLMM23,DBLP:conf/www/Seifer0LS24},
whereas the query mapping has not.

We build on and extend our existing algorithm~\cite{DBLP:conf/www/Seifer0LS24} that infers SHACL shapes 
for results of SPARQL \construct queries on RDF graphs by constructing 
a DL knowledge base entailing valid output shapes. 
To define our query mapping, we must extend the SPARQL fragment considered in~\cite{DBLP:conf/www/Seifer0LS24} with more
\emph{generic} triple patterns where variables extend over concept or role names
to capture the semantic properties of \gcorename.
This leads to a number of required adaptations of the method presented in this work,
since the knowledge base must encode the potential presence of previously unexpected 
concept or role names.
The dataflow of the extension to~\cite{DBLP:conf/www/Seifer0LS24} is indicated by {\color{mirrorworld}red} arrows in \Cref{trfm:fig:overview}.
It is important to note that our mapping approach serves the formalization and 
implementation of schema inference; whether the queries are executed on a PG database 
or an RDF store is a separate question.

\subsection{Contributions}

In summary, the contributions of this paper are as follows:

\begin{enumerate}
   \item We determine the impact of PG-specific expressiveness when it comes to shape inference for 
         composable graph queries. Thus, we handle edges that have identities, as well as nodes and edges 
         annotated with labels and key-value properties that persist, implicitly, in results.
   \item Despite common DLs lacking the corresponding language constructs, we resolve the mismatch through 
         reification, which carries over to structures of shapes and queries. 
   \item Instead of reification at the DL level, we define an intermediate layer,
         introducing mappings from ProGS to SHACL and \gcorename to SPARQL \construct over 
         PGs reified in RDF.
   \item We develop the metatheory regarding the soundness of inferred shapes 
         and the semantic equivalence of mapped shapes and queries.
\end{enumerate}

\subsection{Roadmap of the Paper}\label{trfm:p:roadmap}
\Cref{trfm:sec:foundation} introduces the underlying formalisms, as well as our running examples.
In \Cref{trfm:sec:formal}, we formalize the essential problem addressed in this work. 
\Cref{trfm:sec:map} and \Cref{trfm:sec:extend} solve the problem in two steps: 
first, we develop a family of mappings from property graph abstractions to RDF graph abstractions (\Cref{trfm:sec:map}).
On these grounds, we develop an inference method, which is an extension of~\cite{DBLP:conf/www/Seifer0LS24} (\Cref{trfm:sec:extend}).
\Cref{trfm:sec:meta} develops the corresponding metatheory.
In \Cref{trfm:app:impl}, we briefly outline our prototype implementation of this approach.
Finally, \Cref{trfm:sec:related} discusses related work while \Cref{trfm:sec:conclusion} concludes the paper.

Several appendices are referenced by the paper for additional details, 
such as extended examples and full proofs.

\section{Foundations}
\label{trfm:sec:foundation}

We first introduce the graph data models, query languages, and shape-based validation 
languages relevant to our work; they can be separated into two groups, 
namely \emph{PG-related} languages and 
\emph{RDF-related} languages 
pertaining to the upper and lower layers of \Cref{trfm:fig:overview}.

\subsection{RDF Graphs}

We define RDF graphs based on the W3C specification~\cite{rdf}, excluding blank nodes and literals;
we consider the latter as ordinary IRIs to simplify the formalization of the data model.
We assume four pairwise disjoint, finite\footnote{This is not a restriction to the problem, since the sets are arbitrary,
but it simplifies a few definitions.} but arbitrary, and sufficiently large sets of IRIs, 
namely concept names $\ConceptNames$, individual names $\IndividualNames$, role names $\RoleNames$, 
and the special role name $\{\rdftype\}$.
For the sake of consistency, we use description-logic terminology throughout all abstractions
(\eg, \emph{concept name} instead of \emph{class}).
Commonly, RDF uses so-called prefixes to abbreviate (parts of) IRIs.
In formal definitions, we express most IRIs through ordinary symbols, 
with a few exceptions, such as $\rdftype$.
In examples, we also omit prefixes for the example domain;
that is, we simply write $\concept{Agent}$ instead of $\prefixiri{ex}{Agent}$.
Similarly, we sometimes write $\rdftypea$ instead of $\rdftype$.

\begin{definition}[Simple RDF Graph]
  An \emph{RDF graph} $G$ is a finite set of triples of the form 
  $\triple{a}{p}{b}$ or $\triple{a}{\rdftype}{A}$
  where $a,b \in \IndividualNames$, 
  $A \in \ConceptNames$, and 
  $p \in \RoleNames$.
\end{definition}

\noindent
\begin{minipage}{0.50\textwidth}
\begin{example}
  \label{trfm:ex:basicrdf}
  Consider the following set\\
  of triples as an example RDF graph
  $\rdfexa{1} = $\\
  $\{\triplei{s}{\rdftypea}{Agent}, \triplei{s}{name}{Smith}, \triplei{p}{\rdftypea}{Person},$\\
  $\triplei{s}{obs}{p}\}$.
\end{example}
\end{minipage}%
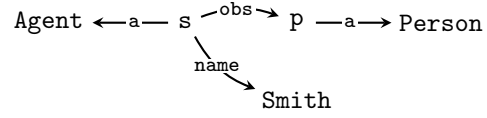
\begin{minipage}{0.50\textwidth}
  \centering
  \begin{tikzpicture}[
    every label/.style={inner sep=0.0em}
]
    \node[new node] (s) {\concept{s}};
    \node[new node] (a) [left = 1.0cm of s] {\concept{Agent}};
    \node[new node] (p) [right = 1.0cm of s]  {\concept{p}};
    \node[new node] (pp) [right = 1.0cm of p] {\concept{Person}};
    \node[new node] (n) [below = 0.2cm of p]  {\concept{Smith}};

    \draw [-] (s) edge[arrout, bend left=20] node[pglabel]{\prop{obs}} (p);
    \draw [-] (s) edge[arrout] node[pglabel]{\rdftypea} (a);
    \draw [-] (p) edge[arrout] node[pglabel]{\rdftypea} (pp);
    \draw [-] (s) edge[arrout, bend right=20] node[pglabel]{\prop{name}} (n);

\end{tikzpicture} 
\end{minipage}

\subsubsection{SPARQL Queries}
\label{trfm:ss:sccq}

We define a subset of SPARQL queries referred to as \emph{Extended Conjunctive \construct Queries} (\eccqname) 
in analogy to \emph{Simple Conjunctive \construct Queries} (as defined in~\cite{DBLP:conf/www/Seifer0LS24}) 
in \Cref{trfm:def:eccq:syntax}.
We choose this conjunctive fragment, as is commonly done, to keep the scope of our algorithms manageable;
we extend\footnote{\eccqname is not a true extension of the language presented 
  in~\cite{DBLP:conf/www/Seifer0LS24}, since we omit certain patterns, \eg, $\atomC{a}{C}$. 
  While they could be easily included, making \eccqname a true extension of this language, 
  they are not required for the remainder of this work.}
  the fragment of SPARQL chosen in~\cite{DBLP:conf/www/Seifer0LS24} with the language features required 
  for our intended mappings.

To this end, we introduce \emph{generic} patterns such as $\atomC{x}{y}$, 
written \texttt{?x a ?y} in concrete SPARQL syntax, where both $x$ and $y$ are variables.
Here, we need to introduce an additional syntactic restriction for generic variables, as we will
discuss in more depth later:
if the query includes the atomic query pattern $\atomC{x}{y}$, then $y$ must not occur again in 
another atomic query pattern.
That is, we can copy all concept names of $x$ generically but cannot introduce restrictions on this meta level;
for example, a query with patterns $\{\atomC{x}{y}, \atomC{z}{y}\}$ is not allowed,
since it would constrain the bindings of $x$ and $z$ through $y$.
Thus, we write $\atomGenc{x}$ to indicate that the generic variable $\cvar{x}$ is syntactically 
dependent on $x$ in this way.
We also introduce certain filter expressions for only these generic variables written $\cvar{x}$,
filtering certain concept or role names from their bindings.

\begin{definition}[Syntax of \eccqname]
  \label{trfm:def:eccq:syntax}
  We define atomic patterns $t$ and filter expressions $f$ as follows:
  \begin{equation*}
    t \Coloneqq\ \atomC{x}{A} \mid \atomP{x}{y}{p} \mid \atomP{x}{a}{p} \mid \atomGenc{x} \mid \atomGenp{x} \quad\quad\quad
    f \Coloneqq\ \atomCNot{x}{A} \mid \atomPNot{x}{p}
  \end{equation*}
  with $A \in \ConceptNames$, $a \in \IndividualNames$, $p \in \RoleNames$, and
  variables $x, y, \cvar{x}, \rvar{x}, \ivar{x} \in \Variables$.
  Let the function \emph{vars} denote the set of all variables 
  that occur in a set of atomic patterns or filter expressions.
  An \eccqname $q$ is defined as $\eccqformal$
  where \emph{template} $\eccqtemplate$ and \emph{pattern} $\eccqpattern$ consist of finite sets
  of atomic patterns $t$ 
  and the \emph{filters} $\eccqfilter$ of a finite set of filter expressions $f$ such that 
  $\operatorname{vars}(\eccqfilter) \subseteq \operatorname{vars}(\eccqpattern)$. 
  For each atomic pattern $\atomGenc{x}$ (and similarly $\atomGenp{x}$), there is a unique variable 
  $\cvar{x}$ (similarly $\rvar{x}$ and $\ivar{x}$) per ordinary variable $x$, 
  which does not occur again in the query, except
  in filter conditions.
  If a generic atomic pattern of this form occurs in $\eccqtemplate$, the same pattern must 
  also occur in $\eccqpattern$.
\end{definition}

To formalize the semantics of \eccqname in \Cref{trfm:def:eccq:semantics}, we define valuations 
$\mu : \Variables \rightarrow \IndividualNames \cup \ConceptNames \cup \RoleNames$ 
as mappings from variables to individual, concept, and role names.
Note that the syntactic restrictions on $\cvar{x}$ and $\rvar{x}$ ensure that the respective
types of these bindings cannot be mixed up, since these variables can only occur in the correct context.
Two valuations $\mu_1$ and $\mu_2$ are \emph{compatible}, denoted $\mu_1 \sim \mu_2$, if for every 
variable $x$ occurring in both $\mu_1$ and $\mu_2$, $\mu_1(x) = \mu_2(x)$.
A finite set of valuations is denoted $\Omega$.
We define the \emph{join} of two sets of valuations $\Omega_1$ and $\Omega_2$ as 
$\Omega_1 \bowtie \Omega_2 = \{\mu_1 \cup \mu_2 \mid \mu_1 \in \Omega_1, \mu_2 \in \Omega_2, \mu_1 \sim \mu_2 \}$.

Note that in \Cref{trfm:def:eccq:syntax} we do not include the usual syntactic restriction 
$\operatorname{vars}(\eccqtemplate) \subseteq \operatorname{vars}(\eccqpattern)$.
We therefore need the final clause in \Cref{trfm:def:eccq:semantics} about variables from atomic
patterns $t \in \eccqtemplate$ that do not occur in $\operatorname{vars}(\eccqpattern)$.
Here, we assume the generation of fresh IRIs for such variables, which will be required by our mappings later.
While this is a deviation from ordinary SPARQL semantics\footnote{Usually, SPARQL \construct templates
can feature blank node identifiers (\url{https://www.w3.org/TR/sparql12-query/\#templatesWithBNodes})
to this effect.}, our variant simplifies the formalization by avoiding the need to represent
blank nodes in our RDF graph model and throughout the paper.
In practical implementations, ordinary blank nodes should be used to comply with standard SPARQL
semantics; indeed, our implementation generates SPARQL queries featuring blank node identifiers 
in their templates as well.

\begin{definition}[Semantics of \eccqname]
  \label{trfm:def:eccq:semantics}
  The result of evaluating an \eccqname $\eccqformal$ over an RDF graph $G$ is the RDF graph denoted 
  $\eccqeval{\eccqformal}{G}$ and defined as
  \[
    \eccqeval{\eccqformal}{G} = \{\mu(t) \mid \mu \in \eccqpatterneval{\eccqpattern, F}{G}, t \in \eccqtemplate \}
  \]
  where $\eccqpatterneval{P, F}{G} =\{\mu \mid \mu \in\ \bowtie_{t \in P} \eccqtripleeval{t}{G},\, \text{for all}\, f \in F : \mu \vdash f\}$ and
  $\eccqtripleeval{t}{G} = \{\mu \mid \mu(t) \in G\}$.
  $f$ is a filter condition
  such that $\mu \vdash (\atomCNot{x}{A})$ is true iff $\mu(\cvar{x})$ is not $A$,
  and $\mu \vdash (\atomPNot{x}{p})$ is true iff $\mu(\rvar{x})$ is not $p$.
  In a slight abuse of notation, we write $\mu(t)$ 
  to mean $\mu(\atomC{u}{A}) = \triple{\mu(u)}{\rdftype}{A}$,
  and similarly for the remaining cases.
  If a variable $x$ in $t \in \eccqtemplate$ does not occur in $\operatorname{vars}(\eccqpattern)$,
  then $\mu(x)$ produces a fresh IRI. %
  This IRI is unique for each variable $x$ and each unique mapping $\mu$.
\end{definition}

\noindent
\begin{minipage}{0.70\textwidth}
  \begin{example}
    \label{trfm:ex:basicsccq}
    Consider the 
    \eccqname $\eccqexa{1} = \eccquery{\atomC{\varname{x}}{\concept{Agent}},%
    \atomC{\varname{y}}{\concept{POI} }}{\atomP{\varname{x}}{\varname{y}}{\prop{obs}}}$,
    selecting pairs of observers and observed before constructing a new graph labeling observers 
    as $\concept{Agent}$ and observed as $\concept{POI}$.
    When evaluating this query on graph $\rdfexa{1}$ from \Cref{trfm:ex:basicrdf},
    written $\sparqleval{\eccqexa{1}}{\rdfexa{1}}$,
    we first match the pattern $\eccqpattern$ on the graph, resulting in the single mapping
    $\mu$ where $\mu(\varname{x}) = s$ and $\mu(\varname{y}) = p$.
    Thus, the template $\eccqtemplate$ is instantiated only once, namely for $\mu$,
    resulting in the graph $\{\triplei{s}{\rdftypea}{Agent}, \triplei{p}{\rdftypea}{POI}\}$.
    
    \hspace*{4mm} Using generic patterns, we could, in this particular case, rewrite the query as
    $\eccquery{\atomGenc{\varname{x}},\atomC{\varname{y}}{\concept{POI} }}{\atomGenc{\varname{x}},\atomP{\varname{x}}{\varname{y}}{\prop{obs}}}$,
    copying all concepts for bindings for $\varname{x}$ instead of explicitly adding $\concept{Agent}$.
    Since $s$ has only the concept $\concept{Agent}$ in $\eccqexa{1}$, evaluation would
    yield the same result graph.
  \end{example}

\end{minipage}%
\begin{minipage}{0.30\textwidth}
  \centering
  \begin{tikzpicture}[
    every label/.style={inner sep=0.0em}
]
    \node[new node] (s) {\concept{s}};
    \node[new node] (a) [right = 1.0cm of s] {\concept{Agent}};
    \node[new node] (p) [below = 1.0cm of s]  {\concept{p}};
    \node[new node] (pp) [right = 1.0cm of p] {\concept{POI}};

    \draw [-] (s) edge[arrout] node[pglabel]{\rdftypea} (a);
    \draw [-] (p) edge[arrout] node[pglabel]{\rdftypea} (pp);

\end{tikzpicture} 
\end{minipage}

\subsubsection{SHACL Shapes}

SHACL is a validation language for RDF graphs, where a shape consists of two components:
the target query selecting a subset of nodes and a constraint these nodes must conform to.
We present here the Seifer \etal~\cite{DBLP:conf/www/Seifer0LS24} definition of SHACL shapes,
utilizing description logics to define the semantics of target queries and constraints
based on work by Bogaerts \etal~\cite{DBLP:conf/lpnmr/BogaertsJB22}.

According to this definition, SHACL shapes are $\DLogics$ axioms, \ie target query subsumed by the constraint, 
both encoded as concept descriptions~(\cref{sts:def:dlogics}), and their semantics can be reduced to 
checking consistency with a \emph{validation knowledge base} constructed for a given RDF graph $G$ by 
its unique-name, closed-world, and domain-closure assumptions.
We refer to~\cite{DBLP:conf/dlog/2003handbook} for the definition of
the semantics of $\DLogics$.
We give here the syntax of $\DLogics$ axioms, which can express a subset of SHACL, and
the core idea of the validation semantics introduced in~\cite{DBLP:conf/www/Seifer0LS24};
we refer to that paper for the full derivation and proofs.

\begin{definition}[$\DLogics$ Axioms]
  \label{sts:def:dlogics}
  \emph{Concept descriptions} are defined as follows:
  \begin{align*} 
    C &\Coloneqq
     \top \mid \bot \mid A \mid \{a\} \mid C \sqcap C \mid C \sqcup C
     \mid \neg C \mid \exists \rho.C
     \mid \forall \rho.C\\
     \rho &\Coloneqq p \mid p^-
  \end{align*} 
  where $A$, $a$, and $p$ stand for concept names, individual names, and role names, respectively,
  and $\top$ and $\bot$ are two special concept names.
  Given two concept descriptions $C$ and $D$, and two role descriptions $\rho_1$ and $\rho_2$,
  $C \sqsubseteq D$ and $\rho_1 \sqsubseteq \rho_2$ are $\DLogics$ \emph{axioms}, which we also call \emph{shapes}.
  The left-hand side ($C$ and $\rho_1$) we also call the \emph{target} and the right-hand side ($D$ and $\rho_2$) 
  the \emph{constraint} of the shape.
\end{definition}

Semantically, SHACL shapes target specific nodes of a graph, 
relying on constraints to validate these selected nodes.
A consistency check for the subsumption axioms defined in \Cref{sts:def:dlogics}
with the validation knowledge base constructed for a given RDF graph $G$ 
(\Cref{sts:def:validation-rdf-semantics}) encodes these semantics.
\Cref{trfm:ex:basicshacl} demonstrates the intuition.

\begin{definition}[Validation Semantics (Adapted from~\cite{DBLP:conf/www/Seifer0LS24})]%
  \label{sts:def:validation-rdf-semantics}
  The axioms of a simple RDF graph $G$, 
  denoted $\TBox_G$, are the TBox consisting 
  of the following $\DLogics$ axioms:
  \begin{enumerate}
  \item \emph{Domain Closure Assumption (DCA)}:
    $\top \equiv \bigsqcup_{a \in \IndividualNames}\{a\}$.
  \item \emph{Unique Name Assumption (UNA)}:
   $\{a\} \sqcap \{b\} \equiv \bot$, for each pair of distinct\\individual names $a, b \in \IndividualNames$.
  \item \emph{Closed-World Assumption (CWA)}: 
    \begin{itemize}
      \item $A \equiv \bigsqcup_{\atomC{a}{A}\in G}\{a\}$, for each concept name $A \in \ConceptNames$, 
      \item $\exists p.\{a\} \equiv \bigsqcup_{\atomP{b}{a}{p}\in G} \{b\}$, and
      \item $\exists p^-.\{a\} \equiv \bigsqcup_{\atomP{a}{b}{p}\in G} \{b\}$, for each role name
        $p \in \RoleNames$ and each individual name $a \in \IndividualNames$.
    \end{itemize}
  \end{enumerate}
  $(\TBox_G, G)$ is the \emph{validation knowledge base} of $G$. 
  A graph $G$ is \emph{proof-valid} according to a set $\Sigma$ of $\DLogics$ axioms if and only if $\Sigma$ is consistent with the validation knowledge base of $G$, \ie the knowledge base $(\TBox_G \cup \Sigma, G)$ admits a model.
\end{definition}

We write $\shaclvalid{G}{S}$ to indicate that $G$ is proof-valid regarding the set of shapes $S$
and $\shaclvalid{G}{s}$ to indicate that $G$ is proof-valid regarding a single shape $s$.
We use $S$ and $\Sigma$ interchangeably.
In \Cref{trfm:ex:basicshacl} we give an example shape and outline both the intuitive
reason why this shape validates the running example graph and the formal reason
considering the validation knowledge base.

\begin{example}
  \label{trfm:ex:basicshacl}
  Consider $\shaclexa{1} = \concept{Agent} \sqsubseteq \exists \prop{obs}.\concept{Person}$ as an example.
  Intuitively, the shape expresses that all agents must observe at least one person.
  The graph $\rdfexa{1}$ (\Cref{trfm:ex:basicrdf}) is valid regarding $\shaclexa{1}$ since all targets, \ie $s$, which is
  an instance of $\concept{Agent}$, conform to the constraint, given both
  $\triplei{s}{\prop{obs}}{p} \in \rdfexa{1}$ and 
  $\triplei{p}{\rdftypea}{\concept{Person}} \in \rdfexa{1}$.

  The validation knowledge base, constructed from the 
  closed-world 
  ($\concept{Agent} \equiv \{\concept{s}\}$,
  $\concept{Person} \equiv \{\concept{p}\}$, 
  $\exists\prop{obs}.\{p\} \equiv \{\concept{s}\}$, and 
  $\exists\prop{obs}^-.\{s\} \equiv \{\concept{p}\}$),
  domain-closure ($\{\concept{s}\} \sqcup \{\concept{p}\} \equiv \top$), 
  and unique-name assumptions ($\{\concept{s}\} \sqcap \{\concept{p}\} \equiv \bot$) 
  for $\rdfexa{1}$ is consistent with this subsumption axiom, too:
  essentially, there is the concept $\{\concept{s}\}$ which is equal to $\concept{Agent}$, 
  and also the range of $\concept{obs}$;
  in turn, the domain of $\prop{obs}$ is $\{\concept{p}\}$, which is also equal to $\concept{Person}$.
\end{example}

\subsection{Property Graphs}

Property graphs (\Cref{trfm:def:pg}) differ from RDF graphs in that edges have identities, 
and both nodes and edges can be annotated with sets of labels and sets of key-value pairs, 
the so-called properties.
There are various notions of property graphs; we base our work largely
on~\cite{DBLP:conf/sigmod/AnglesABBFGLPPS18}, that is, the graph model used for G-CORE,
which is very similar to the graph model used in the context of 
Cypher\footnote{\url{https://opencypher.org/}} as well.
We do not include first-class paths, however.

We define a set of labels $L=L_N \cup L_E$ where $L_N$ is a finite set of node labels 
and $L_E$ a finite set of edge labels;
a finite set of property names (or keys) $K=K_N \cup K_E$ where $K_N$ is a finite set of node keys 
and $K_E$ a finite set of edge keys;
and a finite set of literal values $V$.
Without loss of generality, we assume that literal values are strings.
Finally, let $\textrm{FS}(s)$ denote all finite subsets of a set $s$, including
the empty set.
See \Cref{trfm:ex:basicpg} for the running example PG.

\begin{definition}[Property Graph]
  \label{trfm:def:pg}
  A \emph{property graph} (\propgraph) is a tuple $G=(N,E,\rho,\lambda,\sigma)$.
  $N$ and $E$ denote two disjoint, finite sets of node and edge identifiers, respectively.
  We write $n \in N$, $e \in E$, and $u \in N \cup E$.
  $\rho : E \rightarrow (N \times N)$ is a total function assigning edges to nodes;
  $\lambda : (N \cup E) \rightarrow \textrm{FS}(L)$ is a total function assigning labels to nodes or edges; 
  and $\sigma : (N \cup E)\times K\rightarrow \textrm{FS}(V)$ is a total function assigning properties 
  to nodes or edges and for which a set of tuples $(x,k)\in(N\cup E) \times K$ exists 
  such that $\sigma(x,k)\ne\emptyset$.
\end{definition}

\noindent
\begin{minipage}{0.70\textwidth}
  \begin{example}
    \label{trfm:ex:basicpg}
    Consider $\pgexa{1} = (N_1, E_1, \rho_1, \lambda_1, \sigma_1)$, where
    $N_1 = \{\idname{1},\idname{2}\}$ and $E_1 = \{\idname{10}\}$ are the sets of node and edge IDs.
    A single edge is defined as $\rho_1(\idname{10}) = (\idname{1},\idname{2})$.
    Labels are defined as
    $\lambda_1(\idname{1}) = \{\labelname{Agent\,}\}$,
    $\lambda_1(\idname{2}) = \{\labelname{Person\,}\}$, and
    $\lambda_1(\idname{10}) = \{\labelname{obs\,}\}$.
    Finally, we give the non-empty cases for the property function as
    $\sigma_1(\idname{1}, \labelname{name}) = \{\stringvalue{Smith}\}$ and
    $\sigma_1(\idname{10}, \labelname{since}) = \{\stringvalue{2002}\}$.

    $\quad$ Compare this graph to the RDF graph $\rdfexa{1}$ in \Cref{trfm:ex:basicrdf}:
    RDF graphs do not support properties on edges; thus, there is no direct equivalent for 
    $\sigma_1(10, \labelname{since})= \{\stringvalue{2002}\}$ in RDF.
    Otherwise, the graphs are equivalent.
  \end{example}
\end{minipage}%
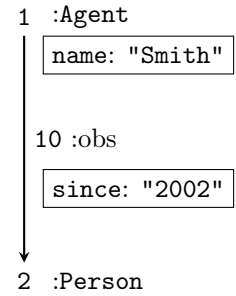
\begin{minipage}{0.30\linewidth}
  \centering
  \begin{tikzpicture}[
    every label/.style={inner sep=0.0em}
    ]

    \node[new node, label={0:\ :\concept{Agent}}] (s) {\idname{1}};
    \node[new node, label={0:\ :\concept{Person}}] (p) [below = 3.0cm of s]  {\idname{2}};

    \node[draw,minimum height = 0.5cm, minimum width = 1cm] (pp1) [below right = 0.1cm of s] {\propname{name}: \stringvalue{Smith}};
    \node[draw,minimum height = 0.5cm, minimum width = 1cm] (pp2) [below = 1.25cm of pp1] {\propname{since}: \stringvalue{2002}};

    \path[->]
        (s) edge [arrout] node [right] {\idname{10}\ :obs} (p);
\end{tikzpicture}
\end{minipage}

\subsubsection{Property Graph Queries}
Composable property graph queries return graphs and function similarly to \eccqname queries.
We define \emph{Simple \gcorename Queries} (\gcname) as a subset 
of \gcorename~\cite{DBLP:conf/sigmod/AnglesABBFGLPPS18}.
Our conjunctive fragment (\cref{trfm:def:g:syntax}) supports matching multiple atomic patterns,
including the most common constraints on labels and properties,
and constructing a new graph from multiple atomic patterns
while explicitly adding or removing labels and properties.
In short, the fragment allows for selecting a subgraph, reshaping it, and adding constant elements.
It does not support whole-graph operations, such as the union of a constructed graph with the original graph.

The basic semantics of \gcname are similar to \eccqname:
a \emph{pattern} $\gcpattern$ defines variables and constrains their
bindings through clauses $\gcwhere$.
These bindings are used to construct a new graph, specified in the \emph{template} $\gctemplate$,
together with \emph{set} ($\gcset$) and \emph{remove} ($\gcremove$) clauses for explicitly adding
or removing labels and properties.
Unlike in the case of \eccqname, labels and properties persist in the result graph by default.

We give here in short the syntax and semantics of \gcname adapted from~\cite{DBLP:conf/sigmod/AnglesABBFGLPPS18}, 
with the following change:
since we allow only for a subset of filter expressions, we push these filters into the node (or edge) pattern.
This is equivalent to the concrete syntax of \gcorename where we would write, \eg, \texttt{(x:Person)} to filter all nodes \texttt{x} for the label \texttt{Person}.

\begin{definition}[Syntax of \gcname]
  \label{trfm:def:g:syntax}
  A \gcname $\gcformalname$ consists of a set $\gcpattern$ of atomic components 
  $\gcpatterncomponent$ (\emph{pattern}) and a set $\gctemplate$ of atomic components $\gctemplatecomponent$
  (\emph{template}), defined by the following grammar:
  \begin{align*}
    \gcpatterncomponent & \Coloneqq \gcoreNodeF{x_n}{\gcwhere} \mid \gcoreEdgeF{x_n}{x_n}{x_e}{\gcwhere} &
    \gctemplatecomponent & \Coloneqq \gcoreNodeS{x_n}{\gcsr} \mid \gcoreEdgeS{x_n}{x_n}{x_e}{\gcsr} 
  \end{align*}
  Here, $x_n$ is a node variable, and $x_e$ is an edge variable.
  $\gcwhere$ is a set of expressions of the form
  $\whereLabelRelative{l}$, $\wherePropertyRelative{k}$, or $\wherePropertyEqualsRelative{k}{v}$,
  where $l \in L$ is a label, $k \in K$ is a property name, and $v \in V$ is a property value.
  $\gcsr$ is a set of assignments for setting
  ($\gcset l$ and $\gcset k = v$) or removing ($\gcremove l$ and $\gcremove k$) labels or property names.
  We also write $\gcsr_L$ for referring to operations involving labels and $\gcsr_P$ for referring to operations involving properties.
  Furthermore, we write, \eg, $\gcremove_x l \in \gctemplate$ to indicate that 
  $\gcoreNodeS{x}{\gcsr} \in \gctemplate$ and $\gcremove l \in \gcsr$.
  As an additional syntactical constraint of the template, edge variables that occur in the pattern 
  may only occur when their associated node variables also occur in the template.
\end{definition}

For the semantics of \gcname, we again consider valuations $\mu$ on node or edge variables $x_n$ or $x_e$, 
so that $\mu(x_n) \in N$ and $\mu(x_e) \in E$, and similar definitions for compatibility $\sim$ and the \emph{join} operation $\bowtie$ as for $\eccqname$.
We also define the union of two PGs $G_1 \cup G_2$ as $(N_1 \cup N_2, E_1 \cup E_2, \rho, \lambda, \sigma)$
where 
$\forall e \in E_1 \cup E_2 : \rho(e) = \rho_1(e)\ \textrm{if}\ e \in E_1\ \textrm{else}\ \rho_2(e)$,
$\forall x \in N_1 \cup N_2 \cup E_1 \cup E_2, k \in K: \lambda(x) = \lambda_1(x) \cup \lambda_2(x) $,
and $\sigma(x, k) = \sigma_1(x,k) \cup \sigma_2(x,k)$.
The semantics for evaluating \gcorename queries is given in \Cref{trfm:def:g:semantics},
with the intuition outlined in \Cref{trfm:ex:basicsgc}.

\begin{definition}[Semantics of \gcname]
  \label{trfm:def:g:semantics}
  The result of evaluating a \gcname $g = \gcformalname$ over a PG $G$ is defined as the PG constructed by
  \[
    \gceval{q}{G} = \gcceval{\gctemplate}{\Omega, G}\ \textrm{where}\ \Omega = \gcmeval{\gcpattern}{G}.
  \]
  That is, similarly to \eccqname queries, sets of bindings $\Omega$ are constructed from the pattern
  and fill out the template.
  $\gcmeval{\gcpattern}{G}$ and $\gcceval{\gctemplate}{\Omega, G}$ are
  specified in \Cref{trfm:def:g:semantics:match} and \Cref{trfm:def:g:semantics:construct} 
  based on~\cite{DBLP:conf/sigmod/AnglesABBFGLPPS18}, respectively.
\end{definition}

\begin{definition}[\gcname Pattern Semantics]
  \label{trfm:def:g:semantics:match}
  The set of bindings $\Omega$ resulting from evaluation of a \gcname pattern $\gcpattern$ over a graph $G$, 
  written $\gcmeval{\gcpattern}{G}$, is defined by the following cases:
  \begin{minipage}{0.55\textwidth}
    \begin{align*}
      \gcmeval{\gcoreNodeF{x}{\gcwhere}}{G} &= \{ \mu\ |\ \mu(x) \in N,\\
              & \qquad\qquad\forall \gcwhere_i \in \gcwhere : \gcweval{\gcwhere_i}{\mu(x), G} = \top\}\\
      \gcmeval{\gcoreEdgeF{x}{y}{z}{\gcwhere}}{G} &= \{ \mu\ |\ \mu(x), \mu(y) \in N, \mu(z) \in E,\\
              & \qquad\qquad \rho(\mu(z)) = (\mu(x), \mu(y)),\\
              &\qquad\qquad\forall \gcwhere_i \in \gcwhere : \gcweval{\gcwhere_i}{\mu(z), G} = \top\}\\
      \gcmeval{\gcpattern}{G} &=\ \bowtie_{\gcpatterncomponent \in \gcpattern} \gcmeval{\gcpatterncomponent}{G}
    \end{align*}
  \end{minipage}
  \begin{minipage}{0.45\textwidth}
    \begin{align*}
      \gcweval{\whereLabelRelative{l}}{u, G} & = \top\ \mathrm{iff}\ l \in \lambda(u)\\
      \gcweval{\wherePropertyRelative{k}}{u, G} & = \top\ \mathrm{iff}\ \sigma(u, k) \ne \emptyset\\
      \gcweval{\wherePropertyEqualsRelative{k}{v}}{u, G} & = \top\ \mathrm{iff}\ v \in \sigma(u, k)
    \end{align*}
  \end{minipage}
\end{definition}

\begin{definition}[\gcname Template Semantics]
  \label{trfm:def:g:semantics:construct}
  Evaluation of a \gcname template $\gctemplate$ over a graph $G$, with bindings $\Omega$, 
  written $\gcceval{\gctemplate}{\Omega, G}$, is defined by the following cases:
  \begin{align*}
    \gcceval{\gcoreNodeS{x}{\gcsr}}{\mu, G} &= \{\{v\}, \emptyset, \emptyset, \gcseval{\gcsr_L}{v,G}, \gcseval{\gcsr_P}{v,G}\}\\
    \gcceval{\gcoreEdgeS{x}{y}{z}{\gcsr}}{\mu, G} &= \{\{v, u\}, \{e\}, \{e \shortmaparrow (v,u)\}, \gcseval{\gcsr_L}{e,G}, \gcseval{\gcsr_P}{e,G}\}\\
    \gcceval{\gctemplate}{\Omega, G} &=\ \cup_{\gctemplatecomponent \in \gctemplate, \mu \in \Omega} \gcceval{\gctemplatecomponent}{\mu, G} 
  \end{align*}
  where $v = \mu(x)$, $u = \mu(y)$, and $e = \mu(z)$
  if the variables are in the domain of $\mu$, or fresh identities otherwise.
  Evaluation of set and remove clauses is defined as
  \begin{align*}
    \gcseval{\gcsr_L}{o_N,G} & = (\lambda_{o_N} \cup \{o_N \shortmaparrow l_N \mid \gcset l_N \in \gcsr_L\}) \backslash \{o_N \shortmaparrow l_N \mid \gcremove l_N \in \gcsr_L\} \\
    \gcseval{\gcsr_P}{o_N,G} & = (\sigma' \cup \sigma_{\gcset}) \backslash \sigma_{\gcremove}
  \end{align*}
  where
  \begin{align*}
    \sigma' & = \{(o_N, k_N) \shortmaparrow v \mid (o_N, k_N) \shortmaparrow v \in \sigma \wedge (o_N, k_N) \shortmaparrow v' \not\in \sigma_{\gcset}\}\\
    \sigma_{\gcset} & = \{(o_N, k_N) \shortmaparrow v \mid \gcset k_N = v \in \gcsr_P\}\\
    \sigma_{\gcremove} & = \{(o_N, k_N) \shortmaparrow v \mid (o_N, k_N) \shortmaparrow v \in \sigma \wedge \gcremove k_N \in \gcsr_P\}
  \end{align*}
\end{definition}

\noindent
\begin{minipage}{0.75\textwidth}
  \begin{example}
    \label{trfm:ex:basicsgc}
    Recall the \eccqname $\eccqexa{1}$ from \Cref{trfm:ex:basicsccq}.
    A \emph{seemingly} similar \gcname could be defined as
    \[
      \gcexa{0} = \{
        \gcoreNodeS{\varname{x}}{\{\gcset \labelname{Agent}\}}, 
      \gcoreNodeS{\varname{y}}{\{\gcset \labelname{POI}\}}\}
        \Leftarrow 
        \{\gcoreEdge{\varname{x}}{\varname{y}}{(\varname{e}, \{\whereLabelRelative{\labelname{obs}}\})}\}
    \]

    Evaluating $\gceval{\gcexa{0}}{\pgexa{1}}$ (\cf~\Cref{trfm:ex:basicpg}) produces first
    the single mapping $\mu$ with $\mu(\varname{x}) = \idname{1}$, $\mu(\varname{y}) = \idname{2}$,
    and $\mu(\varname{e}) = \idname{10}$;
    then, the template constructs the graph
    $\pgexa{0} = (\{\idname{1},\idname{2}\}, \emptyset, \emptyset, \lambda_0, \sigma_0)$ with
    $\lambda_0(\idname{1}) = \{\labelname{Agent}\}$,
    $\lambda_0(\idname{2}) = \{\labelname{Person}, \labelname{POI}\}$, and
    $\sigma_0(\idname{1}, \labelname{name}) = \{\stringvalue{Smith}\}$.
    That is, all labels and properties persist in the output graph.
    The following update to $\gcexa{0}$ has equivalent semantics to the \eccqname query:
    \[
      \gcexa{1} = \{\gcoreNodeS{\varname{x}}{\{\gcremove \labelname{name}\}}, \gcoreNodeS{\varname{y}}{\{\gcset \labelname{POI}, \gcremove \labelname{Person}\}}\} \Leftarrow 
      \{\gcoreEdge{\varname{x}}{\varname{y}}{(\varname{e}, \{\whereLabelRelative{\labelname{obs}}\})}\}
    \]
    Evaluation produces the graph
    $(\{\idname{1},\idname{2}\}, \emptyset, \emptyset, \lambda_1, \emptyset)$ with
    $\lambda_1(\idname{1}) = \{\labelname{Agent}\}$ and
    $\lambda_1(\idname{2}) = \{\labelname{POI}\}$, as we would expect (\cf~\Cref{trfm:ex:basicsccq}).
  \end{example}
\end{minipage}%
\begin{minipage}{0.25\textwidth}
  \centering
  \begin{tikzpicture}[
    every label/.style={inner sep=0.0em}
    ]

    \node[new node, label={0:\ :\labelname{Agent}}] (s) {\idname{1}};
    \node[new node, label={0:\ :\labelname{POI}}] (p) [below = 1.0cm of s]  {\idname{2}};
\end{tikzpicture}

\end{minipage}

\subsubsection{Simple Property Graph Shapes (\simpleprogs)}

We define the syntax and semantics of a subset 
of the \emph{Property Graphs Shapes Language}~\cite{DBLP:conf/semweb/SeiferLS21} (ProGS)
we call \emph{Simple Property Graph Shapes Language} (\simpleprogs).
The shape language is heavily inspired by SHACL and adapted for property graphs
by differentiating between node and edge shapes, introducing constraints for property annotations,
and for constraining the nodes (or edges, respectively) reachable from edges (or nodes, respectively).

We omit named shapes, simplifying shape semantics by avoiding recursive shapes (non-recursive named shapes are merely syntactic sugar), 
as well as some other features such as qualified number restrictions or complex path expressions.
A \simpleprogs shape is either a node shape $\nodeshapeanon$ or an edge shape $\edgeshapeanon$,
where $q_N$ and $q_E$ are the \emph{target queries}, 
and $\phi_N$ and $\phi_E$ are the respective \emph{constraints}
(see \Cref{trfm:def:progs:target} and \Cref{trfm:def:progs:constraint}).
Each node (or edge, respectively) of the graph that matches the target must satisfy the constraint.

\begin{definition}[Node and Edge Targets]
  \label{trfm:def:progs:target}
  Given target node or edge IDs $n$ or $e$, node or edge labels $l_N$ or $l_E$, 
	and node or edge properties $k_N$ or $k_E$,
  node or edge targets are defined as
	\begin{align*}
		q_N &\Coloneqq n \mid l_N \mid k_N\\
		q_E &\Coloneqq e \mid l_E \mid k_E
	\end{align*}
\end{definition}

\begin{definition}[Node and Edge Constraints]
  \label{trfm:def:progs:constraint}
  Node and edge constraints are defined by
  \begin{align*}
    \phi_N &\Coloneqq
      \top
      \mid n
      \mid l_N
      \mid \neg\phi_N
      \mid \phi_N \wedge \phi_N
      \mid \exists k_N.(= v)
      \mid \exists k_N.\top
      \mid \existsright{\phi_E}
      \mid \existsleft{\phi_E}
      \\
    \phi_E & \Coloneqq
      \top
      \mid e
      \mid l_E
      \mid \neg\phi_E
      \mid \phi_E \wedge \phi_E
      \mid \exists k_E.(= v)
      \mid \exists k_E.\top
      \mid \ \Rightarrow \phi_N
      \mid \ \Leftarrow \phi_N
  \end{align*}
  where in particular
      $\exists k_N.(= v)$ and $\exists k_N.\top$ require properties with certain ($v$) or any ($\top$) values,
      $\existsright{\phi_E}$ requires an outgoing edge conforming to $\phi_E$, and
      $\Rightarrow \phi_N$ requires the target node to conform to $\phi_N$.
      The remaining cases are similar.
\end{definition}

\Cref{trfm:def:sprogssemantics} introduces the validation semantics for a PG under a set of \simpleprogs shapes.
\Cref{trfm:fig:nodetargetqueries} defines the semantics of target node and target edge queries.
\Cref{trfm:fig:nconstrainteval} defines the evaluation semantics of node and edge constraints.
These definitions differ from~\cite{DBLP:conf/semweb/SeiferLS21}, using direct evaluation of constraints on target nodes instead of a semantics of faithful assignments, which is not needed since we omit recursion.

\begin{figure}[t]
  \centering
  \begin{align*}
    \progsevalq{n} &= \{n\} & \progsevalq{e} &= \{e\} \\
    \progsevalq{l_N} &= \{ n \mid n \in N \land l_N \in \lambda(n) \} &
    \progsevalq{l_E} &= \{ e \mid e \in E \land l_E \in \lambda(e) \} \\
    \progsevalq{k_N} &= \{ n \mid n \in N \land \sigma(n, k_N) \neq \emptyset\} &
    \progsevalq{k_E} &= \{ e \mid e \in E \land \sigma(e, k_E) \neq \emptyset\}
  \end{align*}
  \caption{Evaluation of target node and target edge queries.}
  \label{trfm:fig:nodetargetqueries}
\end{figure}

\newcommand{\literaltrue}{\emph{true}}
\newcommand{\trueif}{\literaltrue\ \mathrm{if}\ }
\newcommand{\trueifnot}{\trueif\ \mathrm{not}\ }
\newcommand{\logicaland}{\ \mathrm{and}\ }
\newcommand{\suchthat}{\ \mathrm{such}\ \mathrm{that}\ }

\begin{figure}[t]
  \begin{align*}
    \progsevaln{\top} &= \literaltrue \\
    \progsevaln{n'} &= \trueif n' = n\\
    \progsevaln{l_N} & = \trueif\ l_N \in \lambda(n)\\
    \progsevaln{\neg \phi_N} &= \trueifnot \progsevaln{\phi_N}\\
  \progsevaln{\phi_N^1 \land \phi_N^2} &= \trueif \progsevaln{\phi_N^1}\logicaland\progsevaln{\phi_N^2}\\
    \progsevaln{\exists k.(= v)} &\ = \trueif v \in \sigma(n,k)\\
    \progsevaln{\exists k.\top} &\ = \trueif \sigma(n,k) \ne \emptyset \\
    \progsevaln{\existsright{\phi_E}} &\ = \trueif \exists e \in E, n'\in N \suchthat \rho(e) = (n, n') \logicaland \progsevale{\phi_E}\\
    \progsevaln{\existsleft{\phi_E}} & \ = \trueif \exists e \in E, n'\in N \suchthat \rho(e) = (n', n) \logicaland \progsevale{\phi_E}\\
    \progsevale{\Rightarrow \phi_N} & \ = \progsevaln[n_2]{\phi_N} \textrm{ where } (n_1,n_2) = \rho(e)\\
    \progsevale{\Leftarrow \phi_N} & \ = \progsevaln[n_1]{\phi_N} \textrm{ where } (n_1,n_2) = \rho(e)
  \end{align*}
  \caption{Evaluation of node and edge constraints over a property graph $G$, omitting some cases for edges that are analogous to the node variants.}
  \label{trfm:fig:nconstrainteval}
\end{figure}

\begin{definition}[\simpleprogs Validation]
  \label{trfm:def:sprogssemantics}
  A PG $G$ is valid regarding a set of \simpleprogs shapes $S$ 
  if it is valid regarding all shapes $s \in S$;
  validity regarding a single shape $s =\,\nodeshapeanon$ (or $s =\,\edgeshapeanon$, respectively), 
  denoted $\shaclvalid{G}{s}$, is defined as follows:
  \begin{align*}
    &\shaclvalid{G}{\,\nodeshapeanon} = \forall n \in \progsevalq{q_N} : \progsevaln{\phi_N} = \literaltrue\\
    &\shaclvalid{G}{\,\edgeshapeanon} = \forall e \in \progsevalq{q_E} : \progsevale{\phi_E} = \literaltrue
  \end{align*}
\end{definition}

In a slight abuse of notation, we also write $\shaclvalid{G}{S}$ to indicate that 
$G$ is valid regarding each $s \in S$.
\Cref{trfm:ex:basicprogs} shows two example shapes.

\begin{example}
  \label{trfm:ex:basicprogs}
  Recall $\shaclexa{1}$ from \Cref{trfm:ex:basicshacl}.
  We can express an equivalent shape using \simpleprogs with 
	$\progsexa{1} =\,\nodeshapeanon[\labelname{Agent}][\existsright{(\labelname{obs}\ \wedge \ \Rightarrow \labelname{Person})}]$.
	This is a node shape that targets all nodes with the node label $\labelname{Agent}$;
	it requires for validity that these nodes have an outgoing edge ($\existsright{\phi_E}$)
	where the constraint $\phi_E$\,---\,an edge constraint\,---\,must validate at least one edge.
	For this edge, we require the edge label $\labelname{obs}$, and then, expressed with $\Rightarrow \phi_N$,
	that the corresponding destination node conforms to $\phi_N$, \ie, has the node label $\labelname{Person}$.

  We can also express shapes that target edges, such as $\progsexa{2} =\,\edgeshapeanon[\labelname{obs}][\ \Leftarrow \labelname{Agent}\ \wedge \ \Rightarrow \labelname{Person}]$, which states that edges labeled with $\labelname{obs}$ must have an origin node labeled $\labelname{Agent}$ ($\Leftarrow \labelname{Agent}$) 
	and the destination node labeled $\labelname{Person}$ ($\Rightarrow \labelname{Person}$).
  The example graph $\pgexa{1}$ introduced in \Cref{trfm:ex:basicpg} is valid regarding both 
	$\progsexa{1}$ and $\progsexa{2}$.
\end{example}

\section{Formalizing the Problem}
\label{trfm:sec:formal}

The essential problem statement of this work is formalized in the signature of \Cref{trfm:algorithm}:
We aim to decide, given a \simpleprogs shape $s$, whether all possible output graphs of a \gcname query $q$
are valid regarding this shape.
The set of \emph{possible output graphs} is defined as the set of graphs produced by $q$ 
from any graph that is valid regarding a given set of \simpleprogs shapes $\shapesin$, which
we call the \emph{input shapes} of our problem.
With this decision procedure, we can find the set of all shapes characterizing the possible
output graphs by enumerating a finite set of candidates and testing each candidate.

We refer to our prior work~\cite{DBLP:conf/www/Seifer0LS24} for a detailed discussion on when 
such a finite set exists or how to deal with cases where it does not exist.
In short, for restricted classes of shapes, the set of relevant candidates that need to be considered 
can be shown to be finite (\eg, $\DLogics$-based shapes); for other cases, heuristics can be employed.
Note how this allows our approach to work for compositions of queries: the output
of a prior step becomes the input of the next\footnote{
Indeed, since, as we will later explore in detail, the input shapes are encoded as $\DLogics$ axioms, 
compositions of multiple queries do \emph{not} require full exploration of the explicit set of output shapes
but can rely on the inferred axioms directly.}.

\begin{algorithm}[ht]
  \caption{High-level pseudocode outlining our core approach.}
  \label{trfm:algorithm}
  \begin{algorithmic}[1]
      \REQUIRE A finite set of shapes \simpleprogs $\shapesin$, a \gcname $q = \gcformal$, and a \simpleprogs shape $s$.
      \ENSURE Does $\shaclvalid{\sparqleval{q}{\graphin}}{s}$ hold for every graph $\graphin$ where $\shaclvalid{\graphin}{\shapesin}$? 
      \STATE $s' \gets \sprogstoshacl{s}$
      \STATE $\shapesin' \gets \{ \sprogstoshacl{s_i} \mid s_i \in \shapesin \}$
      \STATE $q' \gets \sgcoretoeccq{q}$
      \STATE $\Sigma \gets \operatorname{infer}(\shapesin', q')$
      \RETURN \textbf{if} $\operatorname{test}(\Sigma, s')$ \textbf{then} $\textsc{yes}$ \textbf{else} $\textsc{unknown}$
  \end{algorithmic}
\end{algorithm}

As motivated in the introduction (\Cref{trfm:fig:overview}), 
we solve this problem by first mapping all components (that is, graphs, shapes, and queries)
to an intermediate RDF layer and then utilizing an extended version of our algorithm presented 
in~\cite{DBLP:conf/www/Seifer0LS24}.
In \Cref{trfm:algorithm}, we outline this approach in pseudocode: 
here, the functions $\sprogstoshacl{s}$ and $\sgcoretoeccq{q}$ are the mapping functions for encoding
\simpleprogs shapes and \gcname queries in \shaclname and \eccqname, respectively.
A final function for mapping the graphs themselves is not actually utilized\,---\,since our approach
does not depend on concrete graph instances\,---\,but is required for the sake of formalizing the remaining
mappings.

The functions \emph{infer} and \emph{test} refer to the extension of~\cite{DBLP:conf/www/Seifer0LS24}:
the function \emph{infer} takes the RDF-mapped query and shapes, producing a set of DL axioms $\Sigma$;
this step extends the algorithm from~\cite{DBLP:conf/www/Seifer0LS24}, accounting for the additional features included 
in \eccqname (\cf \Cref{trfm:def:eccq:syntax}).
The function \emph{test} takes these axioms and checks whether it can be concluded that $s'$ is guaranteed 
to hold for all possible output graphs; this step utilizes entailment.

Based on this formal problem signature, we define \Cref{trfm:maintheorem}. 
We will construct the missing components of \Cref{trfm:algorithm} in the following two sections
(the family of mappings in \Cref{trfm:sec:map} and the extended inference approach in \Cref{trfm:sec:extend}).
Finally, in \Cref{trfm:sec:meta}, we will show that the theorem holds.

\begin{theorem}[Soundness of \Cref{trfm:algorithm}]
   \label{trfm:maintheorem}
   Given a finite set of shapes $\shapesin$, a \gcname $q = \gcformal$, and a shape $s$,
   then $\shaclvalid{\sparqleval{q}{\graphin}}{s}$ holds for every graph $\graphin$ 
   where $\shaclvalid{\graphin}{\shapesin}$ if
   \[\operatorname{test}(\operatorname{infer}(\{\sprogstoshacl{s_i} \mid s_i \in \shapesin\}, \sgcoretoeccq{q}),\ \sprogstoshacl{s}).\]
\end{theorem}

\section{A Family of Mappings on Reified Property Graphs}
\label{trfm:sec:map}

As outlined in \Cref{trfm:fig:overview} and the previous section, we introduce a family of mappings: 
we map PGs $G$ to RDF graphs $\pgtordf{G}$,
utilizing reification for first-class edges (\Cref{trfm:sub:pgtordf}); 
this mapping is not directly utilized by our approach
but instead comprises the formal basis on which the remaining mappings are defined.
We map \simpleprogs shapes $s$ to $\DLogics$-encoded SHACL shapes $\sprogstoshacl{s}$,
accounting for the reification introduced with the graph mapping (\Cref{trfm:sub:progstoshacl}).
And finally, we map \gcname $q$ to \eccqname $\sgcoretoeccq{q}$,
accounting for both reification and various
semantic differences between the languages (\Cref{trfm:sub:gcoretosparql}).

\subsection{PG to RDF Mapping}
\label{trfm:sub:pgtordf}

We first assume a utility function $\toiriname$ that maps the individual elements of PGs to unique IRIs.
The details of this function are left abstract; however, in real-world implementations, such a function
could, for example, utilize prefixes to differentiate between PG constructs, including
node and edge IDs, labels, properties, and values. 
We also introduce $\mnte$ (\emph{node-to-edge}) 
and $\metn$ (\emph{edge-to-node}) as role names,
and $\mnode$ and $\medge$ as special concept names distinct from $\ConceptNames$.
The former are used for reification of first-class edges,
the latter for differentiating nodes and edges explicitly%
\footnote{For the sake of readability, we assign concrete names, using \texttt{m} to hint at some unique \emph{meta} prefix
that could be used in practice; they can be arbitrary but unique names.}.

The straightforward mapping $\pgtordf{G}$ is specified in \Cref{trfm:def:graphmap} and demonstrated 
in \Cref{trfm:ex:maprdf}.
It retains all information and can be reversed by simply inverting the presented mapping function.
It utilizes $\toiriname$ to map all nodes and edges, including their associated labels and properties,
to RDF triples; IRIs for both PG nodes and PG edges are derived from their IDs.
Edges are reified using two triples and the special role names $\mnte$ and $\metn$.

\begin{definition}[Mapping from PG to RDF Graph]
  \label{trfm:def:graphmap}
  The corresponding RDF graph for a PG $G = (N,E,\rho,\lambda,\sigma)$ is the graph constructed by
  \[
    \pgtordf{G} \coloneq \{\pgtordfn{n} \mid n \in N\} \cup \{\pgtordfe{e} \mid e \in E\}
  \]
  where $\pgtordfname$ for nodes is defined as
  \begin{align*}
    \pgtordfn{n} \coloneq
        \ & \{\triple{\toiri{n}}{\rdftype}{\toiri{l_N}} \mid l_N \in \lambda(n)\} \\
          & \cup \{\triple{\toiri{n}}{\toiri{k_{N}}}{\toiri{v}} \mid k_{N} \in K_{N}, v \in \sigma(n, k_{N})\}\\
        & \cup \{  \triple{\toiri{n}}{\rdftype}{\mnode}  \}
  \end{align*}
  and for edges (with $\rho(e) = (n_1, n_2)$) as
  \begin{align*}
    \pgtordfe{e} \coloneq
        \ & \{\triple{\toiri{n_1}}{\mnte}{\toiri{e}}, \triple{\toiri{e}}{\metn}{\toiri{n_2}}\} \\
        & \cup \{ \triple{\toiri{e}}{\rdftype}{\toiri{l_E}} \mid l_E \in \lambda(e)  \}\\
        & \cup \{ \triple{\toiri{e}}{\toiri{k_{E}}}{\toiri{v}} \mid k_{E} \in K_{E}, v \in \sigma(e, k_{E}) \}\\
        & \cup \{ \triple{\toiri{n}}{\rdftype}{\medge}  \}.
  \end{align*}
\end{definition}

\begin{example}
  \label{trfm:ex:maprdf}
  We map the graph $\pgexa{1}$ from \Cref{trfm:ex:basicpg} to RDF,
  assuming a reasonable definition for $\toiriname$ and omitting the display of any prefixes,
  to obtain the following set of triples for $\pgtordf{\pgexa{1}}$,
  as visualized side-by-side in \Cref{trfm:fig:sidebyside}:
  \begin{align*}
    \{ & \triplei{\idname{1}}{\rdftypea}{$\mnode$}, \triplei{\idname{1}}{\rdftypea}{Agent}, \triplei{\idname{1}}{name}{Smith}, \triplei{\idname{2}}{\rdftypea}{$\mnode$}, \triplei{\idname{2}}{\rdftypea}{Person},\\
       & \triplei{\idname{10}}{\rdftypea}{$\medge$}, \triplei{\idname{10}}{\rdftypea}{obs}, \triplei{\idname{10}}{since}{2002}, \triplei{\idname{1}}{$\mnte$}{\idname{10}}, \triplei{\idname{10}}{$\metn$}{\idname{2}}\}.
  \end{align*}
\end{example}

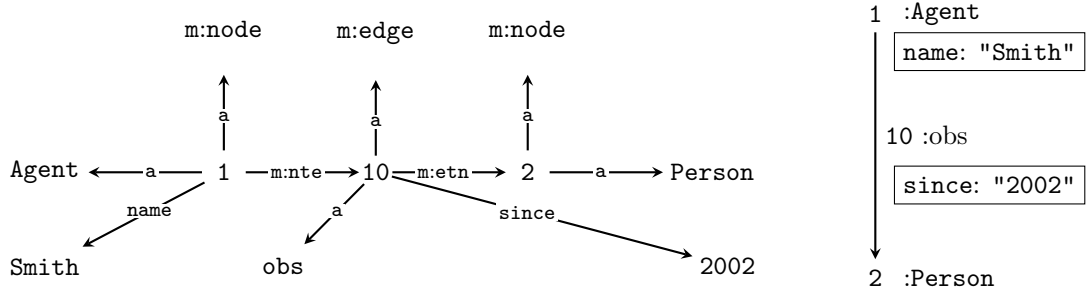
\begin{figure}[t]
\begin{minipage}{0.73\textwidth}
  \begin{tikzpicture}[
    every label/.style={inner sep=0.0em}
]
    \node[new node, inner sep=0.1cm] (s) {\idname{1}};
    \node[new node] (a) [left = 1.5cm of s] {\concept{Agent}};
    \node[new node, inner sep=0.0cm] (o) [right = 1.5cm of s]  {\idname{10}};
    \node[new node, inner sep=0.1cm] (p) [right = 1.5cm of o]  {\idname{2}};
    \node[new node] (pp) [right = 1.5cm of p] {\concept{Person}};
    \node[new node] (n) [below = 0.1cm of a]  {\concept{Smith}};
    \node[new node, xshift=0.2cm] (obs) [below = 0.1cm of pp]  {\concept{2002}};
    \node[new node] (b) [left = 5.0cm of obs] {\prop{obs}};
    \node[new node] (sn) [above = 1.0cm of s]  {$\mnode$};
    \node[new node] (pn) [above = 1.0cm of p]  {$\mnode$};
    \node[new node] (oe) [above = 1.0cm of o]  {$\medge$};

    \draw [-] (s) edge[arrout] node[pglabel]{$\mnte$} (o);
    \draw [-] (s) edge[arrout] node[pglabel]{$\rdftypea$} (sn);
    \draw [-] (p) edge[arrout] node[pglabel]{$\rdftypea$} (pn);
    \draw [-] (o) edge[arrout] node[pglabel]{$\rdftypea$} (oe);
    \draw [-] (o) edge[arrout] node[pglabel]{$\rdftypea$} (b);
    \draw [-] (o) edge[arrout] node[pglabel]{$\metn$} (p);
    \draw [-] (s) edge[arrout] node[pglabel]{$\rdftypea$} (a);
    \draw [-] (p) edge[arrout] node[pglabel]{$\rdftypea$} (pp);
    \draw [-] (s) edge[arrout] node[pglabel]{\prop{name}} (n);
    \draw [-] (o) edge[arrout] node[pglabel]{\prop{since}} (obs);

\end{tikzpicture} 
\end{minipage}
\hfill
\begin{minipage}{0.23\textwidth}
  \begin{tikzpicture}[
    every label/.style={inner sep=0.0em}
    ]

    \node[new node, label={0:\ :\concept{Agent}}] (s) {\idname{1}};
    \node[new node, label={0:\ :\concept{Person}}] (p) [below = 3.0cm of s]  {\idname{2}};

    \node[draw,minimum height = 0.5cm, minimum width = 1cm] (pp1) [below right = 0.1cm of s] {\propname{name}: \stringvalue{Smith}};
    \node[draw,minimum height = 0.5cm, minimum width = 1cm] (pp2) [below = 1.25cm of pp1] {\propname{since}: \stringvalue{2002}};

    \path[->]
        (s) edge [arrout] node [right] {\idname{10}\ :obs} (p);
\end{tikzpicture}
\end{minipage}
\caption{The reified RDF encoding $\pgtordf{\pgexa{1}}$ (left) and the original graph $\pgexa{1}$ (right). 
  Note how the edge with ID $10$ (right) becomes a normal RDF node via reification (left).}
\label{trfm:fig:sidebyside}
\end{figure}

\subsection[\simpleprogs to SHACL (ALCHOI) Mapping]{\simpleprogs to SHACL ($\DLogics$) Mapping}
\label{trfm:sub:progstoshacl}

We define a function $\sprogstoshaclname$ for mapping \simpleprogs shapes to SHACL shapes.
Since we follow our prior work~\cite{DBLP:conf/www/Seifer0LS24} in representing SHACL shapes as 
$\DLogics$ axioms, we map to $\DLogics$ directly.
That is, we map both target queries and constraints to concept descriptions
and shapes to subsumption axioms (\Cref{trfm:def:shapemap}).

\begin{definition}[Mapping sProGS to $\DLogics$]
  \label{trfm:def:shapemap}
  A node shape $\nodeshapeanon$ is mapped to the axiom $\sprogstoshacl{\nodeshapeanon} = \sprogstoshaclnt{q_N} \sqsubseteq \sprogstoshaclnc{\phi_N}$, 
  where $\sprogstoshaclnt{q_N}$ and $\sprogstoshaclnc{\phi_N}$ are defined in \Cref{trfm:fig:maptargetn} and \Cref{trfm:fig:mapconstraintn}, respectively.
  With equivalent definitions (\Cref{trfm:fig:maptargete} and \Cref{trfm:fig:mapconstraintn}), we can map edge shapes as 
  $\sprogstoshacl{\edgeshapeanon} = \sprogstoshaclet{q_E} \sqsubseteq \sprogstoshaclec{\phi_E}$.
\end{definition}

We write $\sprogstoshacl{S}$ as a shorthand for mapping each $s \in S$.
Most constructs from the syntax of \simpleprogs translate directly to $\DLogics$, 
similarly to how SHACL is defined in terms of $\DLogics$ in~\cite{DBLP:conf/www/Seifer0LS24}.
Interesting cases include the constraint $\Rightarrow \phi_N$, which maps to 
$\exists\ \metn\ .\ \sprogstoshaclnc{\phi_N}$.
This node constraint requires the special role name $\metn$ 
according to the reification defined for PGs,
since the original \simpleprogs constraint refers to the target node of an edge.
Given that this edge is an ordinary node in the RDF encoding, we must invoke $\exists \metn$
to navigate to the node representing that target and then use the mapped constraint $\phi_N$.
Similarly, $\existsright{\phi_E}$ is mapped to $\exists\ \mnte\ .\ \sprogstoshaclec{\phi_E}$.
\Cref{trfm:ex:mapshacl} shows the full mappings for the running example.

\begin{figure}[t]
    \centering
    \begin{minipage}[t]{0.25\textwidth}
        \centering
        \vspace{-3cm}
        \begin{tabular}{ l l }
            \simpleprogs& $\DLogics$\\
            \hline
            $\sprogstoshaclnt{n}$ & $\{\toiri{n}\}$\\
            $\sprogstoshaclnt{l_N}$ & $\toiri{l_N}$\\
            $\sprogstoshaclnt{k_N}$ & $\exists\ \toiri{k_N}\ .\ \top$
        \end{tabular}
        \caption{Mapping for node target queries.}
        \label{trfm:fig:maptargetn}
        \vfill
        \vspace{1cm} 
        \begin{tabular}{ l l }
            \simpleprogs & $\DLogics$\\
            \hline
            $\sprogstoshaclet{e}$ & $\{\toiri{e}\}$\\
            $\sprogstoshaclet{l_E}$ & $\toiri{l_E}$\\
            $\sprogstoshaclet{k_E}$ & $\exists\ \toiri{k_E}\ .\ \top$
        \end{tabular}
        \caption{Mapping for edge target queries.}
        \label{trfm:fig:maptargete}
    \end{minipage}
    \hfill
    \begin{minipage}[t]{0.6\textwidth}
        \centering
        \begin{tabular}{ l l }
            \simpleprogs & $\DLogics$\\
            \hline
            $\sprogstoshaclnc{\top}$                         & $\top$\\
            $\sprogstoshaclnc{n}$                            & $\toiri{n}$\\
            $\sprogstoshaclnc{l_N}$                          & $\toiri{l_N}$\\
            $\sprogstoshaclnc{\neg \phi_N}$                  & $\neg \sprogstoshaclnc{\phi_N}$\\
            $\sprogstoshaclnc{\phi_N^1 \wedge \phi_N^2}$     & $\sprogstoshaclnc{\phi_N^1} \sqcap \sprogstoshaclnc{\phi_N^2}$ \\
            $\sprogstoshaclnc{\exists k_N.(= v)}$              & $\exists\ \toiri{k_N}\ .\ v $ \\
            $\sprogstoshaclnc{\exists k_N.\top}$              & $\exists\ \toiri{k_N}\ .\ \top $ \\
            $\sprogstoshaclnc{\exists^{\leftarrow} \phi_E}$     & $\exists\ \metn^-\ .\ \sprogstoshaclec{\phi_E}$ \\
            $\sprogstoshaclnc{\exists^{\rightarrow} \phi_E}$ & $\exists\ \mnte\ .\ \sprogstoshaclec{\phi_E}$ \\
            \ldots & \ldots \\
            $\sprogstoshaclec{\Rightarrow \phi_N}$ & $\exists\ \metn\ .\ \sprogstoshaclnc{\phi_N}$ \\
            $\sprogstoshaclec{\Leftarrow \phi_N}$  & $\exists\ \mnte^-\ .\ \sprogstoshaclnc{\phi_N}$ \\
        \end{tabular}
        \caption{Mapping \simpleprogs node and edge constraints to $\DLogics$. 
          We omit some cases for edge constraints that are equivalent to node constraint cases.}
        \label{trfm:fig:mapconstraintn}
    \end{minipage}
\end{figure}

\begin{example}
  \label{trfm:ex:mapshacl}
  We map the shapes $\progsexa{1}$ and $\progsexa{2}$ from \Cref{trfm:ex:basicprogs} to 
  equivalent $\DLogics$ subsumption axioms.
  For the node shape $\progsexa{1}$, this results in the following shape:
  \begin{align*}
    \sprogstoshacl{\progsexa{1}} &= \sprogstoshacl{\nodeshapeanon[\labelname{Agent}][\existsright{(\labelname{obs}\ \wedge \ \Rightarrow \labelname{Person})}]} \\
    &= \labelname{Agent} \sqsubseteq \exists \mnte.(\labelname{obs} \sqcap \exists \metn.\labelname{Person})
  \end{align*}
  Note how the reification is expressed as two nested existential quantifications: we first
  navigate to the reified edge (\emph{node}) via $\mnte$, where not only a label ($\labelname{obs}$)
  but also a particular node constraint holds, to which we navigate via $\metn$; 
  we ensure $\idname{1}\labelname{:Agent} \rightarrow_{\mnte} \idname{10}\labelname{:obs} \rightarrow_{\metn} \idname{2}\labelname{:Person}$
  (\cf \Cref{trfm:fig:sidebyside}).
  Similarly, for the edge shape $\progsexa{2}$ we want $\idname{1}\labelname{:Agent} \leftarrow_{\mnte} \idname{10}\labelname{:obs} \rightarrow_{\metn} \idname{2}\labelname{:Person}$, but from the perspective of all reified edges, instead of nodes, such that in
  \begin{align*}
    \sprogstoshacl{\progsexa{2}} &= \sprogstoshacl{\edgeshapeanon[\labelname{obs}][\ \Leftarrow \labelname{Agent}\ \wedge \ \Rightarrow \labelname{Person}]} \\
        &= \labelname{obs} \sqsubseteq \exists \mnte^- . \labelname{Agent} \sqcap \exists \metn . \labelname{Person}
  \end{align*}
  we target an edge to navigate back to the origin (inverse $\mnte$) as well as
  the destination via $\metn$, requiring the node labels $\labelname{Agent}$ and $\labelname{Person}$, respectively.
\end{example}

\subsection{\gcname to \eccqname Mapping}
\label{trfm:sub:gcoretosparql}

Finally, we map \gcname to \eccqname.
The principal challenge for this mapping\,---\,aside from accounting for reification\,---\, is 
that in \gcname the semantics for constructing labels and properties are different from \eccqname 
in that all properties and labels are kept for constructed nodes or edges
unless there are remove clauses.
In \eccqname, on the other hand, all triples in the result graph must be explicitly constructed
by atomic patterns in the query template.

We solve this issue by utilizing generic concept and role variables with appropriate filter conditions 
to define a mapping $\sgcoretoeccq{q}$ for a \gcname $q$ (\Cref{trfm:def:querymap}).
To this end, we introduce an intermediate representation referred to as IQL (\emph{Intermediate Query Language}) 
in \Cref{trfm:def:iql} that aligns with the two types of patterns of \gcname,
\ie matching nodes or edges, on the \gcname side
with full query patterns or templates on the \eccqname side.
This aligns parts of queries on both sides that compose naturally.
Note how the decision to syntactically include where clauses in node or edge patterns 
of \gcname aids this composability. 

Via this intermediate layer, we associate \gcname variables with generic copy patterns 
expressed in \eccqname that explicitly construct the required labels and properties.
\emph{Set} clauses translate to explicit construction in \eccqname, 
while \emph{remove} clauses translate to filter conditions on generic variables in the pattern
of \eccqname queries, removing the offending labels from the bindings for these generic variables.
We demonstrate this mapping in \Cref{trfm:ex:querymap}.

\begin{definition}[Mapping $\gcname$ to $\eccqname$]
  \label{trfm:def:querymap}
  Given a \gcname $q$, we define the mapping function $\sgcoretoeccqname(q)$ by introducing 
  the intermediate query language IQL (\Cref{trfm:def:iql}) 
  and several utility functions for mapping to and from IQL, as 
  \begin{align*}
    \sgcoretoeccq{q} := \iqltoeccq{\sgcoretoiql{q}}
  \end{align*}
  where the following utility functions (defined in \Cref{trfm:fig:gciqlboth}) are utilized:
  \begin{align*}
    \sgcoretoiql{\gcformalname} &= \bigcup_{\gctemplatecomponent \in \gctemplate} \sgcoretoiqlc{\gctemplatecomponent, \operatorname{vars}(\gcpattern)} \cup \bigcup_{\gcpatterncomponent \in \gcpattern} \sgcoretoiqlm{\gcpatterncomponent}\\
    \iqltoeccq{\iqlformal} &= \bigcup_{t \in \iqltemplate} \iqltoeccqc{t} \cup \bigcup_{p \in \iqlpattern} \iqltoeccqm{p}
  \end{align*}
\end{definition}

\begin{definition}[Intermediate Query Language]
  \label{trfm:def:iql}
  We define the syntax of the intermediate query language (IQL) as follows:
  \begin{align*}
    q &:= \iqlformal\\
    p &:= \operatorname{M_N}(x_n, W_L, W_K, W_V, R_L, R_K) 
          \mid \operatorname{M_E}(x_n, x_n, x_e, W_L, W_K, W_V, R_L, R_K) \\
    t &:= \operatorname{C_N}(x_n, A_L, A_K, V) 
          \mid \operatorname{C_E}(x_n, x_n, x_e, A_L, A_K, V)
  \end{align*}
  where $W_L$ is a set of label patterns $\whereLabelRelative{l}$,
  $W_K$ is a set of key patterns $\wherePropertyRelative{k}$,
  and $W_V$ is a set of patterns $\wherePropertyEqualsRelative{k}{v}$.
  Similarly, $A_L$ is a set of \emph{set label} patterns $\gcset l$,
  $A_K$ is a set of \emph{set key-value} patterns $\gcset k = v$,
  and $V$ is a set of variables.
  Finally, $R_L$ is a set of \emph{remove label} patterns $\gcremove l$, and
  $R_K$ is a set of \emph{remove key} patterns $\gcremove k$.
  The semantics of IQL are implicitly defined through the mappings
  in \Cref{trfm:fig:gciqlboth}.
\end{definition}

\begin{figure*}
    \centering
    \input{figures/gciqleccq}
\caption{
    Component-wise mappings from \gcname to \iqlname and \iqlname to \eccqname, both for 
    query templates and patterns.
    We write $\gcsr_x$ to refer to the union of all sets of add or remove patterns $\gcsr$, 
    such that $\gcoreNodeS{x}{\gcsr} \in \gctemplate$ or $\gcoreEdgeS{x_1}{x_2}{x}{\gcsr} \in \gctemplate$.}
\label{trfm:fig:gciqlboth}
\end{figure*}

\begin{example}
  \label{trfm:ex:querymap}
  The example \gcname $\gcexa{2}$ in \Cref{trfm:fig:ex:querymapping} (left) includes the major \gcname features,
  such as matching on an edge pattern, setting and removing labels, and restricting results with constraints.
  The resulting \eccqname $\sgcoretoeccq{\gcexa{2}}$ is shown in \Cref{trfm:fig:ex:querymapping} (right).
  The query is suitable for application to $\pgexa{1}$ from \Cref{trfm:ex:basicpg}.
  We focus here on the single variable $\varname{y}$, which occurs in two atomic patterns:
  $\gcoreNode{(\varname{y}, \whereLabelRelative{\labelname{Person}})}$ (in $\gcpattern$) and 
  $\gcoreNodeS{\varname{y}}{\{\gcset \labelname{POI}, \gcremove \labelname{Person}\}}$ (in $\gctemplate$).

$\operatorname{M_N}(\varname{y}, \{:\labelname{Person}\}, \emptyset, \emptyset, \{\gcremove\labelname{Person}\}, \emptyset\})$
and $\operatorname{C_N}(\varname{y}, \{\gcset\labelname{POI}\}, \emptyset, \{\varname{x},\varname{\varname{y}},\varname{e}\})$ are the 
  corresponding IQL terms.
  The former maps to the \eccqname pattern including the generic pattern $\atomGenc{\varname{y}}$, 
  copying all labels for $\varname{y}$, 
  the explicit pattern $\atomC{\varname{y}}{\labelname{Person}}$ selecting instances of $\labelname{Person}$,
  but also the filter $\atomCNot{\varname{y}}{\labelname{Person}}$, which prevents bindings,
  and thus also the construction, for label $\labelname{Person}$ for the generic variable $\cvar{\varname{y}}$,
  included because of the remove term in the original \gcname query.
  The \eccqname template includes, along with the generic patterns, the explicit construction of
  $\atomC{\varname{y}}{\labelname{POI}}$.
\end{example}

\begin{figure*}[t!]
    \begin{subfigure}[t]{0.18\textwidth}
        \centering
        \begin{align*}
          \{&\gcoreNodeS{\varname{x}}{\emptyset},\\
            &\gcoreEdgeS{\varname{x}}{\varname{y}}{\varname{e}}{\emptyset},\\
            &\gcoreNodeS{\varname{y}}{\{\tikzmarknode{c}{\gcset \labelname{POI}}, \tikzmarknode{a}{\gcremove \labelname{Person}\}}}%
          \}\\
            &\Leftarrow \\
           \{%
            &\gcoreNode{(\varname{x}, \{\wherePropertyEqualsRelative{\propname{name}}{\stringvalue{Smith}}\})},\\
            &\gcoreEdge{\varname{x}}{\varname{y}}{(\varname{e}, \{\whereLabelRelative{\labelname{obs}}\})},\\
            &\gcoreNode{(\varname{y}, \{\whereLabelRelative{\labelname{Person}}\})}%
         \}
        \end{align*}
    \end{subfigure}%
    \hfill
    \begin{subfigure}[t]{0.78\textwidth}
        \centering
        \begin{align*}
          \{&
            \atomGenN{\varname{x}},
            \atomGenc{\varname{x}},
            \atomGenp{\varname{x}},\\
            &
            \atomGenE{\varname{e}},
            \atomGenc{\varname{e}},
            \atomGenp{\varname{e}},
            \atomPin{\varname{x}}{\varname{e}},
            \atomPout{\varname{e}}{\varname{y}},\\
            &
            \atomGenN{\varname{y}},
            \atomGenc{\varname{y}},
            \atomGenp{\varname{y}},
            \tikzmarknode{d}{\atomC{\varname{y}}{\labelname{POI}}}
              \}\\
            & \leftarrow \\
          (\{&
            \atomGenN{\varname{x}},
            \atomGenc{\varname{x}},
            \atomGenp{\varname{x}},
            \atomP{\varname{x}}{\labelname{Smith}}{\labelname{name}},\\
            &
            \atomGenE{\varname{e}},
            \atomGenc{\varname{e}},
            \atomGenp{\varname{e}},
            \atomC{\varname{e}}{\labelname{obs}},
            \atomPin{x\varname{}}{\varname{e}},
            \atomPout{\varname{e}}{\varname{y}},\\
            &
            \atomGenN{\varname{y}},
            \atomGenc{\varname{y}},
            \atomGenp{\varname{y}},
            \atomC{\varname{y}}{\labelname{Person}}
              \},\\
          \{ &
            \tikzmarknode{b}{\atomCNot{\varname{y}}{\labelname{Person}}},
            \atomPNot{\varname{x}}{\mnte},
            \atomPNot{\varname{e}}{\metn},
              \atomPNot{\varname{y}}{\mnte}
              \})
        \end{align*}
    \end{subfigure}
  \caption{\gcname $\gcexa{2}$ on the left, and \eccqname $\sgcoretoeccq{\gcexa{2}}$ on the right.
    Note how we visually align the lines of the left and right columns for their common variables, 
    hinting at the IQL rules, \eg, the $\labelname{POI}$ set label and respective triple pattern (blue).
    The only exception is the inclusion of $\atomCNot{\varname{y}}{\labelname{Person}}$ (yellow)
    in the filter patterns of the \eccqname (right column),
    highlighting why the IQL pattern $M_N$ needs the remove clauses as part of their arguments:
    remove clauses from the template of \gcname affect the pattern (filters) of \eccqname.}
  \label{trfm:fig:ex:querymapping}
  \begin{tikzpicture}[overlay, remember picture]
    \node[fill=yellow, opacity=0.2, fit=(a), inner sep=2pt, rounded corners] {};
    \node[fill=yellow, opacity=0.2, fit=(b), inner sep=2pt, rounded corners] {};
    \node[fill=blue, opacity=0.1, fit=(c), inner sep=2pt, rounded corners] {};
    \node[fill=blue, opacity=0.1, fit=(d), inner sep=2pt, rounded corners] {};
  \end{tikzpicture}
\end{figure*}

\section{Shape Inference for \eccqname}
\label{trfm:sec:extend}

As introduced in the formalization of our core problem in \Cref{trfm:sec:formal},
we aim to infer a set of $\DLogics$ axioms $\Sigma$ from an \eccqname $q$ 
and a set of shapes $\shapesin$ (written $\Sigma = \operatorname{infer}(\shapesin, q)$), 
such that $\Sigma \models s_o$ ($\operatorname{test}(\Sigma,s_o)$ in \Cref{trfm:sec:formal}) 
implies that $s_o$ is a shape validating all possible output graphs of $q$, 
under the precondition that the input graphs of $q$ are guaranteed to be valid regarding $\shapesin$.

\subsection{Encoding Constraints in DL}

\begin{figure*}[t]
  \centering
  \begin{subfigure}[b]{.18\textwidth}
    \centering
    \begin{tikzpicture}[
    every label/.style={inner sep=0.0em}
]
    \node[new node,label={0:$\eAi,\eBi$}] (a) {$\ea$};
    \node[new node] (c) [below = 0.4cm of a] {};
    \node[new node,label={0:$\eEi,\eAi$}] (b) [below = 0.2cm of a]  {$\eb$};
    \node[new node] (d) [below = 0.4cm of b] {};
    \node[new node,label={0:$\eAi$}] (e) [below = 0.2cm of b]  {$c$};
\end{tikzpicture}
    \caption{An example input graph.}
    \label{trfm:fig:input}
  \end{subfigure}%
  \hfill%
  \begin{subfigure}[b]{.13\textwidth}
    \centering
    \begin{tikzpicture}[
    every label/.style={inner sep=0.0em}
]
    \node[new node,label={0:$\exv$}] (a) {$\ea$};
    \node[new node] (c) [below = 0.4cm of a] {};
    \node[new node,label={0:$\eyv$}] (b) [below = 0.2cm of a]  {$\eb$};
\end{tikzpicture}
    \caption{Variable concepts.}
    \label{trfm:fig:variables}
  \end{subfigure}%
  \hfill%
  \begin{subfigure}[b]{.14\textwidth}
    \centering
    \begin{tikzpicture}[
    every label/.style={inner sep=0.0em}
]
    \node[new node,label={0:$\eAo,\eBo$}] (a) {$\ea$};
    \node[new node] (c) [below = 0.4cm of a] {};
    \node[new node,label={0:$\eEo$}] (b) [below = 0.2cm of a]  {$\eb$};
\end{tikzpicture}
    \caption{Output graph of $q$.}
    \label{trfm:fig:out}
  \end{subfigure}
  \hfill%
  \begin{subfigure}[b]{.14\textwidth}
    \centering
    \begin{tikzpicture}[
    every label/.style={inner sep=0.0em}
]
    \node[new node,label={0:$\eAo,\eBo$}] (a) {$\ea$};
    \node[new node] (c) [below = 0.4cm of a] {};
    \node[new node,label={0:$\eEo,\eAo$}] (b) [below = 0.2cm of a]  {$\eb$};
\end{tikzpicture}
    \caption{Output graph of $q'$.}
    \label{trfm:fig:outprime}
  \end{subfigure}
  \caption{Example for graphs involved in constructing the validation knowledge base (\cf~\cite{DBLP:conf/www/Seifer0LS24}) 
    for the query $q$ with $\eccqpattern = \eccqtemplate = \{\atomC{x}{\eA}, \atomC{x}{\eB}, \atomC{y}{\eE}\}$ 
    from \Cref{trfm:ex:broken}. This case demonstrates the potential issue arising when adding 
    the atomic pattern $\atomGenc{y}$ in $q'$, since variable $y$ matches 
    $\eb$ which is an instance of $\eA$; this violates the assumptions about the 
    relationship of $\ddot{\eA}$ and $\ddot{\eB}$ in the output graph.
    Colors visualize different namespaces, while floating labels indicate concept names (\eg, $a$ is both $\eA$ and $\eB$).}
  \label{trfm:fig:matchingexample}
\end{figure*}

We build on and extend \cite{DBLP:conf/www/Seifer0LS24},
which essentially approximates axioms of the 
closed-world assumption from sub-expressions of the query 
and joins them with the input shapes and a few additional axioms.
To this end, three name\-spaces need to be distinguished:
those of the input graphs, 
of the subgraphs matched by the query, 
and of the new output graphs constructed as a result.
The extensions of concepts involved clearly differ between these three name\-spaces; 
thus, we follow~\cite{DBLP:conf/www/Seifer0LS24} in renaming any concept name $A$ in the input name\-spaces as
$\dot{A}$ and $\ddot{A}$ in the matched and output name\-spaces.
We do the same for each role name $p$ (with $\dot{p}$ and $\ddot{p}$).

Query variables can be represented with \emph{variable concepts},
that is, fresh concept names $V_{x}$ for each variable $x$.
The extensions of these concepts are equal to the set of all of their bindings.
Thus, variable concepts remain the same in all name\-spaces
and are not subject to any renaming.
They are constrained by the basic graph patterns of the query pattern $\eccqpattern$ 
they occur in and therefore play a key role in encoding the input-output relationship of queries.

The essential extension of our SPARQL subset \eccqname over the subset considered in~\cite{DBLP:conf/www/Seifer0LS24}
are generic copy operations through variables for concepts and properties
and certain filter expressions on these variables.
The minimal \Cref{trfm:ex:broken} shows how this invalidates the axioms inferred by~\cite{DBLP:conf/www/Seifer0LS24},
while also giving the intuition of this method.

\begin{example}
  \label{trfm:ex:broken}
  Assume a query $q$ with 
  $\eccqpattern = \eccqtemplate = \{\atomC{\varname{x}}{\eA}, \atomC{\varname{x}}{\eB}, \atomC{\varname{y}}{\eE}\}$ 
  and no input shapes.
  The axioms inferred by~\cite{DBLP:conf/www/Seifer0LS24} would include
  $\{\exvb \equiv \eA \sqcap \eB, \eyvb \equiv \eE, \ddot{\eA} \equiv \exvb, \ddot{\eB} \equiv \exvb, \ddot{\eE} \equiv \eyvb\}$.
  We demonstrate these for an example input graph (\Cref{trfm:fig:input}) in \Cref{trfm:fig:matchingexample}.
  That is, we define the variable concept $\exvb$ (visualized as a graph in \Cref{trfm:fig:variables}) 
  in terms of its bindings,
  namely the intersection of $\eA$ and $\eB$, and the new concept names in the output 
  name\-space $\ddot{\eA}$ and $\ddot{\eB}$ in terms of these $\exvb$.
  Since these axioms entail, \eg, $\ddot{\eA} \sqsubseteq \ddot{\eB}$, just as in the example \Cref{trfm:fig:out}, 
  we can conclude that the shape $\eA \sqsubseteq \eB$ holds on all output graphs.

  If we were to extend query $q$ as $q'$ by adding
  $\{\atomGenc{\varname{y}}\}$ to both $\eccqpattern$ and $\eccqtemplate$,
  the axiom $\ddot{\eA} \sqsubseteq \ddot{\eB}$ \emph{should} no longer be entailed;
  indeed, variable $\varname{y}$ might now also include\,---\,such as in 
  \Cref{trfm:fig:outprime}\,---\,bindings 
  that are of type $\eA$ exclusive or $\eB$ and modify $\ddot{\eA}$ or $\ddot{\eB}$ 
  through the generic pattern $\atomGenc{\varname{y}}$, thereby invalidating this subsumption.
\end{example}

We can remedy this broken inference introduced in \Cref{trfm:ex:broken} 
by including \emph{components} of all variables that have generic copy operations 
in the definition of concepts like $\ddot{\concept{\eA}}$.
Such components are defined in \Cref{trfm:def:utilcc} and included in axioms constructed 
by the adapted method outlined in \Cref{trfm:def:cwa}.
We write $\vcg(q)$ (\emph{variable connectivity graph}) to mean a graph where variables $\var(q)$ are nodes
and there exists an edge between two variables 
if and only if they both occur in a pattern (\eg, $x$ and $y$ in $\atomP{x}{y}{p}$) in $\sgcoretoeccq{q}$.
\Cref{trfm:ex:broken} can therefore be fixed by using this extension as shown in \Cref{trfm:ex:fixed}.
We demonstrate a full example of the algorithm in \Cref{trfm:sub:rax}. 

\begin{example}
  \label{trfm:ex:fixed}
  With the extended method, we infer the set of axioms including
  $\{\exvb \equiv \eA \sqcap \eB, \eyvb \equiv E, \ddot{\eA} \equiv \exvb \sqcup \eyvb^\eA \sqcup \exvb^\eA, \ddot{\eB} \equiv \exvb \sqcup \exvb^\eB \sqcup \eyvb^\eB, \ddot{E} \equiv \eyvb \sqcup \exvb^E \sqcup \eyvb^E\}$
   where $\eyvb^\eA \equiv \eyvb \sqcap \eA$, $\eyvb^\eB \equiv \eyvb \sqcap \eB$, $\eyvb^\eE \equiv \eyvb \sqcap E$ and $\exvb^\eA \equiv \exvb^\eB \equiv \exvb^E \equiv \bot$.

  Here, the axioms defined for the output name\-space (\eg, $\ddot{\eA} \equiv \exvb \sqcup \exvb^\eA \sqcup \eyvb^\eA$) include additional component concepts
  (\eg, $\eyvb^\eA$) for all variables in the query.
  Thus, we ensure that the possibility of atomic pattern $\atomGenc{\varname{y}}$ constructing instances of $\ddot{\eA}$ in output graphs is accounted for.
  As a result, we can no longer \emph{erroneously} infer $\eA \sqsubseteq \eB$ without additional knowledge in terms of input shapes
  about the relationship of $\eA$, $\eB$, and $\eE$.
\end{example}

\begin{definition}[Concept Components]\label{trfm:def:utilcc}
  We define the \emph{concept components} $\operatorname{Cc}(q)$ of an \eccqname $q = \eccqformal$ 
  as the set of axioms including:
  \begin{enumerate}
    \item For all $x \in \var(P)$ and $A \in \mathcal{V}$, if $\atomGenc{x} \in \eccqpattern \cap \eccqtemplate$ 
      and $(\atomCNot{x}{A}) \not\in \eccqfilter$, then $\vconcept{x}^{A} \equiv \vconcept{x} \sqcap A$.
      Otherwise, $\vconcept{x}^{A} \equiv \bot$.

    \item For all $x,y \in \var(P)$ and $p \in \mathcal{V}$, if $\atomGenp{x} \in \eccqpattern \cap \eccqtemplate$,
      $(\atomPNot{x}{p}) \not\in \eccqfilter$ and $\atomP{x}{y}{p}\not\in \eccqpattern$,
      then $\vconcept{x}^{p} \equiv \vconcept{x} \sqcap \exists p . \vconcept{x}^{p,o}$.
      Otherwise, $\vconcept{x}^{p} \equiv \bot$ and $\vconcept{x}^{p,o} \equiv \bot$.
  \end{enumerate}
  Here, $\mathcal{V} = \voc(q)\cup\voc(\shapesin)$ is the \emph{vocabulary} of all concept, role, 
  and individual names that occur in the query or in input shapes.
\end{definition}

\begin{definition}[Axiom Encoding]\label{trfm:def:cwa}
  The \emph{axiom encoding} of an \eccqname $q$
  and a set of $\DLogics$ shapes $\shapesin$,
  denoted $\operatorname{infer}(\shapesin, q)$,
  is the minimal set of $\DLogics$ axioms that include $\operatorname{Cc}(q)$ (\Cref{trfm:def:utilcc}) 
  as well as the following axioms:
  \begin{enumerate}
    \item The set of input shapes $\shapesin$.
    \item For each variable $x$ in $\var(q)$, the axiom 
        \begin{align*}
          \vconcept{x} \sqsubseteq \textstyle
          & {\bigsqcap_{\atomC{x}{A} \in P}} A
            \sqcap {\bigsqcap_{\atomP{x}{u}{p} \in P}} \exists p.\Concept{u}
            \sqcap {\bigsqcap_{\atomP{u}{x}{p} \in P}} \exists p^-.\Concept{u},
        \end{align*}
        and if $\vcg(P \cup H)$ is acyclic regarding $x$, then also the axiom
        \begin{align*}
        \vconcept{x} \sqsupseteq \textstyle
          & {\bigsqcap_{\atomC{x}{A} \in P}} A
          \sqcap {\bigsqcap_{\atomP{x}{u}{p} \in P}} \exists p.\Concept{u}
          \sqcap {\bigsqcap_{\atomP{u}{x}{p} \in P}} \exists p^-.\Concept{u}.
        \end{align*}

    \item For each concept name $A$ in $\voc(\eccqtemplate)$,
        $\ddot{A} \equiv \bigsqcup_{\atomC{u}{A} \in H} \Concept{u} \sqcup \bigsqcup_{x\in\var(H)} \vconcept{x}^A$.

    \item For each $p$ in $\voc(\eccqtemplate)$
        \begin{align*}
          \textstyle\bigsqcup_{\atomP{u}{v}{p} \in  \eccqtemplate} \Concept{u} &\sqsubseteq \exists \ddot{p}.\Concept{v} ,\\
          \exists \ddot{p}.\Concept{v} &\sqsubseteq \textstyle\bigsqcup_{\atomP{u}{v}{p} \in  \eccqtemplate} \Concept{u} \sqcup \textstyle\bigsqcup_{\atomGenp{x} \in H} \vconcept{x},\\
          \exists \ddot{p}.\top &\equiv {\textstyle\bigsqcup_{\atomP{u}{v}{p} \in  \eccqtemplate}} (\Concept{u} \sqcap \exists \ddot{p}.\Concept{v}) \sqcup \textstyle\bigsqcup_{x\in\var(H)} \vconcept{x}^p,\\
          \textstyle\bigsqcup_{\atomP{u}{v}{p} \in  \eccqtemplate} \Concept{v} &\sqsubseteq \exists \ddot{p}^-.\Concept{u} , \\
          \exists \ddot{p}^-.\Concept{u} &\sqsubseteq \textstyle\bigsqcup_{\atomP{u}{v}{p} \in  \eccqtemplate} \Concept{v} \sqcup \textstyle\bigsqcup_{\atomGenp{x} \in H} \vconcept{x}^{p,o}, \\
          \exists \ddot{p}^-.\top &\equiv {\textstyle\bigsqcup_{\atomP{u}{v}{p} \in P}} (\Concept{v} \sqcap \exists \ddot{p}^-.\Concept{u}).
        \end{align*}

  \item
    For each $p \in \sccqpattern$, $\dot{p} \sqsubseteq p$.

  \item
    For each $p \in \sccqpattern$, $p \sqsubseteq \dot{p}$
    if all atomic patterns including $p$ in $P$ have the form $\atomP{x}{y}{p}$ where $x$, $y$ occur in no other atomic pattern in $\sccqpattern$ and $x \neq y$.

  \item
    For each pair $p,r$ with $\atomP{x}{y}{p} \in \sccqpattern$ and $\atomP{x}{y}{r} \in \sccqtemplate$ or $\atomP{y}{x}{r} \in \sccqtemplate$
    \begin{enumerate}
        \item $\dot{p} \sqsubseteq \ddot{r}$ (if $\atomP{x}{y}{r} \in \sccqtemplate$)
              or $\dot{p} \sqsubseteq \ddot{r}^-$ (if $\atomP{y}{x}{r} \in \sccqtemplate$) ---
              if $P$ does not contain any other atomic patterns with $p$, and
        \item $\ddot{r} \sqsubseteq \dot{p}$ (if $\atomP{x}{y}{r} \in \sccqtemplate$)
              or $\ddot{r}^- \sqsubseteq \dot{p}$ (if $\atomP{y}{x}{r} \in \sccqtemplate$) ---
              if $H$ does not contain any other atomic patterns with $r$.
    \end{enumerate}
  \end{enumerate}
  \begin{center}
    Cases 1, 2, and 5-7 are taken from, and 3 and 4 are adapted from~\cite{DBLP:conf/www/Seifer0LS24}.
  \end{center}
\end{definition}

\subsection{Enriching the Query}

Generic copy operations can invalidate output shapes in certain cases, as shown in \Cref{trfm:ex:fixed}.
On the other hand, they enable means of semantically enriching the \eccqname query, leading to
additional shapes that \emph{can} be shown to hold.
To this end, we define the extension rules in \Cref{trfm:def:semext} that essentially copy constraints from the 
query pattern to the template in cases where this is allowed by generic patterns.
We include an example in \Cref{trfm:sub:rax}.

\begin{definition}[Extension of \eccqname]\label{trfm:def:semext}
  The semantics-preserving extension of a \eccqname $q = \eccqformal$, written $\operatorname{Ext}(q)$, is defined by inserting the following atomic patterns
  to the template $\eccqtemplate$:
  \begin{enumerate}
    \item If $\atomGenc{x} \in \eccqpattern \cap \eccqtemplate$, then for each 
      $\atomC{x}{A} \in \eccqpattern$ such that $(\atomCNot{x}{A}) \not\in \eccqfilter$, 
      we include the atomic pattern $\atomC{x}{A}$ in $\eccqtemplate$, and
    \item if $\atomGenp{x} \in \eccqpattern \cap \eccqtemplate$, then for each 
      $\atomP{x}{a}{p} \in \eccqpattern$ such that $(\atomPNot{x}{p}) \not\in \eccqfilter$, 
      we include the atomic pattern $\atomP{x}{a}{p}$ in $\eccqtemplate$.
  \end{enumerate}
\end{definition}

\subsection{Example}
\label{trfm:sub:rax}

We now demonstrate some axioms in $\Sigma$ constructed for the problem inputs
$\progsexac{S}{1} = \{\progsexa{1}, \progsexa{2}\}$ (\Cref{trfm:ex:basicprogs}) 
and $\gcexa{2}$ (\Cref{trfm:ex:querymap}).
Again, we will focus on axioms related to concept names, as they are more intuitive,
yet in principle very similar to axioms related to role names.
As a first step, we need to map the query to \eccqname as shown in \Cref{trfm:ex:querymap}
and then extend the mapped query to obtain the
\eccqname input as $\eccqexa{2} = \operatorname{Ext}(\sgcoretoeccq{\gcexa{2}})$.
The extended query template includes the atomic patterns
$\{\atomC{\varname{e}}{\labelname{obs}}, \atomP{\varname{x}}{\stringvalue{Smith}}{\labelname{name}}\}$
in addition to the atomic patterns that are shown in \Cref{trfm:fig:ex:querymapping},
which are copied from the query pattern due to the presence of generic patterns
for $\varname{x}$ and $\varname{e}$.
On the other hand, since the only constraining atomic pattern for $\varname{y}$, 
namely $\atomC{\varname{y}}{\labelname{Person}}$,
is also included in the filter $\atomCNot{\varname{y}}{\labelname{Person}}$, there is no template extension
for variable $\varname{y}$, even though there are generic patterns (\cf~\Cref{trfm:def:semext}).

According to \Cref{trfm:def:cwa}.1 we need to include the input shapes in $\Sigma$ 
after mapping them to $\DLogics$ axioms via $\sprogstoshacl{\progsexac{S}{1}}$.
These axioms encode constraints that are known to hold on the input graph by definition,
since we only consider valid input graphs (\cf~\Cref{trfm:ex:mapshacl}).
According to \Cref{trfm:def:cwa}.2 we next define variable concepts encoding known constraints 
for the bindings of all query variables. 
We include the following axioms in $\Sigma$:
\begin{align*}
  \eevb &\equiv \exists \metn.\eyvb\ \sqcap\ \exists \mnte^{-}. \exvb\ \sqcap\ \labelname{obs}\\
  \exvb &\equiv \exists \labelname{name}.\labelname{Smith}\ \sqcap\ \exists \mnte.\eevb\\
  \eyvb &\equiv \exists \metn^{-}.\eevb\ \sqcap\ \labelname{Person}
\end{align*}

These variable concepts are almost direct translations of the query pattern;
for example, variable $\varname{x}$ occurs in $\atomP{\varname{x}}{\labelname{Smith}}{\labelname{name}}$ 
and $\atomPin{\varname{x}}{\varname{e}}$ (\Cref{trfm:ex:querymap}), which translates to the intersection of 
$\exists \labelname{name}.\labelname{Smith}$ and $\exists \mnte.\eevb$ in the definition of $\exvb$.
Note how reification again plays a role here, with the presence of the role name $\mnte$.

Constraints related to concept names that explicitly occur in the query template are 
defined according to \Cref{trfm:def:cwa}.3
in terms of these variable concepts, including the concept components that were introduced 
in \Cref{trfm:def:utilcc}.
In the following set, we also include one axiom inferred from the rules for properties,
namely $\eevb \sqsubseteq \exists \ddot{\metn}.\eyvb$.
\begin{align*}
  \ddot{\labelname{obs}} & \equiv \eevb \sqcup \cpart{e}{obs} 
                              & \ddot{\labelname{POI}} &\equiv \eyvb \sqcup \cpart{x}{POI} \sqcup \cpart{y}{POI}\\
  \ddot{\labelname{Agent}} &\equiv \cpart{x}{Agent} \sqcup \cpart{y}{Agent} & \cpart{x}{POI} &\equiv \exvb \sqcap \labelname{POI} \\
  \cpart{x}{Agent} &\equiv \exvb \sqcap \labelname{Agent} & \cpart{y}{POI} &\equiv \eyvb \sqcap \labelname{POI} \\
  \cpart{y}{Agent} &\equiv \eyvb \sqcap \labelname{Agent} & \eevb &\sqsubseteq \exists \ddot{\metn}.\eyvb \\
  \cpart{e}{obs} &\equiv \eevb
\end{align*}

There are a few noteworthy aspects to these axioms:
firstly, note how $\ddot{\labelname{POI}}$ (defined with dots, since it differs from the concept $\labelname{POI}$,
which may occur in input graphs) is defined in terms of $\eyvb$, because it exclusively occurs in the atomic
pattern $\atomC{\varname{y}}{\labelname{POI}}$ of the query template, 
and $\eyvb$ encodes the bindings of variable $\varname{y}$.
Furthermore, we must include the concept components $\cpart{x}{POI}$ (and also $\cpart{y}{POI}$, though
this is redundant in this case) in the definition of $\ddot{\labelname{POI}}$ since variable 
$\varname{x}$ occurs with the generic pattern $\atomGenc{\varname{x}}$,
meaning that the new concept $\ddot{\labelname{POI}}$ may extend over some bindings for $\varname{x}$, too.
We treat $\ddot{\labelname{Agent}}$ similarly, with the difference that there is no actual construct
pattern featuring $\labelname{Agent}$; thus, its definition consists only of concept components.

The complete set of inferred axioms entails, \eg,
$\ddot{\labelname{obs}} \sqsubseteq \exists \ddot{\metn} . \ddot{\labelname{POI}}$.
Therefore, we know that $\labelname{obs} \sqsubseteq \exists \metn . \labelname{POI}$ 
validates all output graphs of the query.
For the final shapes, we drop the namespace distinction, since it is not part
of our (or any regular) query semantics.

\section{Metatheory}
\label{trfm:sec:meta}

To encode our inference task in $\DLogics$, 
we first map shapes and queries to $\DLogics$ and \eccqname
then infer a set of axioms $\Sigma$ from these mapped queries and shapes (\cf~\Cref{trfm:algorithm}).
Therefore, we need to prove that any mapped shape $C \sqsubseteq D$ holds on all output graphs of the query
if $\Sigma \models \ddot{C} \sqsubseteq \ddot{D}$.
That is, if our inferred set of axioms $\Sigma$ entails the respective axiom (\ie encoded shape)
in the output graph name\-space annotated with two dots.

The primary outcome of the following considerations 
is the soundness of \Cref{trfm:algorithm}, as expressed in \Cref{trfm:maintheorem}.
To this end, we first prove the semantic equivalence for our mappings in
\Cref{trfm:prop:smap} and \Cref{trfm:prop:qmap}.
\fullproofsref

\begin{proposition}[Soundness of Shape Mapping]
   \label{trfm:prop:smap}
   Given a PG $G$ and a set of \simpleprogs shapes $S$, the graph $G$ is valid regarding each
   shape $s \in S$ if and only if the mapped graph $\pgtordf{G}$ is valid regarding 
   the mapped shapes $\sprogstoshacl{s}$.
\end{proposition}

\begin{proofsketch}
   We prove the proposition in two steps.
   First, we show that given a \simpleprogs shape $s$, the targets of $s$ in $G$ 
   correspond to the set of targets of $\sprogstoshacl{s}$ in $\pgtordf{G}$.
   Then, we show that for any node (or edge) in $G$, it conforms to the constraint of a shape $s$ if and only if
   the mapped node in $\pgtordf{G}$ conforms to the constraint of $\sprogstoshacl{s}$.
\end{proofsketch}

\begin{proposition} [Soundness of Query Mapping]
   \label{trfm:prop:qmap}
   Given a PG $G$ and a \gcname $q$, the following holds:
   $\gceval{q}{G} = \rdftopg{\eccqeval{\sgcoretoeccq{q}}{\pgtordf{G}}}$,
   where $\rdftopgname$ is the inverse mapping of $\pgtordfname$.
   That is, evaluating the query $q$ is the same as evaluating the mapped query on the mapped graph
   and then mapping the result graph back.
\end{proposition}

\begin{proofsketch}
   We separate the \emph{matching} and \emph{constructing} parts of queries for this proof.
   To this end, we rely on the composable intermediate query language (IQL), 
   making individual node or edge patterns compose neatly with full \eccqname patterns (or templates).
   We define an equivalence relation between sets of bindings $\Omega$ 
   for queries expressed as both \gcname and \eccqname, which is not straightforward, 
   given that \gcname queries bind only nodes and edges (while labels and properties are implicit), but the
   corresponding reified labels and properties in \eccqname are explicitly bound as variables.
   We prove that both \gcname and mapped \eccqname produce equivalent bindings.
   In the second part of the proof, we show 
   that\,---\,given sets of equivalent bindings\,---\,the 
   constructed graphs are again the same, modulo the mapping relation.
\end{proofsketch}

It remains to show that shapes entailed by the axioms $\Sigma$ constructed according to \Cref{trfm:def:cwa} are guaranteed to validate output graphs of the query, modulo renaming in the appropriate name\-space.

\begin{proposition}
   \label{trfm:prop:cwa}
   Given an \eccqname $q$ and a set of $\DLogics$ shapes $\shapesin$ it holds for all RDF graphs $G$ 
   with $\shaclvalid{G}{\shapesin}$ 
   that $\operatorname{infer}(\shapesin, q) \models \ddot{s}$ implies $\shaclvalid{\eccqeval{q}{G}}{s}$.

\end{proposition}

\begin{proofsketch}
   For this proof we first formalize the different conceptual graphs involved 
   (according to the name\-spaces separated with one or two dots) 
   as joined \emph{extended graphs} and then show
   for the axioms constructed by the rules of \Cref{trfm:def:cwa}
   that all extended graphs are indeed valid regarding these axioms.
   From this, we then follow (utilizing some auxiliary lemmas) that any other entailed shapes, 
   when renamed accordingly, validate all possible output graphs of the query.
   This proof builds on the proofs in~\cite{DBLP:conf/www/Seifer0LS24}.
\end{proofsketch}

Finally, we show that extending the \eccqname preserves query semantics.

\begin{proposition}
   \label{trfm:prop:semext}
   For any \eccqname $q$ and RDF graph $G$, $\eccqeval{q}{G} = \eccqeval{\operatorname{Ext}(q)}{G}$.
\end{proposition}

\begin{proofsketch}
   This proposition follows directly from the specified semantics of \eccqname
   since all concept or role names explicitly added to the template of $q$ 
   by rules for $\operatorname{Ext}(q)$ (\Cref{trfm:def:semext})
   are also constructed by the respective generic patterns $\atomGenc{x}$ and $\atomGenp{x}$, 
   except those that occur in filter conditions,
   which are excluded in the definition of $\operatorname{Ext}(q)$.
\end{proofsketch}

With these propositions, we can finally address the main claim of this paper.
We rewrite the theorem initially introduced as \Cref{trfm:maintheorem} in \Cref{trfm:maintheoremadapted},
replacing the function \emph{test} with concrete entailment and including the family of mappings.
The proof follows directly from the propositions discussed in this section.

\begin{theorem}(Reformulation of \Cref{trfm:maintheorem})
   \label{trfm:maintheoremadapted}
   Given a finite set of \simpleprogs shapes $\shapesin$, a \gcname $q = \gcformal$, and a 
   \simpleprogs shape $s$,
   then $\shaclvalid{\sparqleval{q}{\graphin}}{s}$ holds for every graph $\graphin$ 
   where $\shaclvalid{\graphin}{\shapesin}$ if
   $\operatorname{infer}(\sprogstoshacl{\shapesin}, \sgcoretoeccq{q}) \vdash \sprogstoshacl{\ddot{s}}$.
\end{theorem}

\begin{proofsketch}
  This follows directly from the propositions introduced in this section. 
\end{proofsketch}

\section{Implementation and Evaluation}
\label{trfm:app:impl}

We implemented the algorithm described in this paper, 
including both the mapping and inference components, in Scala.
To this end, we built on our previous implementation of~\cite{DBLP:conf/www/Seifer0LS24} 
(see \cite{darus-3977_2024}), 
which is available under a free software license.
Our fork of this implementation is available on GitHub\footnote{\url{https://github.com/softlang/s2s}}~\cite{DARUSV2} as well.

We implemented parsers, as well as our formally specified mappings, 
for both G-CORE to SPARQL \construct and ProGS to SHACL;
modified the inference algorithm of the implementation~\cite{darus-3977_2024} with our extended variant as specified,
constructing different sets of description logic axioms that are then passed 
to the HermiT~\cite{hermit} reasoner for inference;
implemented unit and end-to-end test cases that cover G-CORE and the extended SPARQL features;
and validated the algorithm based on query execution over randomly generated graphs.

\subsection{Experimental Setup and Results}

To validate the correctness of our implementation and gather empirical evidence, in addition
to the proofs provided by this work, on the correctness of our method, we generated random
problem instances, consisting of a set of input shapes and queries.
We generated random queries first and then sampled shapes based on their vocabularies, 
extended with additional random components that do not occur in queries.
For each sample input, we applied our algorithm to obtain a set of output shapes; to this end,
we generated a set of candidate shapes and checked each for entailment.
We only kept samples that produced at least one valid output shape.
We covered four different scenarios labelled corresponding to the nature of the shapes involved. 
While \emph{small} and \emph{large} mainly refer to the number of shapes,
the \emph{wide} sample includes conjunction and disjunction, whereas the \emph{deep}
sample includes more deeply nested shapes.

\Cref{trfm:fig:generators} shows the distribution of a few key properties across 
the experiments we performed, including the basic structure of queries and shapes, 
as well as the candidate shapes that were considered.
For additional details, we refer to our extended version or the documentation
of the implementation itself.
The upper half of the tables describes queries, while the lower half corresponds to shapes.
For both \emph{wide} and \emph{large} we report only data on shapes, since the queries
are mostly identical to the \emph{deep} scenario.
Regarding queries, we indicate the number of basic node or edge atomic clauses in the pattern
and template, the number of set, remove, or where clauses, and the distinct number of labels
and keys.
Regarding shapes, we include the number of sampled input shapes, their structure (number of basic
components and clauses in conjunction or disjunction), the number of candidates over the
vocabulary tested as potential output shapes, and the inferred set of output shapes,
\ie, the subset of candidates that is entailed by the constructed knowledge base.
Our scenarios were designed to cover various interactions of query structures
and shapes and facilitate our validation approach (random sample graphs), 
not necessarily to correspond to real-world data, where a baseline is also not available.

\begin{figure}[t]
    \centering
    \begin{subfigure}[t]{0.48\textwidth}
        \centering
        \begin{tabular}{lrr}
            \toprule
            Type & Avg. & Med. \\
            \midrule
            MATCH & 1.00 & 1.00 \\
            CONSTRUCT & 1.00 & 1.00 \\
            WHEN Clauses & 3.15 & 3.00 \\
            SET Clauses & 1.59 & 2.00 \\
            REMOVE Clauses & 0.99 & 1.00 \\
            Distinct Labels & 3.67 & 4.00 \\
            Distinct Keys & 1.31 & 1.00 \\
            \midrule
            Input Shapes & 1.00 & 1.00 \\
            Components & 2.16 & 2.00 \\
            Clauses & 1.00 & 1.00 \\
            Candidates & 256.24 & 240.00 \\
            Output Shapes & 7.54 & 2.00 \\
            \bottomrule
        \end{tabular}
        \caption{Sample \emph{small} (25,000)}
        \label{fig:gen_small}
    \end{subfigure}
    \hfill
    \begin{subfigure}[t]{0.48\textwidth}
        \centering
        \begin{tabular}{lrr}
            \toprule
            Type & Avg. & Med. \\
            \midrule
            MATCH & 1.21 & 1.00 \\
            CONSTRUCT & 2.21 & 2.00 \\
            WHEN Clauses & 3.20 & 3.00 \\
            SET Clauses & 1.57 & 2.00 \\
            REMOVE Clauses & 0.99 & 1.00 \\
            Distinct Labels & 3.64 & 4.00 \\
            Distinct Keys & 1.33 & 1.00 \\
            \midrule
            Input Shapes & 2.63 & 3.00 \\
            Components & 2.96 & 3.00 \\
            Clauses & 1.00 & 1.00 \\
            Candidates & 408.25 & 370.00 \\
            Output Shapes & 7.53 & 2.00 \\
            \bottomrule
        \end{tabular}
        \caption{Sample \emph{deep} (25,000)}
        \label{fig:gen_deep}
    \end{subfigure}
    \vspace{1em}
    \begin{subfigure}[t]{0.48\textwidth}
        \centering
        \begin{tabular}{lrr}
            \toprule
            Type & Avg. & Med. \\
            \midrule
            Input Shapes & 2.61 & 3.00 \\
            Components & 4.95 & 5.00 \\
            Clauses & 1.86 & 2.00 \\
            Candidates & 510.25 & 468.00 \\
            Output Shapes & 9.85 & 2.00 \\
            \bottomrule
        \end{tabular}
        \caption{Sample \emph{wide} (25,000)}
        \label{fig:gen_wide}
    \end{subfigure}
    \hfill
    \begin{subfigure}[t]{0.48\textwidth}
        \centering
        \begin{tabular}{lrr}
            \toprule
            Type & Avg. & Med. \\
            \midrule
            Input Shapes & 5.07 & 5.00 \\
            Components & 2.16 & 2.00 \\
            Clauses & 1.00 & 1.00 \\
            Candidates & 610.32 & 576.00 \\
            Output Shapes & 11.18 & 2.00 \\
            \bottomrule
        \end{tabular}
        \caption{Sample \emph{large} (25,000)}
        \label{fig:gen_large}
    \end{subfigure}
    \caption{Overview of statistical measures across different generator configurations.}
    \label{trfm:fig:generators}
\end{figure}

Our validation framework operates fully on the mapped SHACL and SPARQL level, 
taking the input SHACL shapes mapped from ProGS as the basis for 
generating a random RDF graph valid regarding these shapes.
The generation approach is rather brute force, to not bias the sample graphs with
additional information from shapes or queries, taking only information from the vocabulary
of the problem instance while removing any parts that violate input shapes.
We next applied the (mapped-from \gcorename) SPARQL \construct query to this sample graph
to obtain the corresponding output graph.
Finally, we validated the output graph with the inferred output shapes.
For each sample, we report either the class of a generation-related issue (leading to vacuous satisfaction
of output shapes), output graph validation failure (shape violations), or success (non-vacuously valid graphs).

In $100,000$ samples we observed \emph{no} validation failures, providing evidence towards
the correctness of our method and implementation.
A total of $37,959$ samples were full successes, 
meaning that the output graph was not only valid regarding the output shapes but also included
all their targets.
In $12,858$ cases, the output graph only included a subset of the targets present in output shapes,
while in $266$ cases, the input graph included only a subset of all input shape targets;
we consider these partial successes.
The remainder ($48,917$ samples) were other structural issues, where in most cases result graphs 
were either entirely empty ($26,591$) or did not include any of the required target nodes ($15,560$),
while in fewer cases we were unable to produce a suitable input graph ($6,622$), 
or the query exceeded the allowed execution time of $60$ seconds ($144$).
These results did not vary significantly between the different classes.

Note that in all of these cases, modulo the few timeouts where we do not know the result, 
the output shapes were still vacuously satisfied.
Since our validation attempts each sample generation up to $100$ times, we hypothesize that the reason 
for these cases mostly lies in the structure of the sample itself: 
in these cases, graphs conforming to input shapes
are unlikely to be matched by the query, since both do not align well.
Again, we decided against optimizing for these scenarios to avoid overfitting generated
samples to \emph{preferred} cases, which could bias results.

We validated the implementation of the 
validation itself with $1,000$ samples over the same four classes, where output shapes ($7$ per sample) 
were drawn randomly from the set of candidate shapes instead of utilizing the inference procedure; here, 
we obtain $76.7\%$ validation failures and only $0.2\%$ complete successes, 
indicating a very high probability for detecting failures.
Our complete results, including additional metadata for each sample, results for the individual classes of samples, 
all details regarding the structure and generation of samples, and the tools to replicate these and similar
experiments, are provided with our implementation.

\subsection{Runtime Performance}

While the original project, on which our implementation is based, includes tooling to demonstrate 
feasibility in terms of runtime performance (\cf~\cite{DBLP:conf/www/Seifer0LS24}), the results would be
difficult to compare between the different classes of queries and shapes used in our implementation.
Additionally, we lack datasets for realistic properties of \gcorename or ProGS queries and shapes.
Instead, we therefore report only a preliminary indication of the feasibility of runtime performance,
within the same order of magnitude as the prior results,
by measuring the inference time (testing all shape candidates reported in~\Cref{trfm:fig:generators}) 
for all $100,000$ samples generated for the correctness validation above.
Here, we observe average (median) inference times per sample of 
$13.94$ms ($11.00$ms) for \emph{small},
$34.81$ms ($25.00$ms) for \emph{deep},
$42.23$ms ($29.50$ms) for \emph{wide}, and
$51.68$ms ($36.00$ms) for \emph{large} samples on commodity hardware (Ryzen 7 9700X),
without any dedicated Scala optimization.

\section{Related Work}
\label{trfm:sec:related}

Mappings between query languages and related (graph) data models have been considered for various languages 
(\cf~\cite{DBLP:journals/corr/abs-1910-03118}), including mappings among graph query languages (\eg,~\cite{DBLP:conf/grades/ThakkarPLA18}) as well as
between graph query languages and SQL (\eg,~\cite{DBLP:journals/ws/Rodriguez-MuroR15,DBLP:journals/dke/ChebotkoLF09,DBLP:conf/grades/SteerALCVV17}).
Closely related to our mapping problem, SPARQL \select queries have been mapped to Cypher~\cite{DBLP:conf/dexa/ZhaoGSHW23,DBLP:conf/comad/AgrawalSM22}.
However, in this direction, the divide between RDF and PG is less challenging, given that PG can directly represent RDF graphs.
Reification of PG as RDF (\eg,~\cite{DBLP:conf/semweb/NguyenYTLBB19,DBLP:journals/access/TomaszukAT20}) 
has been discussed under consideration of querying performance~\cite{DBLP:conf/grades/KhayatbashiFH22}, 
though only under manually defined, ad hoc reformulation of Cypher queries in SPARQL \select.

At the data model level, the lack of interoperability between RDF and PG is a known issue~\cite{DBLP:conf/amw/AnglesTT19}, 
and mappings from RDF to PG have been considered as well, \eg,~\cite{DBLP:journals/access/AnglesTT20}.
In~\cite{DBLP:journals/pacmmod/RabbaniLBH24}, the authors consider both RDF graphs and SHACL shapes, transforming them to PG and PG-Schema, respectively.
RDF-star~\cite{DBLP:conf/amw/Hartig17} is an extension to RDF for the representation of metadata as nested triples, allowing for similar features as PG,
thus working towards bridging the formal gap between these data models.
Indeed, both PG mappings to RDF-star~\cite{DBLP:conf/i-semantics/Hartig19,DBLP:conf/grades/KhayatbashiFH22}, 
and the reverse~\cite{DBLP:conf/semweb/AbuodaDKH22} have been defined.
In particular, this also includes execution of SPARQL-star queries over PG~\cite{DBLP:conf/i-semantics/Hartig19}.

In contrast, in our work, we leverage reification in pure RDF for the sake of utilizing DL via SHACL,
and we are concerned with bridging the semantic differences between graph construction in \gcorename,
with its implicit labels and properties,
and SPARQL \construct, with its explicit triple construction,
which are not addressed in previous work.

Inference of shapes or schemas can be approached essentially in two ways: 

\emph{Option 1}~---~inference from instance data, \eg, concrete graphs or tables; 

\emph{Option 2}~---~inference from queries and other functions constructing data.

\noindent
The first option has been considered for RDF using various algorithms based on logical, statistical, or ML-driven inference
for constructing SHACL or ShEx~\cite{DBLP:conf/i-semantics/PrudhommeauxGS14} shapes from sample graph instances~\cite{DBLP:conf/semweb/SpahiuMP18,DBLP:journals/kbs/Fernandez-Alvarez22,DBLP:journals/pvldb/RabbaniLH23,DBLP:conf/sac/Mihindukulasooriya18,DBLP:journals/semweb/OmranTMH23,DBLP:conf/icdt/GrozLSW22,DBLP:conf/semweb/BonevaDFG19}.
Similar approaches have been considered recently for PG and their schema languages~\cite{DBLP:conf/edbt/LbathBH21,DBLP:conf/edbt/BonifatiDM22,DBLP:conf/icdt/GrozLSW22}.

In our work, we follow the second option and do not consider any instance data
but instead infer shapes valid in the context of any possible graph that queries might operate on in the future. 
However, inference from concrete graphs may complement our approach when initial input shapes for queries should be determined. 
Our approach leverages a prior method~\cite{DBLP:conf/www/Seifer0LS24}, as discussed throughout the paper, but we cover PG and G-CORE through mappings; 
in due course, we also extend the subset of SPARQL considered in~\cite{DBLP:conf/www/Seifer0LS24} with generic copy patterns.

Related work on the inference of SHACL shapes from direct mappings~\cite{DBLP:conf/semweb/ThapaG21} 
and RML rules~\cite{DBLP:conf/kcap/DelvaSOALD21} also follows the second option, 
but these approaches lack the arbitrary construction through query patterns we support, 
in addition to explicit input constraints.
Views over relational databases can be considered equivalent mechanisms to composable graph queries;
SQL schemas have been inferred for such views in~\cite{DBLP:journals/tods/KlugP82,DBLP:journals/pvldb/FanMHLW08,DBLP:conf/sigmod/Stonebraker75,DBLP:journals/tods/JacobsAK82}
though suffer from infeasibility of first-order formulas (as compared to DL) or restricted semantics of 
functional or join dependencies, which are orthogonal to constructs supported by ProGS in our case.
\cite{DBLP:conf/pods/BonevaGHMS23} solves a similar problem, focusing on basic cardinality constraints 
(as graph schemas) and formally defined graph transformations (conjunctive path queries) 
over a concrete query language.
They, too, rely on description logics as an encoding mechanism.

\section{Conclusion}
\label{trfm:sec:conclusion}

In this paper, we have presented an algorithm for inferring ProGS shapes 
validating all possible property graphs constructed by queries specified in a subset of G-CORE.
Our algorithm encodes constraints inferred from the underlying query, 
as well as any known shapes that hold on the inputs of the query.
Internally, we rely on the description logic $\DLogics$ 
by constructing a knowledge base that entails a sound set of shapes 
validating all output graphs of this query.

Since this DL encoding requires a reification step 
for first-class edges with labels and properties, 
we choose a clean, intermediate abstraction layer between property graphs and $\DLogics$, 
namely the SPARQL \construct query language and SHACL shapes, 
over property graphs reified in RDF.
This also allows us to reuse an underlying algorithm defined for these languages, 
which we must extend to support a larger subset of SPARQL, 
including the features required by our mapping.
We prove the soundness of this shape inference approach 
and the semantic equivalence of the shapes and queries constructed by our mappings.

As future work, a larger, non-conjunctive subset of G-CORE could be supported by extending 
both the mapping and the capabilities of the underlying inference approach on SPARQL queries,
supporting operators such as UNION in SPARQL.
Such extensions may approach a limit when considering additional generic 
or type-level query constructs.
For example, while we support a specific set of generic patterns in SPARQL, 
arbitrary constraints on type-level variables 
(\ie variables that bind concepts) cannot be easily encoded in DL.

The query mapping itself also poses interesting questions:
While in our work, we do not intend to execute the SPARQL \construct queries, 
the mappings could inform implementations of G-CORE over SPARQL triple stores.
To this end, the RDF reification would have to be investigated empirically 
regarding runtime and space performance and possibly adapted, 
given that our mapping favors a formal ease of use over efficiency concerns.

\bibliography{paper}

\newpage
\appendix

\section[Appendix A -- Example]{Concrete Example}
\label{trfm:app:concret}

\Cref{trfm:lst:gone} and \Cref{trfm:lst:stwo} show how the example from \Cref{trfm:fig:ex:querymapping}
could be expressed in concrete G-CORE and SPARQL syntax, respectively.
We also present the full experimental data from our empirical validation (\Cref{trfm:app:impl})
in 
\Cref{trfm:tab:datafullone},
\Cref{trfm:tab:datafulltwo},
\Cref{trfm:tab:datafullthree}, and
\Cref{trfm:tab:datafullfour}.

\begin{figure}[H]
\begin{lstlisting}[
    caption={The G-CORE query from \Cref{trfm:fig:ex:querymapping} expressed using concrete G-CORE syntax.
    In this case, all constraints, as well as set and remove clauses are expressed in separate blocks.
    Alternatively, in the MATCH clause, \eg, \texttt{(y)} could also be expresses as \texttt{(y:Person)},
    resembling more closely the abstract syntax we introduce in \Cref{trfm:def:g:syntax}.},
    label=trfm:lst:gone,language=gcore,frame=tlrb]
CONSTRUCT (x)-[e]->(y)
  SET y:POI
  REMOVE y:Person
MATCH (x)-[e]->(y)
WHERE x.name = "Smith"
  AND y:Person
  AND e:obs
\end{lstlisting}
\end{figure}

\begin{figure}[H]
\begin{lstlisting}[
    caption={The SPARQL query from \Cref{trfm:fig:ex:querymapping} expressed using concrete SPARQL syntax.
    \emph{Fresh} (unique) variables introduced for generic copying operations are prefixed with underscores.
    Note the use of filter expressions on these variables.},
    label=trfm:lst:stwo,language=sparql,frame=tlrb]{sparql}
CONSTRUCT {
    ?x a m:node . ?x ?_1 ?_2 .
    ?x m:nte ?e . ?e a m:edge . ?e ?_3 ?_4 . 
    ?e m:etn ?y . ?y a m:node . ?y ?_5 ?_6 . ?y a ln:POI
} WHERE {
    ?x a m:node . ?x ?_1 ?_2 . ?x kn:name v:Smith . 
    ?x m:nte ?e . ?e a m:edge . ?e ?_3 ?_4 . ?e a le:obs . 
    ?e m:etn ?y . ?y a m:node . ?y ?_5 ?_6 . ?y a ln:Person .
    FILTER ( ?_1 != m:nte ) .
    FILTER ( ?_3 != m:etn ) .
    FILTER ( ?_5 != m:nte ) .
    FILTER ( ?_6 != ln:Person )
}
\end{lstlisting}
\end{figure}

\begin{table}[htbp]
\centering
\caption{gen\_small (25000 samples) --- Time: 13.94s (Median: 11.00s)}
\label{trfm:tab:datafullone}
\begin{tabular}{lrrrr}
\toprule
Type & Min & Max & Average & Median \\
\midrule
Node Variables & 1.00 & 4.00 & 2.27 & 2.00 \\
Edge Variables & 0.00 & 2.00 & 1.50 & 2.00 \\
Atoms MATCH & 1.00 & 1.00 & 1.00 & 1.00 \\
Atoms CONSTRUCT & 1.00 & 1.00 & 1.00 & 1.00 \\
WHEN Clauses & 0.00 & 5.00 & 3.15 & 3.00 \\
SET Clauses & 0.00 & 3.00 & 1.59 & 2.00 \\
REMOVE Clauses & 0.00 & 2.00 & 0.99 & 1.00 \\
Distinct Labels & 0.00 & 9.00 & 3.67 & 4.00 \\
Distinct Keys & 0.00 & 5.00 & 1.31 & 1.00 \\
Input Shapes & 1.00 & 1.00 & 1.00 & 1.00 \\
Output Shapes & 1.00 & 186.00 & 7.54 & 2.00 \\
Cand. Shapes & 20.00 & 862.00 & 256.24 & 240.00 \\
Shape Comp. & 1.00 & 3.00 & 2.16 & 2.00 \\
Shape Clauses & 1.00 & 1.00 & 1.00 & 1.00 \\
Shape Negation & 0.00 & 1.00 & 0.07 & 0.00 \\
\bottomrule
\end{tabular}
\end{table}

\begin{table}[htbp]
\centering
\caption{gen\_deep (25000 samples) --- Time: 34.81s (Median: 25.00s)}
\label{trfm:tab:datafulltwo}
\begin{tabular}{lrrrr}
\toprule
Type & Min & Max & Average & Median \\
\midrule
Node Variables & 1.00 & 5.00 & 3.52 & 4.00 \\
Edge Variables & 0.00 & 6.00 & 2.54 & 2.00 \\
Atoms MATCH & 1.00 & 3.00 & 1.21 & 1.00 \\
Atoms CONSTRUCT & 1.00 & 3.00 & 2.21 & 2.00 \\
WHEN Clauses & 0.00 & 5.00 & 3.20 & 3.00 \\
SET Clauses & 0.00 & 3.00 & 1.57 & 2.00 \\
REMOVE Clauses & 0.00 & 2.00 & 0.99 & 1.00 \\
Distinct Labels & 0.00 & 9.00 & 3.64 & 4.00 \\
Distinct Keys & 0.00 & 5.00 & 1.33 & 1.00 \\
Input Shapes & 1.00 & 4.00 & 2.63 & 3.00 \\
Output Shapes & 1.00 & 262.00 & 7.53 & 2.00 \\
Cand. Shapes & 20.00 & 1740.00 & 408.25 & 370.00 \\
Shape Comp. & 1.00 & 4.00 & 2.96 & 3.00 \\
Shape Clauses & 1.00 & 1.00 & 1.00 & 1.00 \\
Shape Negation & 0.00 & 4.00 & 0.28 & 0.00 \\
\bottomrule
\end{tabular}
\end{table}

\begin{table}[htbp]
\centering
\caption{gen\_wide (25000 samples) --- Time: 42.23s (Median: 29.50s)}
\label{trfm:tab:datafullthree}
\begin{tabular}{lrrrr}
\toprule
Type & Min & Max & Average & Median \\
\midrule
Node Variables & 1.00 & 5.00 & 3.52 & 4.00 \\
Edge Variables & 0.00 & 6.00 & 2.53 & 2.00 \\
Atoms MATCH & 1.00 & 3.00 & 1.21 & 1.00 \\
Atoms CONSTRUCT & 1.00 & 3.00 & 2.20 & 2.00 \\
WHEN Clauses & 0.00 & 5.00 & 3.18 & 3.00 \\
SET Clauses & 0.00 & 3.00 & 1.56 & 2.00 \\
REMOVE Clauses & 0.00 & 2.00 & 0.98 & 1.00 \\
Distinct Labels & 0.00 & 9.00 & 3.61 & 4.00 \\
Distinct Keys & 0.00 & 5.00 & 1.33 & 1.00 \\
Input Shapes & 1.00 & 4.00 & 2.61 & 3.00 \\
Output Shapes & 1.00 & 256.00 & 9.85 & 2.00 \\
Cand. Shapes & 26.00 & 1884.00 & 510.25 & 468.00 \\
Shape Comp. & 1.00 & 7.00 & 4.95 & 5.00 \\
Shape Clauses & 1.00 & 2.00 & 1.86 & 2.00 \\
Shape Negation & 0.00 & 4.00 & 0.39 & 0.00 \\
\bottomrule
\end{tabular}
\end{table}

\begin{table}[htbp]
\centering
\caption{gen\_large (25000 samples) --- Time: 51.68s (Median: 36.00s)}
\label{trfm:tab:datafullfour}
\begin{tabular}{lrrrr}
\toprule
Type & Min & Max & Average & Median \\
\midrule
Node Variables & 1.00 & 5.00 & 3.54 & 4.00 \\
Edge Variables & 0.00 & 6.00 & 2.55 & 3.00 \\
Atoms MATCH & 1.00 & 3.00 & 1.22 & 1.00 \\
Atoms CONSTRUCT & 1.00 & 3.00 & 2.21 & 2.00 \\
WHEN Clauses & 0.00 & 5.00 & 3.16 & 3.00 \\
SET Clauses & 0.00 & 3.00 & 1.54 & 2.00 \\
REMOVE Clauses & 0.00 & 2.00 & 0.98 & 1.00 \\
Distinct Labels & 0.00 & 9.00 & 3.59 & 4.00 \\
Distinct Keys & 0.00 & 5.00 & 1.32 & 1.00 \\
Input Shapes & 3.00 & 7.00 & 5.07 & 5.00 \\
Output Shapes & 1.00 & 258.00 & 11.18 & 2.00 \\
Cand. Shapes & 56.00 & 2502.00 & 610.32 & 576.00 \\
Shape Comp. & 1.00 & 3.00 & 2.16 & 2.00 \\
Shape Clauses & 1.00 & 1.00 & 1.00 & 1.00 \\
Shape Negation & 0.00 & 4.00 & 0.32 & 0.00 \\
\bottomrule
\end{tabular}
\vspace{-0.6cm}
\end{table}

\section[Appendix B -- Proofs]{Proofs}
\label{trfm:app:proofs}

This section includes proofs for all propositions in the main paper, 
starting with some preliminary definitions used throughout these proofs.

\subsection{Preliminaries}
\label{trfm:proof:prelim}

We first summarize the validation semantics of $\DLogics$-based SHACL shapes we introduced 
in~\cite{DBLP:conf/www/Seifer0LS24} (based on work by Bogaerts \etal~\cite{DBLP:conf/lpnmr/BogaertsJB22}), 
including also the semantics of $\DLogics$, for the sake of completeness.

\begin{definition}[Summary of $\DLogics$ Semantics (\cf~\cite{DBLP:conf/dlog/2003handbook})]
  An $\DLogics$ \emph{knowledge base} $\mathcal{K}$ is a pair $(\TBox,\ABox)$,
  where $\TBox$ is a finite set of axioms and $\ABox$ is a finite set of assertions.
  In a slight abuse of notation, given an ABox $\ABox$, we write $\ABox$ to refer to the knowledge base 
  $(\emptyset, \ABox)$, and given a TBox $\TBox$ we write $\TBox$ to refer to $(\TBox, \emptyset)$.
  An \emph{interpretation} $\Int$ is a pair $(\Delta^\Int, \cdot^\Int)$ consisting of a set 
  $\Delta^\Int$, called the \emph{domain}, and
  a function $\cdot^\Int$ such that we have
  for each individual name $a \in \IndividualNames$, an element $a^\Int \in \Delta^\Int$;
  for each concept name $A \in \ConceptNames$, a subset $A^\Int \subseteq \Delta^\Int$; and
  for each role name $p \in \RoleNames$, a relation $p^\Int \subseteq \Delta^\Int \times \Delta^\Int$.
  The function $\cdot^\Int$ is extended to concept descriptions as follows: 

  \begin{equation*}
    \begin{aligned}[t]
      \bot^\Int &= \emptyset, \\
      \top^\Int &= \Delta^\Int,\\
      \{a\}^\Int &= \{a^\Int\},\\
      (C \sqcap D)^\Int &= C^\Int \cap D^\Int,\\
      (C \sqcup D)^\Int &= C^\Int \cup D^\Int,\\
      (\neg C)^\Int &= \top^\Int \setminus C^\Int\\
      (\exists p.C)^\Int
      &=
        \{ d \in \Delta^\Int \mid (d,e) \in p^\Int\ \text{with}\ e \in C^\Int  \},\\ 
      (\exists p^-.C)^\Int
      &=
        \{ d \in \Delta^\Int \mid (e,d) \in p^\Int\ \text{with}\ e \in C^\Int\},\\
      (\forall p.C)^\Int
      &=
        \{ d \in \Delta^\Int \mid \text{for all } e \in \Delta^\Int,
        \\ & \qquad \qquad \qquad \text{if } (d,e) \in p^\Int \text{ then } e \in C^\Int \},\\
      (\forall p^-.C)^\Int
      &=
        \{ d \in \Delta^\Int \mid \text{for all } e \in \Delta^\Int,
        \\ & \qquad \qquad \qquad \text{if } (e,d) \in p^\Int \text{ then } e \in C^\Int \}.
    \end{aligned}
  \end{equation*}
   
  An interpretation $\Int$ is a \emph{model} of a knowledge base $\KB = (\TBox,\ABox)$ 
  if and only if $C^\Int \subseteq D^\Int$ for every axiom $C \sqsubseteq D$ in $\TBox$; 
  $p^\Int \subseteq r^\Int$ for every axiom $p \sqsubseteq r$ in $\TBox$; 
  $a^\Int \in C^\Int$ for every assertion $\atomC{a}{C}$ in $\ABox$;
  and $(a^\Int, b^\Int) \in p^\Int$ for every assertion $\atomP{a}{b}{p}$ in $\ABox$. 
  Given two knowledge bases, $\KB_1$ and $\KB_2$, $\KB_1$ \emph{entails} $\KB_2$, 
  denoted $\KB_1 \models \KB_2$, if and only if every model $\Int$ of $\KB_1$ is also a model of $\KB_2$.
\end{definition}

For the validation semantics, we interpret RDF graphs as $\DLogics$ A-Boxes $G$, and define SHACL validation based
on the validation knowledge base $(\TBox_{G}, G)$, relying on \Cref{trfm:def:validation-rdf-semantics}.
We say that a graph $G$ is \emph{proof-valid} regarding a set of $\DLogics$ axioms (\eg, a set of shapes) 
if and only if this set is consistent with this validation knowledge base of $G$ (\Cref{trfm:def:alchoi-shape-semantics}).

\begin{definition}[RDF Graph Validation Semantics]%
  \label{trfm:def:validation-rdf-semantics}
  The $\DLogics$ axioms $\TBox_G$ of an RDF graph $G$ are the TBox consisting of the:
  \begin{enumerate}
  \item \emph{Domain Closure Assumption (DCA)}:
    $\top \equiv \bigsqcup_{a \in \IndividualNames}\{a\}$.
  \item \emph{Unique Name Assumption (UNA)}:
   $\{a\} \sqcap \{b\} \equiv \bot$, for each pair of distinct individual names $a, b \in \IndividualNames$.
  \item \emph{Closed-World Assumption (CWA)}:
    \begin{itemize}
      \item $A \equiv \bigsqcup_{\atomC{a}{A}\in G}\{a\}$, for each concept name $A \in \ConceptNames$,
      \item $\exists p.\{a\} \equiv \bigsqcup_{\atomP{b}{a}{p}\in G} \{b\}$, and
      \item $\exists p^-.\{a\} \equiv \bigsqcup_{\atomP{a}{b}{p}\in G} \{b\}$, 
        for each role name $p \in \RoleNames$, and for each individual name $a \in \IndividualNames$.
    \end{itemize}
  \end{enumerate}
\end{definition}

\begin{definition}[Semantics of $\DLogics$ Shapes]
  \label{trfm:def:alchoi-shape-semantics}
  A graph $G$ is \emph{valid} regarding a set $\allshapes$ of $\DLogics$ shapes, denoted $\shaclvalid{G}{\allshapes}$, if and only if $G$ is proof-valid according to $\allshapes$;
  that is, if and only if $S$ is consistent with this validation knowledge base of $G$.
\end{definition}

We next formalize extended graphs, that were informally introduced in \Cref{trfm:sec:extend}. 
An extended graph unifies four different graphs: The input graph (which is given), the intermediate graph defined
by renaming concept and role names matched by the query, the output graph constructed by the query (again modulo renaming),
as well as the graph defined by variables and their bindings.
Note, that the notion of extended graphs is only formally required in proofs; they are not explicitly constructed in our method.
\begin{definition}[Extended Graph]\label{trfm:def:extended-graph}
  Given an RDF graph $\graphin$ and a query $\eccqformal$, we define
  \begin{enumerate}
    \item the \emph{intermediate graph} $\graphmed \coloneqq \bigcup_{\mu \in \eccqpatterneval{P,F}{\graphin}} \mu(P)$,
    \item the \emph{variable concept graph} $\graphvar$ containing an assertion $\atomC{a}{\vconcept{x}}$ if and only if there exists a valuation 
      $\mu \in \eccqpatterneval{P,F}{\graphin}$ and $\mu(x)=a$,
    \item the \emph{output graph} $\graphout \coloneqq \eccqeval{q}{\graphin}$,
    \item and the \emph{extended graph} $\graphext \coloneqq \graphin \cup \dot{G}_{\mathrm{med}} \cup \graphvar \cup \ddot{G}_{\mathrm{out}}$.
\end{enumerate}
\end{definition}

\subsection[Shape Mapping]{Shape Mapping (\Cref{trfm:prop:smap})} 
\label{trfm:proof:smap}

We prove, that the mapping for shapes, $\sprogstoshacl{s}$, preserves the validation semantics of ProGS.

\begin{proof}
We consider shape targets and constraints separately.
To this end, we first show that if and only if a node $n$ (or edge $e$) 
is a target of the shape $s$ in a property-graph $G$,
then the corresponding node $n$ in $\pgtordf{G}$ with the name $\toiri{n}$ 
is a target of the shape $\sprogstoshacl{s}$.
We differentiate between the possible target queries of a shape $s$:
\begin{enumerate}
  \item $n$.
        The target query $n$ targets exactly one node in $G$, namely $n$, which is \emph{uniquely} mapped to $\toiri{n}$ in $\pgtordf{G}$, per \Cref{trfm:def:graphmap}.
        The target query $\sprogstoshacl{n} = \{\toiri{n}\}$ targets all individuals in the extension of $\{\toiri{n}\}$ in $\pgtordf{G}$, which by definition is $\toiri{n}$ if and only if $\toiri{n}$ is in $\pgtordf{G}$, which is the case if and only if $n$ was in $G$.
        Thus, if $n$ is a target in $G$, then (and only then) is $\toiri{n}$ a target in $\pgtordf{G}$.

  \item $l_N$.
        This targets all nodes $n \in N$ such that $l_{N} \in \lambda(n)$ in $G$ (\Cref{trfm:fig:nodetargetqueries}).
        The corresponding RDF graph $\pgtordf{G}$ includes, by construction, exactly one triple
        $\triple{\toiri{n}}{\rdftype}{\toiri{l_{N}}}$
        for each $n \in N$, and no other occurrences of
        the concept name $\toiri{l_{N}}$ (\Cref{trfm:def:graphmap}).
        The target query $\sprogstoshacl{l_{N}}$ is mapped to $\toiri{l_{N}}$ (\Cref{trfm:def:shapemap}).
        By definition, this targets exactly the individuals in the extension of $\toiri{l_{N}}$; in the constructed graph this includes exactly the subjects of
        triples of the form $\triple{\toiri{n}}{\rdftype}{\toiri{l_{N}}}$, and thus exactly the mapped nodes $\toiri{n}$ for all $n$ identified above.

  \item $k_N$.
        This targets all nodes $n \in N$ where $\sigma(n, k_{N}) \ne \emptyset$ (\Cref{trfm:fig:nodetargetqueries}).
        The corresponding RDF graph $\pgtordf{G}$ includes, by construction, the triple
        $\triple{\toiri{n}}{\toiri{k_N}}{\toiri{v}}$ if and only if $\sigma(n, k_{N}) = \{v\}$,
        and for each $n \in N$, and no other occurrences of
        the role name $\toiri{k_{N}}$ (\Cref{trfm:def:graphmap}).
        (Note, that by definition we consider only singleton sets of values.)
        The target query maps to $\sprogstoshacl{k_{N}} = \exists \toiri{k_{N}}.\top$ (\Cref{trfm:def:shapemap}).
        By definition, this target query includes all nodes in the extension of $\exists \toiri{k_{N}}.\top$,
        and with analogous reasoning to the previous case exactly the mapped nodes $\toiri{n}$ for all $n$ identified above.

  \item The remaining cases for $e$, $l_E$, and $k_E$ are equivalent to the previous cases.
\end{enumerate}

Next, we consider the constraints.
That is, we first show that when a node $n$ (or edge $e$) satisfies the constraint of shapes $s$ 
in a property-graph $G$, then there is a corresponding node $n$ in $\pgtordf{G}$ with the name $\toiri{n}$ 
that satisfies the constraint of $\sprogstoshacl{s}$.
We show this by assuming that this corresponding node does not exist; 
we cover all cases for constraint components, by induction. 
We start with the base cases.

\begin{enumerate}
  \item $\top$.
        This constraint is vacuously satisfied for any target node.
        The mapped constraint $\sprogstoshacl{\top} = \top$ is also always vacuously satisfied for any target node.

    \item $n$.
        This constraint is satisfied for a target node $n_{t}$ in $G$ if and only if $n_{t} = n$.
        Assume that the node $\toiri{n_{t}}$ in $\pgtordf{G}$ would not satisfy the constraint $\sprogstoshacl{n} = \{\toiri{n}\}$.
        This would imply that $\toiri{n_{t}} \ne \toiri{n}$.
        However, we know that $n_{t} = n$ from which $\toiri{n_{t}} = \toiri{n}$ follows directly.

    \item $l_{N}$.
        This constraint is satisfied for a target node $n_{t}$ if and only if $l_{N} \in \lambda(n_{t})$ in $G$.
        Assume that the node $\toiri{n_{t}}$ in $\pgtordf{G}$ would not satisfy the constraint $\sprogstoshacl{l_{N}} = \toiri{l_{N}}$.
        Then, the individual $\toiri{n_{t}}$ must not be in the extension of $\toiri{l_{N}}$.
        However, by construction, the graph $\pgtordf{G}$ contains the triple $\triple{\toiri{n_{t}}}{\rdftype}{\toiri{l_{N}}}$ and thus $\toiri{n}$ is in the extension of $\toiri{l_{N}}$.
        This contradicts the previous assumption.

        Indeed, the same can be shown for the inverse direction:
        Consider an individual $\toiri{n_{t}}$ that is in the extension of $\toiri{l_{N}}$ and thereby satisfying this constraint in $\pgtordf{G}$.
        Now assume that $\pgtordf{G}$ is constructed from the graph $G$ where the corresponding node $n_{t}$ does not satisfy the corresponding constraint $l_{N}$, \ie it holds that $l_{N} \not\in \lambda(n_{t})$.
        If this is the case, then $\pgtordf{G}$ does not contain the triple $\triple{\toiri{n_{t}}}{\rdftype}{\toiri{l_{N}}}$, since the only way to construct that triple in the mapping is if $l_{N} \in \lambda(n_{t})$.
        However, if the graph does not contain this triple, then $\toiri{n_{t}}$ cannot be in the extension of $\toiri{l_{N}}$, which contradicts the previous assumption.

    \item $\exists k_{N}.(= v)$.
        This constraint is satisfied for a target node $n_{t}$ if and only if $\sigma(n_{t}, k_{N}) = \{v\}$.
        Assume that the individual $\toiri{n_{t}}$ in $\pgtordf{G}$ would not satisfy $\sprogstoshacl{\exists k_{N}.\top} = \toiri{\exists k_{N}.\top}$.
        Then, $\toiri{n_{t}}$ must not be in the extension of $\exists \toiri{k_{N}}.\toiri{v}$.
        However, by construction, the graph $\pgtordf{G}$ contains the triple $\triple{\toiri{n_{t}}}{\toiri{k_{N}}}{\toiri{v}}$ if and only if $\sigma(n_{t}, k_{N}) = \{v\}$.
        Since this is the case, we can follow that $\toiri{n_{t}}$ is indeed in the extension of $\exists \toiri{k_{N}}.\toiri{v}$.
        This is a contradiction, so our previous assumption must be false.

        Indeed, the same can be shown for the inverse direction:
        Consider an individual $\toiri{n_{t}}$ that is in the extension of $\exists \toiri{k_{N}}.\toiri{v}$, thereby satisfying the constraint in $\pgtordf{G}$.
        Now assume that $\pgtordf{G}$ is constructed from the graph $G$ where the corresponding node $n_{t}$ does not satisfy the corresponding constraint $\exists k_{N}.(= v)$, \ie it holds that $v \not\in \sigma(n_{t}, k_{N})$.
        If this is the case, then $\pgtordf{G}$ does not contain the triple $\triple{\toiri{n_{t}}}{\toiri{k_{N}}}{\toiri{v}}$, since the only way to construct that triple in the mapping is if $\sigma(n_{t}, k_{N}) = \{v\}$.
        However, if the graph does not contain this triple, then $\toiri{n_{t}}$ cannot be in the extension of $\toiri{\exists \toiri{k_{N}}.\toiri{v}}$, which contradicts the previous assumption.

    \item $\exists k_{N}.\top$.
        The proof for this case is very similar to the previous case, with the only difference that we only care about $\sigma(n_{t}, k_{N})$ being non-empty.

    \item $\neg \phi_{N}$.
        We assume that the condition under investigation holds for $\phi_{N}$.
        That is, for any target node $n_{t}$ in $G$ that it satisfies $\phi_{N}$ if and only if,
        $\toiri{n_{t}}$ satisfies $\sprogstoshacl{\phi_{N}}$ in $\pgtordf{G}$.
        The negated cases follow directly.

    \item $\phi_{N}^{1} \wedge \phi_{N}^{2}$.
        We assume that the condition under investigation holds for $\phi_{N}^{1}$ and $\phi_{N}^{2}$.
        That is, for any target node $n_{t}$ in $G$ that it satisfies $\phi_{N}^{{1}}$ if and only if,
        $\toiri{n_{t}}$ satisfies $\sprogstoshacl{\phi_{N}^{{1}}}$ in $\pgtordf{G}$ (and the same for $\phi_{N}^{2}$, respectively).
        The conjunction follows directly.

    \item We omit equivalent cases for edges (\eg, $\phi_{E}$) since their proofs work analogously.

    \item $\existsright \phi_{E}$.
        This constraint is satisfied if and only if for the target node $n_{t}$ exists $\rho(e) = (n_{t}, n')$, and $e$ satisfies the constraint $\phi_{E}$.
        Assume that the individual $\toiri{n_{t}}$ in $\pgtordf{G}$ would not satisfy the constraint $\sprogstoshacl{\existsright \phi_{E}} = \exists \pref{m}{nte} . \sprogstoshacl{\phi_{E}}$.
        By construction, we know that $\pgtordf{G}$ contains the triple $\triple{\toiri{n_{t}}}{\pref{m}{nte}}{\toiri{e}}$.
        By our induction hypothesis, we can assume that $\toiri{e}$ satisfies the constraint $\sprogstoshacl{\phi_{E}}$ in $\pgtordf{G}$, since we know that $e$ satisfies $\phi_{E}$ in $G$.
        Because of both of these facts, we can follow that $\toiri{n_{t}}$ is in the extension of $\exists \pref{m}{nte}.\sprogstoshacl{\phi_{E}}$, and therefore satisfies this constraint.
        This contradicts the previous assumption.

        For the inverse direction, the same argument applies:
        Consider an individual $\toiri{n_{t}}$ that is in the extension of $\exists \pref{m}{nte}.\sprogstoshacl{\phi_{E}}$, thereby satisfying the constraint in $\pgtordf{G}$.
        Now assume that $\pgtordf{G}$ is constructed from the graph $G$ where the corresponding node $n_{t}$ does not satisfy the corresponding constraint $\existsright \phi_{E}$, \ie it holds that either $\rho(e) = (n_{t}, n')$ is not in $\rho$, or $e$ does not satisfy $\phi_{E}$.
        By the induction hypothesis, we know that $e$ must satisfy $\phi_{E}$.
        Thus, the edge $\rho(e) = (n_{t}, n')$ must not exist.
        If this is the case, then $\pgtordf{G}$ does not contain the triple $\triple{\toiri{n_{t}}}{\pref{m}{nte}}{\toiri{e}}$, 
        since the only way to construct that triple in the mapping is if there is an $e$ where $\rho(e) = (n_{t}, n')$.
        However, if the graph does not contain this triple, then $\toiri{n_{t}}$ cannot be in the extension 
        of $\exists \pref{m}{nte}.\sprogstoshacl{\phi_{E}}$, which contradicts the previous assumption.

    \item $\existsleft \phi_{E} $.
        This case works exactly analogously to the previous case (only in reverse).

    \item $\Rightarrow \phi_{N}$.
        This case works analogously to the previous cases 
        (except we are considering node constraints instead of edge constraints).

    \item $\Leftarrow \phi_{N}$.
        Again, this case works analogously to the previous case (only in reverse).
\end{enumerate}
\end{proof}

\subsection[Query Mapping]{Query Mapping (\Cref{trfm:prop:qmap})} %
\label{trfm:proof:qmap}

We prove \Cref{trfm:prop:qmap} by defining and proving the following two lemmas.

\begin{lemma}[Invertable Graph Mapping]
   \label{trfm:lemma:inverse}
   Given a PG $G$, then the mapping to an RDF graph can be inverted such that $G = \rdftopg{\pgtordf{G}}$.
\end{lemma}

\begin{proof}
  The mapping algorithm defined in \Cref{trfm:def:graphmap} can be trivially inverted 
  to construct a property graph from an RDF graph
  by matching on triples included in the RDF graph required to construct nodes, edges, labels, and properties.
  The inverted algorithm is partial in the domain of RDF graphs, as there might be triples that are not matched;
  this suffices for \Cref{trfm:lemma:inverse}.
\end{proof}

\begin{lemma} [Soundness of Query Mapping Modulo~\Cref{trfm:lemma:inverse}]
   \label{trfm:lemma:partialsound}
   Given a PG $G$ and a \gcname $q$, 
   then $\pgtordf{\sparqleval{q}{G}} = \sparqleval{\sgcoretoeccq{q}}{\pgtordf{G}}$.
\end{lemma}

\begin{proof}
  We need to show that mapping the result of a \gcname query is equivalent to mapping 
  both the input graph and the query itself.
  We will show the lemma by showing that a \gcname query is equivalent to an 
  IQL (\Cref{trfm:def:iql}) query, given the defined mapping (\Cref{trfm:fig:gciqlboth}), 
  referring to the semantics of \gcname directly, 
  and the semantics of IQL indirectly through the mapping to \eccqname queries (\Cref{trfm:fig:gciqlboth}).

  To this end, we consider the two parts of the queries (referred to here as \emph{matching} 
  and \emph{constructing}) separately.
  We start by proving that the \emph{matching} parts of both query languages produce equivalent results, 
  modulo certain conditions we will define below.
  At the core, the semantics of both languages, as far as the \emph{matching} part of the queries is concerned, 
  differ in that in \gcname we obtain a mapping from variable names to node and edge IDs, 
  whereas in corresponding IQL (and \eccqname) queries we obtain a mapping from (a superset of equivalent) 
  variable names to IRIs, including both the encoded node and edge IDs present in \gcname results, 
  but also all IRIs that correspond to labels and properties, since they are explicitly
  matched by variables, and not implicitly included as for \gcname.

  For the sake of readability, we will refer to, \eg, queries as 
  $q$ (\gcname) and $q'$ (IQL or \eccqname) throughout this proof. 
  Indeed, we apply this notation to \emph{any} objects that can be mapped 
  (such as PG graphs $G$ and RDF graphs $G'$).
  Furthermore, we use IQL and \eccqname somewhat interchangeably; in  particular, 
  we refer to IQL for its syntax and to \eccqname for the semantics of IQL queries, directly.

  \paragraph{Matching}

  We show that the sets of bindings $\Omega$ obtained by a \gcname query $q$ on a 
  PG $G=(N,E,\rho,\lambda,\sigma)$ is equivalent to the set of bindings 
  obtained by the corresponding \eccqname query $q'$ on an RDF graph $G'$.
  The notion of equivalency we use is defined as follows:

  \begin{definition}[Equivalency of Bindings]
    \label{trfm:def:equibind}
    Two sets of valuations $\Omega$ (for \gcname $q$ on $G$) and $\Omega'$ (for \eccqname $q'$ on $G'$)
    are equivalent, written $\Omega \approx \Omega'$, if and only if

    \begin{enumerate}

      \item $\forall \mu \in \Omega, \forall x \in \operatorname{dom}(\mu)$ such that $\mu(x) = n$, then
       \begin{enumerate}
          \item $\exists \mu' \in \Omega' : \mu'(x) = \toiri{n}$,
          \item $\forall l \in L$ if $l \in \lambda(n)$ and $\gcremove l \not\in\gcsr_x$, 
            then $\exists \mu' \in \Omega' : \mu'(c_x) = \toiri{l}$,
          \item $\forall k \in K, \forall v \in V$ if $\sigma(n, k) = \{v\}$ and $\gcremove k \not\in \gcsr_x$, 
            then $\exists \mu' \in \Omega' : \mu'(p_x) = \toiri{k}$ and $\mu'(o_x) = \toiri{v}$.
       \end{enumerate}

      \item $\forall \mu' \in \Omega', \forall x \in \operatorname{dom}(\mu'):$ such that $\mu(x) = \toiri{n}$, 
        then 
       \begin{enumerate}
          \item $\exists \mu \in \Omega : \mu(x) = n$,
          \item $\mu'(c_x) = \toiri{l}$ such that $l \in \lambda(n)$ and $\gcremove l \not\in\gcsr_x$ 
            (or $\mu'(c_x) = \pref{m}{node}$),
          \item $\mu'(p_x) = \toiri{k}$ and $\mu'(o_x) = \toiri{v}$ such that 
            $\sigma(n, k) = \{v\}$ and $\gcremove k \not\in \gcsr_x$ (or $\mu'(p_x) = \pref{m}{node}$),
       \end{enumerate}

      \item $\forall \mu \in \Omega, \forall x, y, z \in \operatorname{dom}(\mu)$ such that 
        $\mu(x) = n_1, \mu(y) = n_2, \mu(z) = e$, then
         \begin{enumerate}
            \item $\exists \mu' \in \Omega' : \mu'(x) = \toiri{n_1}$, $\mu'(y) = \toiri{n_2}$, 
              and $\mu'(z) = \toiri{e}$,
            \item $\forall l \in L$ if $l \in \lambda(e)$ and $\gcremove l \not\in\gcsr_z$, 
              then $\exists \mu' \in \Omega' : \mu'(c_z) = \toiri{l}$,
            \item $\forall k \in K, \forall v \in V$ if $\sigma(e, k) = \{v\}$ and $\gcremove k \not\in \gcsr_z$,
              then $\exists \mu' \in \Omega' : \mu'(p_z) = \toiri{k}$ and $\mu'(o_z) = \toiri{v}$.
         \end{enumerate}

      \item $\forall \mu' \in \Omega', \forall x, y, z \in \operatorname{dom}(\mu'):$ 
        such that $\mu(x) = \toiri{n_1}, \mu(y) = \toiri{n_2}, \mu(z) = \toiri{e}$, then 
        \begin{enumerate}
           \item $\exists \mu \in \Omega$ such that $\mu(x) = n_1, \mu(y) = n_2, \mu(z) = e$,
           \item $\mu'(c_z) = \toiri{l}$ such that $l \in \lambda(e)$ and $\gcremove l \not\in\gcsr_z$ 
             (or $\mu'(c_x) = \pref{m}{node}$),
           \item $\mu'(p_z) = \toiri{k}$ and $\mu'(o_z) = \toiri{v}$ such that $\sigma(e, k) = \{v\}$ 
             and $\gcremove k \not\in \gcsr_z$ (or $\mu'(p_z) = \pref{m}{node}$).
        \end{enumerate}

    \end{enumerate}

    where $c_x$, $p_x$ and $o_x$ are the same fresh variables defined for $x$ in \Cref{trfm:def:eccq:syntax}, 
    and we write, \eg, $n$ and $\toiri{n}$ as implicit conversions between both representations of node 
    (respectively edge) as IDs or IRI.
  \end{definition}

  We will first consider the two patterns $\gcoreNodeF{x_n}{\gcwhere}$ and 
  $\gcoreEdgeF{x_n}{x_n}{x_e}{\gcwhere}$ of \gcname, and then their composition.

  \begin{enumerate}
    \item Case $q_{\gcname} = \gcoreNodeF{x}{\gcwhere}$.
      According to \Cref{trfm:fig:gciqlboth}, this is equivalent to an IQL query pattern of shape 
      $\operatorname{M_N}(x, W_L, W_K, W_V, R_L, R_K)$.
      We need to show equivalency for the bindings $\Omega$ obtained for $\gcoreNodeF{x}{\gcwhere}$ on $G$,
      and $\Omega'$ obtained for the corresponding $\operatorname{M_N}(x, W_L, W_K, W_V, R_L, R_K)$ on $G'$.
      To this end, we show that Case 1 of \Cref{trfm:def:equibind} holds 
      (note, that the precondition of Case 2 is never satisfied, and we thus do not consider this case).
      We show this inductively on the structure of $\gcoreNodeF{x}{\gcwhere}$. 
      We start with the base case where $\gcwhere = \emptyset$, continue with singleton cases, 
      and finally consider union of two sets. 
      \begin{enumerate}
        \item $\gcwhere = \emptyset$.
          Given: $n \in N$ (in $G$) and $\mu(x) = n$.
          Then, by construction of $G'$, $G'$ contains the triple 
          $\triple{\toiri{n}}{\pref{rdf}{type}}{\texttt{node}}$.
          The corresponding \eccqname contains only the triple patterns 
          $\{\atomGenc{x}, \atomGenp{x}, \atomGenN{x}\}$, 
          as well as the filter pattern $\{\atomPNot{x}{\texttt{nte}}\}$.
          We consider the three cases of \Cref{trfm:def:equibind} (1) and (2):
          \begin{itemize}
            \item Case \textbf{1.a} and \textbf{2.a} (\Cref{trfm:def:equibind}).
              By definition of $G'$, $\atomGenN{x}$ is satisfied for $\toiri{n}$;
              Both $\atomGenc{x}$ and $\atomGenp{x}$ are trivially satisfied at least for one $\mu'$ 
              as the graph always contains the triple $\atomGenN{\toiri{n}}$ (and it is only then satisfies).
              Thus, there exists a mapping $\mu' \in \Omega'$ such that $\mu'(x) = \toiri{n}$ exactly in this case.

            \item Case \textbf{1.b} and \textbf{2.b} (\Cref{trfm:def:equibind}).
              For any label $l \in L$ where $l \in \lambda(n)$, by definition $G'$ contains 
              the triple $\triple{\toiri{n}}{\pref{rdf}{type}}{\toiri{l}}$. 
              We first assume that there is no $\gcremove l \in S_x$.
              Then there exists a mapping $\mu' \in \Omega'$ such that 
              $\mu'(x) = \toiri{n}$ and $\mu'(c_x) = \toiri{l}$, since $\atomGenc{x}$ is in the \eccqname.
              (The reverse is true by the same reasoning: If $l \not\in \lambda(n)$, 
              then $G'$ does not contain $\triple{\toiri{n}}{\pref{rdf}{type}}{\toiri{l}}$ by its construction, 
              and the pattern can never be satisfied.)
              If there is a $\gcremove l \in S_x$, then the corresponding \eccqname contains 
              in addition the filter pattern $\atomCNot{x}{\toiri{l}}$.
              According to the semantics of \eccqname (\Cref{trfm:def:eccq:semantics}) 
              we remove exactly the bindings $\mu'$ such that $\mu'(c_x) = \toiri{l}$.
              Thus, such a binding does not exist, as was required.
              (The same reasoning applies for the reverse case.)

            \item Case \textbf{1.c} and \textbf{2.c} (\Cref{trfm:def:equibind}).
              For any property value pair $k \in K, v \in V$ where $\sigma(n, k) = \{v\}$, 
              by definition $G'$ contains the triple $\triple{\toiri{n}}{\toiri{k}}{\toiri{v}}$.
              We first assume that there is no $\gcremove k \in S_x$.
              Then there exists a mapping $\mu' \in \Omega'$ such that $\mu'(x) = \toiri{n}$ 
              and $\mu'(o_{x,k}) = \toiri{v}$, since $\atomGenp{x}$ is in the \eccqname.
              (The reverse is true by the same reasoning: If not $\sigma(n, k) = \{v\}$, 
              then $G'$ does not contain $\triple{\toiri{n}}{\toiri{k}}{\toiri{v}}$ by its construction, 
              and the pattern can never be satisfied.)
              If there is a $\gcremove k \in S_x$, then the corresponding \eccqname contains 
              in addition the filter pattern $\atomPNot{x}{\toiri{k}}$.
              According to the semantics of \eccqname (\Cref{trfm:def:eccq:semantics})
              we remove exactly the bindings $\mu'$ such that $\mu'(p_x) = \toiri{l}$.
              Thus, such a binding does not exist, as was required.
              (The same reasoning applies for the reverse case.)
          \end{itemize}

        \item $\gcwhere = \{\whereLabelRelative{l}\} \cup \gcwhere_r$.
          Given: $n \in N$ (in $G$) and $\mu(x) = n$; equivalency holds for $\gcwhere_r$ (induction hypothesis).
          For any $l \in L$ it holds that if $\{\whereLabelRelative{l}\}$ is satisfied (for $n$), 
          it must be the case that $l \in \lambda(n)$.
          By construction, therefore, $G'$ contains the triple $\triple{\toiri{n}}{\pref{rdf}{type}}{\toiri{l}}$.
          The corresponding \eccqname contains only the additional triple pattern $\{\atomC{x}{\toiri{l}}\}$, 
          which is thus satisfied in $G'$ in at least one $\mu_1'$.
          Since we know that equivalency holds regarding $\gcwhere_r$ (by the induction hypothesis), 
          we know that there exists a $\mu_2'$ for this query where $\mu_2'(x) = \toiri{n}$, 
          and that therefore, $\mu_1' \sim \mu_2'$.
          With the same reasoning as for the previous case, 
          there then exists a mapping $\mu' = \mu_1' \cup \mu_2' \in \Omega'$ where $\mu'(x) = \toiri{n}$.
          For any $l \in L$ where $\{\whereLabelRelative{l}\}$ is not satisfied (for $n$), 
          also $l \not\in \lambda(n)$.
          Then, by construction, $G'$ does not contain the triple
          $\triple{\toiri{n}}{\pref{rdf}{type}}{\toiri{l}}$;
          the inverse of the previous case applies, 
          since $\{\atomC{x}{\toiri{l}}\}$ cannot be satisfied under these conditions.
          Therefore, there does not exist the $\mu'$ in $\Omega'$.

        \item $\gcwhere = \{\wherePropertyEqualsRelative{k}{v}\} \cup \gcwhere_r$.
          Given: $n \in N$ (in $G$) and $\mu(x) = n$; equivalency holds for $\gcwhere_r$ (induction hypothesis).
          For any $k \in K$ and $v \in V$ it holds that if $\{\wherePropertyEqualsRelative{k}{v}\}$ 
          is satisfied (for $n$), it must be the case that $\sigma(n, k) = \{v\}$.
          By construction, therefore, $G'$ contains the triple $\triple{\toiri{n}}{\toiri{k}}{\toiri{v}}$.
          The corresponding \eccqname contains only the additional triple pattern 
          $\atomP{x}{\toiri{v}}{\toiri{k}}$, which is thus satisfied in $G'$ in at least one $\mu_1'$
          where $\mu_1'(x) = \toiri{n}$.
          Since we know that equivalency holds regarding $\gcwhere_r$ (by the induction hypothesis), 
          we know that there exists a $\mu_2'$ for this query where $\mu_2'(x) = \toiri{n}$, 
          and that therefore, $\mu_1' \sim \mu_2'$.
          With the same reasoning as for the previous case, 
          there then exists a mapping $\mu' = \mu_1' \cup \mu_2' \in \Omega'$ such that $\mu'(x) = \toiri{n}$.
          The inverse case works analogously to the second case.

        \item $\gcwhere = \{\wherePropertyRelative{k}\} \cup \gcwhere_r$.
          Given: $n \in N$ (in $G$) and $\mu(x) = n$; equivalency holds for $\gcwhere_r$ (induction hypothesis).
          For any $k \in K$ it holds that if $\{\wherePropertyRelative{k}\}$ is satisfied (for $n$), 
          it must be the case that $\sigma(n, k) = \{v\}$ for some $v \in V$.
          By construction, therefore, $G'$ contains a triple $\triple{\toiri{n}}{\toiri{k}}{\toiri{v}}$.
          The corresponding \eccqname contains only the additional triple pattern $\atomP{x}{o_{x,k}}{\toiri{k}}$, 
          which is thus satisfied in $G'$ in at least one $\mu_1'$, 
          where $\mu'(x) = \toiri{n}$ and $\mu'(o_{x,k}) = \toiri{v}$.
          Since we know that equivalency holds regarding $\gcwhere_r$ (by the induction hypothesis), 
          we know that there exists a $\mu_2'$ for this query where $\mu_2'(x) = \toiri{n}$.
          Furthermore, we know that $o_{x,k} \not\in \operatorname{dom}(\mu_2')$ 
          since the pattern is unique by definition.
          Therefore, $\mu_1' \sim \mu_2'$.
          With the same reasoning as for the previous case, 
          there then exists a mapping $\mu' = \mu_1' \cup \mu_2' \in \Omega'$ where $\mu'(x) = \toiri{n}$.
          The inverse case works analogously to the second case.
      \end{enumerate}

    \item Case $q_{\gcname} = \gcoreEdgeF{x}{y}{z}{\gcwhere}$.
      According to \Cref{trfm:fig:gciqlboth}, this is equivalent to an IQL query pattern of shape 
      $\operatorname{M_E}(x, y, z, W_L, W_K, W_V, R_L, R_K)$.
      We need to show equivalency of the set of bindings $\Omega$ obtained for 
      $\gcoreEdgeF{x}{y}{z}{\gcwhere}$ on $G$, and $\Omega'$ obtained for the corresponding
      $\operatorname{M_E}(x, y, z, W_L, W_K, W_V, R_L, R_K)$.
      To this end, we show that Case 2 of \Cref{trfm:def:equibind} holds.
      (Again, the premise of the other case is not satisfied).
      We show this case inductively on the structure of $\gcoreEdgeF{x}{y}{z}{\gcwhere}$. 
      We start with the base case where $\gcwhere = \emptyset$, continue with singleton cases, 
      and finally consider union of two sets. 
        \begin{enumerate}
          \item $\gcwhere = \emptyset$.
            Given: $n_1, n_2 \in N$, $e \in E$ (in $G$), $\rho(e) = (n_1, n_2)$ 
            and $\mu(z) = e$, $\mu(x) = n_1$, and $\mu(y) = n_2$.
            Then, by construction of $G'$, $G'$ contains and the triples
            $\triple{\toiri{e}}{\pref{rdf}{type}}{\texttt{edge}}$,
            $\triple{\toiri{n_1}}{\pref{rdf}{type}}{\texttt{node}}$,
            $\triple{\toiri{n_2}}{\pref{rdf}{type}}{\texttt{node}}$,
            $\triple{\toiri{n_1}}{\texttt{nte}}{\toiri{e}}$, and
            $\triple{\toiri{e}}{\texttt{etn}}{\toiri{n_2}}$.
            The corresponding \eccqname contains only the triple patterns 
            $\{\atomPin{x}{z}, \atomPout{z}{y},\allowbreak\atomGenc{z}, \atomGenp{z}, \atomGenE{z}\}$, 
            as well as the filter pattern $\{\atomPNot{z}{\texttt{etn}}\}$.

            We consider the three cases of \Cref{trfm:def:equibind} (3) and (4):
            \begin{itemize}

              \item Case \textbf{3.a} and \textbf{4.a} (\Cref{trfm:def:equibind}).
                By definition of $G'$, $\atomGenE{z}$ is satisfied for $\toiri{z}$;
                Both $\atomGenc{z}$ and $\atomGenp{z}$ are trivially satisfied at least for one $\mu'$ 
                as the graph always contains the triple $\atomGenN{\toiri{e}}$.
                Similarly, both $\atomPin{x}{z}$ and $\atomPout{z}{y}$ are satisfied by the definition of $G'$
                (see above) for at least one $\mu'$.
                This $\mu'$ satisfies $\mu'(x) = \toiri{n_1}$, $\mu'(y) = \toiri{n_2}$ and $\mu'(z) = \toiri{e}$,
                given the triples included in $G'$ (by definition) shown above.

              \item Case \textbf{3.b} and \textbf{4.b} (\Cref{trfm:def:equibind}).
                This case is equivalent to Case 1.b for 1.

              \item Case \textbf{3.c} and \textbf{4.c} (\Cref{trfm:def:equibind}).
                This case is equivalent to Case 1.c for 1.
            \end{itemize}
          \item $\gcwhere = \{\whereLabelRelative{l}\} \cup \gcwhere_r$. Analogous to 1.a.ii.
          \item $\gcwhere = \{\wherePropertyEqualsRelative{k}{v}\} \cup \gcwhere_r$. Analogous to 1.a.iii.
          \item $\gcwhere = \{\wherePropertyRelative{k}\} \cup \gcwhere_r$. Analogous to 1.a.iv.
        \end{enumerate}

    \item Case $q_{\gcname} = \{\gcpattern_1 , \gcpattern_2 \}$. 
      For this case, we need to show that the equivalency of two sets of bindings holds under the join operation.
      That is, when $\Omega_1 \approx \Omega'_1$ and $\Omega_2 \approx \Omega'_2$ 
      then $(\Omega_1 \bowtie \Omega_2) \approx (\Omega'_1 \bowtie \Omega'_2)$.
      This follows directly from the definition of $\bowtie$: 
      By the definition of equivalency, both $\Omega$ and $\Omega'$ 
      must share corresponding bindings for common variables,
      while $\Omega'$ might have additional variables.
      Thus, two mappings $\mu_1$ and $\mu_2$ are compatible if and only if $\mu_1'$ and $\mu_2'$ are compatible.

  \end{enumerate}

  \paragraph{Constructing}

  We show that the graph $G_o = (N_o, E_o, \rho_o, \lambda_o, \sigma_o)$ 
  constructed by a \gcname query $q$ on a PG $G=(N,E,\rho,\lambda,\sigma)$ is equivalent to the graph $G'_o$
  obtained by the corresponding \eccqname query $q'$ on an RDF graph $G'$.
  To this end, we already showed that the set of bindings $\Omega$ obtained by the \emph{matching} 
  part of $q$ is equivalent (w.r.t. \Cref{trfm:def:equibind}) 
  to the set of bindings $\Omega'$ obtained by the \emph{matching} part of $q'$ (\ie $\Omega \approx \Omega'$).
  In the following, w.l.o.g., we assume that this is the case;
  similarly, we assume that any fresh IRIs and node or edge IDs generated by either \eccqname
  or \gcname queries, respectively, are the same modulo the graph mapping relation.
  To this end, we refer to them as $a_{\mu(x)}$ or $a_{\mu'(x)}$ where $\mu$ and $\mu'$ are 
  \gcname and \eccqname mappings from equivalent bindings $\Omega$ and $\Omega'$:
  Thus, $\toiri{a}_{\mu(x)} = a_{\mu'(x)}$ holds.

  We will first consider the two patterns $\gcoreNodeF{x_n}{\gcwhere}$ 
  and $\gcoreEdgeF{x_n}{x_n}{x_e}{\gcwhere}$ of \gcname, and then their composition.

  \begin{enumerate}
    \item Case $q_{\gcname} = \gcoreNodeS{x}{\gcsr}$. 
      We prove this construction by inductions on the structure of $\gcsr$.
      Without loss of generality, we assume that $\gcsr$ never contains both 
      $\gcset l$ and $\gcremove l$ (or $\gcset k = v$ and $\gcremove k$);
      indeed, these cases behave like the removal cases.
      For each case, we show that both constructing an RDF output graph $G_o'$ 
      and mapping $G_o$ to an RDF graph $G_o''$ produces the same set of triples.
      \begin{enumerate}
        \item $\gcsr = \emptyset$.
          Here we construct the graph $G_o$ consisting of nodes $n = \mu(x)$ for all bindings $\mu \in \Omega$,
          which is defined as 
          $\cup_{\mu \in \Omega} \{\{\mu(x)\}, \emptyset, \emptyset, \lambda_{\mu(x)}, \sigma_{\mu(x)}\}$
          (\Cref{trfm:def:g:semantics}).
          The query maps to $\operatorname{C_N}(x, A_L, A_K, V)$,
          the semantics of which are defined by the \eccqname construction patterns 
          $\{\atomGenc{x}, \atomGenp{x}, \atomGenN{x}\}$ if $x \in V$ or $\{\atomGenN{x}\}$ otherwise,
          as a union over all $\mu' \in \Omega'$ with $\Omega' \approx \Omega$.
          If $x \not\in V$, then $\{\atomGenN{x}\}$ constructs $a_{\mu'(x)}$, for which a corresponding 
          $a_{\mu(x)}$ is constructed, given the assumption made above.
          Otherwise, we know for each $\mu(x) = n$ in $\Omega$ there exists a $\mu' \in \Omega'$ 
          where $\mu'(x) = \toiri{n}$; 
          the pattern $\atomGenN{x}$ constructs therefore the triple $\atomGenN{x}$ in $G_o'$.
          Mapping $G_o$ to $G_o''$ also includes the triple $\atomGenN{x}$ 
          by definition of the mapping (\Cref{trfm:def:graphmap}).
          The remainder is the same for either case (where we use $n$ for either $n = a_{\mu(x)}$
          or $n = \mu(x)$).
          Secondly, we know that for each $l \in \lambda(n)$ there exists a $\mu' \in \Omega'$ 
          where $\mu'(c_x) = \toiri{l}$.
          The pattern $\atomGenc{x}$ constructs the respective triple 
          $\triple{\toiri{n}}{\pref{rdf}{type}}{\toiri{l}}$ in the output graph $G_o'$.
          Mapping $G_o$ to $G_o''$ also includes this triple, 
          since $\lambda_{\mu(x)}$ above includes $l$, which is mapped to exactly the required triple.
          Finally, we know that for each $k \in K, v \in \sigma(n, k)$ there exists a $\mu' \in \Omega'$ 
          where $\mu'(p_x) = \toiri{k}$ and $\mu'(o_x) = \toiri{v}$;
          thus, pattern $\atomGenp{x}$ constructs the triple 
          $\triple{\toiri{n}}{\toiri{k}}{\toiri{v}}$ for this $\mu'$ in $G_o'$.
          Mapping $G_o$ to $G_o''$ also includes this triple, since $\sigma{\mu(x), k}$ is $v$, 
          which is mapped to the required triple.
          No other triples are included in either construction.
          Therefore, $G_o'$ is equal to $G_o''$.
          
        \item $\gcsr = \{\gcset l\} \cup \gcsr_r$.
          We assume, by the induction hypothesis, that the equivalency of the output graphs 
          $G_o'$ and $G_o''$ holds for $\gcsr_r$.
          The additional set clause $\gcset l$ modifies the graph $G_o$ by adding $l$ to $\lambda(n)$ 
          for all $n \in N$ such that $\mu(x) = n$, for some $\mu \in \Omega$ (which may also be a fresh 
          identifier of the form $a_{\mu(x)}$).
          The equivalent \eccqname $q'$ includes an additional pattern of the form $\atomC{x}{\toiri{l}}$ 
          (and equivalent bindings for $x$).
          By the semantics of \eccqname we know that $G_o'$ therefore includes 
          $\triple{\toiri{n}}{\pref{rdf}{type}}{\toiri{l}}$ for all $\mu'(x) = \toiri{n}$.
          By the definition of our graph mapping, we know that $G_o''$ must include this triple, too, 
          since $l \in \lambda(n)$ for the respective $n$.
          Thus, $G_o' = G_o''$.

        \item $\gcsr = \{\gcremove l\} \cup \gcsr_r$.
          We assume, by the induction hypothesis, that the equivalency of the output graphs 
          $G_o'$ and $G_o''$ holds for $\gcsr_r$.
          The additional remove clause $\gcremove l$ modifies the graph $G_o$ by exclusion of 
          $l$ in $\lambda(n)$ for all $n \in N$ where $\mu(x) = n$, for some $\mu \in \Omega$.
          If such a remove clause exists, $A_L$ does not include the pattern $\gcset l$, 
          by definition, for these $x$.
          Thus, the respective triples are never constructed (see also the base case),
          and $G_o'$ does not include $\triple{\toiri{n}}{\pref{rdf}{type}}{\toiri{l}}$.
          Similarly, $G_o''$ does not include this triple, since $l$ is not in $\lambda(n)$.
          Therefore, $G_o' = G_o''$.

        \item $\gcsr = \{\gcset k = v\} \cup \gcsr_r$.
          This case is equivalent to the case with $\gcset l$.

        \item $\gcsr = \{\gcremove k\} \cup \gcsr_r$.
          This case is equivalent to the case with $\gcremove l$.
      \end{enumerate}

    \item Case $q_{\gcname} = \gcoreEdgeS{x}{y}{z}{\gcsr}$.
      We can again prove this construction by inductions on the structure of $\gcsr$.
      Without loss of generality, we assume that $\gcsr$ never contains both $\gcset l$ 
      and $\gcremove l$ (or $\gcset k = v$ and $\gcremove k$);
      indeed, these cases behave like the removal cases.
      For each case, we show that both constructing an RDF output graph $G_o'$ and mapping $G_o$ 
      to an RDF graph $G_o''$ produces the same set of triples.
      \begin{enumerate}
        \item $\gcsr = \emptyset$.
          Here we construct the graph $G_o$ consisting of nodes $n_1 = \mu(x)$, $n_2 = \mu(y)$, 
          and $e = \mu(z)$ for all bindings $\mu \in \Omega$, defined as
          $\cup_{\mu \in \Omega} \{\{v, u\}, \{e\}, \{e \shortmaparrow (v,u)\}, \lambda_{\mu(x)}, \sigma_{\mu(x)}\}$ (\Cref{trfm:def:g:semantics}).
          The query maps to $\operatorname{C_E}(x, y, z, A_L, A_K, V)$,
          the semantics of which are defined by the \eccqname construction patterns of the form
          $\{\atomPin{x}{z}, \allowbreak\atomPout{z}{y}, \allowbreak\atomGenc{z}, \allowbreak\atomGenp{z}, \allowbreak\atomGenE{z}\}$  if $z \in V$ or $\{\atomPin{x}{z}, \allowbreak\atomPout{z}{y}, \allowbreak\atomGenE{z}\}$ otherwise,
          as a union over all $\mu' \in \Omega'$ with $\Omega' \approx \Omega$.
          If $z \not\in V$, then $\{\atomPin{x}{z}, \allowbreak\atomPout{z}{y}, \allowbreak\atomGenE{z}\}$ 
          constructs $a_{\mu'(z)}$, for which a corresponding 
          $a_{\mu(z)}$ is constructed, given the assumption made above.
          If $z \in V$, then $\mu(z) = e$ in $\Omega$ and there exists a $\mu' \in \Omega'$ 
          where $\mu'(z) = \toiri{e}$.
          The remainder is the same for either case (where we use $e$ for either $e = a_{\mu(z)}$
          or $e = \mu(z)$).
          First, we know for each $\mu(x) = n_1$ and $\mu(y) = n_2$ in $\Omega$ 
          there exists a $\mu' \in \Omega'$ 
          where $\mu'(x) = \toiri{n_1}$ and $\mu'(y) = \toiri{n_2}$.
          The triples 
            $\triple{\toiri{n_1}}{\texttt{nte}}{\toiri{e}}$, and
            $\triple{\toiri{e}}{\texttt{etn}}{\toiri{n_2}}$ are therefore triples in $G_o'$.
          Mapping $G_o$ to $G_o''$ also includes $\atomPin{x}{z}$ and $\atomPout{z}{y}$ by 
          definition of the mapping (\Cref{trfm:def:graphmap}),
          since $\rho(e) = (n_1, n_2)$ in $G_o$.
          Similarly, $\triple{\toiri{e}}{\pref{rdf}{type}}{\texttt{edge}}$ is in $G_o'$ 
          since $\atomGenE{z}$ is in $q'$.
          The triple is also in $G_o''$ by its construction, since $e \in E_o$ (edges of $G_o$).
          The remainder of this case, namely the inclusion of labels $l \in L$ and properties 
          $k \in K$ is exactly analogous to case 1.a),
          with only very minor changes (\ie $e$ instead of $n$).
          No other triples are included in either construction.
          Therefore, $G_o'$ is equal to $G_o''$.

        \item \emph{Remaining Cases.}
          Since there are not further differences between edges and nodes not already covered by the first case,
          the remaining cases are exactly analogous (modulo minor notational differences) 
          to the node cases in Case 1 of this proof.
      \end{enumerate}

    \item Case $q_{\gcname} = \{\gctemplate_1 , \gctemplate_2 \}$. 
      Finally, we show that when $G_{o,1}' = G_{o,1}''$ and $G_{o,2}' = G_{o,2}''$, then also $G_o' = G_o''$.
      $G_o'$ is simply the union of triples in both $G_{o,1}'$ and $G_{o,2}'$ (\Cref{trfm:def:eccq:semantics}).
      $G_o''$ is the graph resulting from mapping $G_o$ to RDF.
      $G_o$ is the (graph) union of the graphs $G_{o,1}$ and $G_{o,2}$ defined $(N_1 \cup N_2, E_1 \cup E_2, \rho, \lambda, \sigma)$ 
      where $\forall e \in E_1 \cup E_2 : \rho(e) = \rho_1(e)\ \textrm{if}\ e \in E_1\ \textrm{else}\ \rho_2(e)$,
      $\forall x \in N_1 \cup N_2 \cup E_1 \cup E_2, k \in K: \lambda(x) = \lambda_1(x) \cup \lambda_2(x) $, 
      and $\sigma(x, k) = \sigma_1(x,k) \cup \sigma_2(x,k)$.
      We consider the individual components of $G_o$.
      \begin{enumerate}
        \item $N_1 \cup N_2$.
          Then $G_o''$ contains $\triple{\toiri{n}}{\pref{rdf}{type}}{\texttt{node}}$ for $n \in N_1 \cup N_2$.
          Given the equivalency of the graphs $G_{o,1}' = G_{o,1}''$ and $G_{o,2}' = G_{o,2}''$,
          $G_o'$ clearly also contains these triples.

        \item $E_1 \cup E_2$.
          Then $G_o''$ contains $\triple{\toiri{e}}{\pref{rdf}{type}}{\texttt{edge}}$ for $e \in E_1 \cup E_2$.
          Given the equivalency of the graphs $G_{o,1}' = G_{o,1}''$ and $G_{o,2}' = G_{o,2}''$,
          $G_o'$ clearly also contains these triples.

        \item $\rho$.
          Then $G_o''$ contains $\triple{\toiri{n_1}}{\texttt{nte}}{\toiri{e}}$, 
          and $\triple{\toiri{e}}{\texttt{etn}}{\toiri{n_2}}$ for $e \in E_1 \cup E_2$ (and the respective $n_1, n_2$).
          Given the equivalency of the graphs $G_{o,1}' = G_{o,1}''$ and $G_{o,2}' = G_{o,2}''$,
          $G_o'$ clearly also contains these triples.

        \item $\lambda$.
          Then $G_o''$ contains $\triple{\toiri{u}}{\pref{rdf}{type}}{\toiri{l}}$ 
          for $u \in N_1 \cup E_1 \cup N_2 \cup E_2$.
          Given the equivalency of the graphs $G_{o,1}' = G_{o,1}''$ and $G_{o,2}' = G_{o,2}''$,
          $G_o'$ clearly also contains these triples.

        \item $\sigma$.
          Then $G_o''$ contains $\triple{\toiri{u}}{\toiri{k}}{\toiri{v}}$ for $u \in N_1 \cup E_1 \cup N_2 \cup E_2$.
          Given the equivalency of the graphs $G_{o,1}' = G_{o,1}''$ and $G_{o,2}' = G_{o,2}''$,
          $G_o'$ clearly also contains these triples.
      \end{enumerate}
  \end{enumerate}
\end{proof}

Finally, the proof for \Cref{trfm:prop:qmap} follows directly from both lemmas.

\begin{proof}
  The proof follows directly from \Cref{trfm:lemma:inverse} and \Cref{trfm:lemma:partialsound}.
\end{proof}

\subsection[Shape Entailment]{Shape Entailment (\Cref{trfm:prop:cwa})} %
\label{trfm:proof:infer}

Following the proofs in~\cite{DBLP:conf/www/Seifer0LS24}, 
it suffices to show that the axioms in $\Sigma$ are indeed valid on all extended graphs.
We can adapt the respective proposition from~\cite{DBLP:conf/www/Seifer0LS24} as given in \Cref{trfm:prop:cwamod}.
Note that \Cref{trfm:prop:smap} and \Cref{trfm:prop:qmap} guarantee that this proposition suffices for 
showing the soundness of our method for property graphs, \simpleprogs and \gcname.

\begin{proposition}[Modification of \Cref{trfm:prop:cwa}]
   \label{trfm:prop:cwamod}
   For every extended graph $\graphext$ of a \eccqname $q$ and a set of $\DLogics$ shapes $\shapesin$
   where $\shaclvalid{\graphin}{\shapesin}$,
   $\shaclvalid{\graphext}{\operatorname{infer}(q, S)}$.
\end{proposition}

We show for each of the seven steps of the inference algorithm separately that \Cref{trfm:prop:cwamod} holds,
defining \Cref{trfm:lem:cwaone} through \Cref{trfm:lem:cwaeleven}.
We refer to the set of axioms inferred in each of the seven steps of
\Cref{trfm:def:cwa} as $\Sigma_{1}$ through $\Sigma_{7}$ as a shorthand.
For each proof, we have the following preliminary assumptions:

\begin{enumerate}
  \item A set of \simpleprogs shapes $S'$, which can be mapped via $\sprogstoshacl{S'}$ 
  to $\DLogics$ axioms (representing SHACL shapes) $S$; these are the input shapes to the problem.
  Any input graph $G'$ is assumed to be valid regarding $S'$.
  Thus, \Cref{trfm:prop:smap} ensures that for any of these input graphs $G'$ also $G = \pgtordf{G'}$ 
  is valid regarding $S = \sprogstoshacl{S'}$.

  \item A \gcname query $q'$, which can be mapped to a \eccqname $q = \eccqformal = \sgcoretoeccq{q'}$; 
    this is the input query of the problem.

  \item Finally, we also build on Proposition $3_{\sts}$ as was defined 
    and proven in~\cite{DBLP:conf/www/Seifer0LS24}.
    We repeat this proposition below (\Cref{trfm:prop:sts}) and adapt it to our notation, 
    but refer to~\cite{DBLP:conf/www/Seifer0LS24} for its proof, which applies here as well.
    NB: The proof is simple and relies on the property that extended graphs are defined as
    the union of four different graphs that do not share a vocabulary for 
    concept names or role names.
\end{enumerate}

\begin{definition}[Modification of Proposition $3_{\sts}$~\cite{DBLP:conf/www/Seifer0LS24}]\label{trfm:prop:sts}
  Given an RDF graph $\graphin$ and an \eccqname $q$,
  let the graphs $\graphmed$, $\graphout$, and $\graphext$ be defined according to \Cref{trfm:def:extended-graph}.
  For every $\DLogics$ axiom $\varphi$ that does not include names with dots,
  the following equivalences hold:
  \begin{enumerate}
  \item
    $\shaclvalid{\graphin}{\{\varphi\}}$, if and only if $\shaclvalid{\graphext}{\{\varphi\}}$.
  \item
    $\shaclvalid{\graphmed}{\{\varphi\}}$, if and only if $\shaclvalid{\graphext}{\{\dot{\varphi}\}}$.
  \item
    $\shaclvalid{\graphout}{\{\varphi\}}$, if and only if $\shaclvalid{\graphext}{\{\ddot{\varphi}\}}$.
  \end{enumerate}
\end{definition}

In the remainder of this section, we prove the individual parts of \Cref{trfm:def:cwa}.
Note that some proofs are directly based on, and others in part inspired by, \cite{DBLP:conf/www/Seifer0LS24}
(or, more precisely, the extended version of that paper~\cite{seifer2024shapes}).
We first start with the auxiliary \Cref{trfm:lem:base}, which shows that generic patterns, 
that is atomic patterns of the form $\atomGenc{x}$ and $\atomGenp{x}$ 
do not restrict variables under certain conditions.

\begin{lemma}
   \label{trfm:lem:base}
   For a \eccqname $q$ with variable $x$, inclusion of the pattern $\atomGenc{x}$ can never change the 
   set of bindings for variable $x$, regardless of the input graph, as long as the query pattern also includes 
   at least one pattern of the form $\atomC{x}{A}$ for any concept name $A$.
   Note that this is always true for queries utilizing our mapping.
   The equivalent holds for $\atomGenp{x}$ and $\atomP{x}{y}{p}$ or $\atomP{x}{a}{p}$.
\end{lemma}

\begin{proof}
  This trivially follows from the definition of \eccqname syntax and semantics.
  First note that for patterns of the form $\atomGenc{x}$ (and $\atomGenp{x}$, respectively),
  the variable $\cvar{x}$ (and $\rvar{x}, \ivar{x}$, respectively) 
  cannot occur again within the query, by definition.
  Consider two query patterns $P_{1} = \atomC{x}{A}$ and $P_{2} = \atomGenc{x}$.
  Let there be an individual name $a$ and graph $G$.
  We can consider the following cases and what that means for the intersection of both patterns.
  \begin{enumerate}
    \item $a \in \mu_{P_{1}}(x)$ and $a \in \mu_{P_{2}}(x)$.
          Trivially, the intersection includes $a$.
    \item $a \notin \mu_{P_{1}}(x)$ and $a \notin \mu_{P_{2}}(x)$.
          Trivially, the intersection does not include $a$, but neither does pattern $\atomC{x}{A}$.
    \item $a \notin \mu_{P_{1}}(x)$, but $a \in \mu_{P_{2}}(x)$.
          Trivially, the intersection does not include $a$, but neither does pattern $\atomC{x}{A}$.
    \item $a \in \mu_{P_{1}}(x)$, but $a \notin \mu_{P_{2}}(x)$.
          This is the interesting case where $\atomGenc{x}$ could possibly
          restrict the bindings for $\atomC{x}{A}$ in the intersection of both patterns.
          However, this case cannot occur.
          Assume that there was such an $a$ in $G$.
          Then this means that $G$ includes $\atomC{a}{A}$.
          However, since $\cvar{x}$ is unrestricted, $\atomGenc{x}$ clearly matches this triple.
          This contradicts the initial assumption; thus, this case cannot exist.
  \end{enumerate}

  Thus, the intersection of $\atomC{x}{A}$ and $\atomGenc{x}$ is the same
  as $\atomC{x}{A}$ in all cases.
  The same can be shown for the $\atomGenp{x}$ case.
\end{proof}

We next consider the lemmas associated with the individual steps of the algorithm.

\begin{lemma}[Pert. to \Cref{trfm:def:cwa}.1]
   \label{trfm:lem:cwaone}
   For every $\graphext$ of a \eccqname $q$, it holds that $\shaclvalid{\graphext}{\Sigma_{1}}$.
\end{lemma}

\begin{proof}[Proof. Adapted from~\cite{DBLP:conf/www/Seifer0LS24}] 
  This follows directly from Proposition $3_{\sts}$. 
  All input graphs are valid regarding $S$ and $\Sigma_{1} = S$, by definition.
\end{proof}

\begin{lemma}[Pert. to \Cref{trfm:def:cwa}.2]
   \label{trfm:lem:cwatwo}
   For every $\graphext$ of a \eccqname $q$, it holds that $\shaclvalid{\graphext}{\Sigma_{2}}$.
\end{lemma}

\begin{proof}[Proof. Inspired by~\cite{DBLP:conf/www/Seifer0LS24}] 
  We consider both directions of this case separately.
  We begin by proving the conditional case by contradiction.
  Let $\vcg(\eccqpattern)$ be acyclic,
  $G$ be a graph and let $x$ be a variable occurring in $\sccqpattern$,
  and let $C$ be a concept defined as
  \begin{align*}
    C \equiv \textstyle
    & {\bigsqcap_{\atomC{x}{A} \in P}} A
    \sqcap {\bigsqcap_{\atomP{x}{u}{p} \in P}} \exists p.\Concept{u}
    \sqcap {\bigsqcap_{\atomP{u}{x}{p} \in P}} \exists p^-.\Concept{u}.
  \end{align*}
  Assume that there is an individual name $c$ in $G$ such that $\atomC{c}{C}$
  is valid in $\graphext$, but $\atomC{c}{\vconcept{x}}$ is \emph{not} valid in $\graphext$.
  We will show that this leads to a contradiction.

  Assume, w.l.o.g., that $\sccqpattern$ includes a single concept assertion $\atomC{x}{A}$
  that includes $x$, and no role assertions of the form $\atomP{x}{b}{r}$ ($b$ is an individual name).
  If there were other atoms of these forms, we could define $A$ as an auxiliary atom
  that is equivalent to their intersection instead.
  Further assume, w.l.o.g., that $\sccqpattern$ does not include generic patterns of the form $\atomGenc{x}$ or $\atomGenp{x}$.
  Patterns of these forms can never restrict the bindings of $x$ (see \Cref{trfm:lem:base}), 
  and if \emph{only} generic patterns occur for a variable $x$ in a query, 
  $\vconcept{x}$ is not defined in this case anyway.
  Finally, assume -- also w.l.o.g. -- that $\graphext$ includes
  \[\{
    \atomC{a}{V_y}, \atomP{a}{c}{r},
    \atomC{c}{A}, \atomP{c}{b}{s},
    \atomC{b}{V_z}
  \}\]

  Let $\Omega$ be the set of all mappings $\mu$ such that $\mu(P) \subseteq G$.
  Then, by definition there exist the mappings $\mu_1, \mu_2 \in \Omega$ 
  such that $\mu_1(y) = a$ and $\mu_2(z) = b$,
  but there does not exist a mapping $\mu\in\Omega$ such that $\mu(x) = c$ 
  (given the assumption about $x$ given above).
  Then also $\mu_1(x) \neq c$ and $\mu_2(x) \neq c$.
  Let $P_y$ be the part of pattern $\eccqpattern$ that \emph{connects} with variables $y$ and $x$, but not $z$.
  Let $P_z$ be the part of pattern $\eccqpattern$ that \emph{connects} with variables $z$ and $x$, but not $y$.
  Finally, let
  \begin{align*}
    \mu_1^y &= \mu_1\big\vert_{\var(P_y)\setminus\{x\}}, & \mu_1^z &= \mu_1\big\vert_{\var(P_z)\setminus\{x\}},\\
    \mu_2^y &= \mu_2\big\vert_{\var(P_y)\setminus\{x\}}, & \mu_2^y &= \mu_2\big\vert_{\var(P_z)\setminus\{x\}}.
  \end{align*}
  Then, $\mu_1 = \mu_1^y \cup \{x \mapsto \mu_1(x)\} \cup \mu_1^z$ and $\mu_2 = \mu_2^y \cup \{x \mapsto \mu_2(x)\} \cup \mu_2^z$.
  Since $\mu_1^y$ and $\mu_2^z$ share no variables, 
  $\mu_3 = \mu_1^y \cup \{x \mapsto c\} \cup \mu_2^z$ is a mapping.
  By the definition of the semantics of \eccqname $\mu_3 \in \Omega$.
  Then $\atomC{c}{\vconcept{x}}$ is valid in $\graphext$.
  This contradicts the initial assumptions, from which we conclude $\vconcept{x} \sqsupseteq C$.

  We next consider the other direction.
  Assume $a$ to be an arbitrary individual name such that $\{a\} \sqsubseteq \vconcept{x}$ 
  is valid in the extended graph $\graphext$.
  For each operand $k$ in the intersection of \Cref{trfm:def:cwa}.2 we show separately that
  $\{a\} \sqsubseteq k$, assuming that the respective component is defined.
  We assume in all three cases, w.l.o.g., that $\sccqpattern$ does not include 
  generic patterns of the form $\atomGenc{x}$ or $\atomGenp{x}$,
  with the same reasoning as before (\Cref{trfm:lem:base}).
  \begin{enumerate}
    \item For $k = {\bigsqcap_{\atomC{x}{A} \in \sccqpattern}} A$:  
      If $\{a\} \sqsubseteq \vconcept{x}$, then by definition $a$ is an instance of variable $x$ 
      in $\sccqpattern$, \ie $a \in \mu(x)$.
      Then for each concept name $A$ occurring in an atomic pattern of the form 
      $\atomC{x}{A} \in \sccqpattern$ there must be $\atomC{a}{A} \in \graphin$ 
      (since otherwise $a \not\in \mu(x)$), so also $\atomC{a}{A} \in \graphext$ for each such $A$.
      Therefore, $\{a\} \sqsubseteq k$.
      
    \item For $k = {\bigsqcap_{\atomP{x}{u}{p} \in\sccqpattern}} \exists p.\Concept{u}$:
      If $\{a\} \sqsubseteq \vconcept{x}$, then by definition $a$ is an instance of variable $x$ 
      in $\sccqpattern$, \ie $a \in \mu(x)$.
      Then for each property name $p$ occurring in an atomic pattern of the form 
      $\atomP{x}{u}{p} \in \sccqpattern$, one of two cases applies:
      if $u$ is an individual name, then there must be $\atomP{a}{u}{p} \in \graphin$, 
      so also $\atomP{a}{u}{p} \in \graphext$ for such $p$.
      If $u$ is a variable name, then there must be $\atomP{a}{b}{p} \in \graphin$, 
      so also $\atomP{a}{b}{p} \in \graphext$, 
      and also $b \in \mu(u)$ (since otherwise $a \not\in \mu(x)$).
      Therefore, $\{a\} \sqsubseteq k$.
    
    \item For $k = {\bigsqcap_{\atomP{u}{x}{p} \in\sccqpattern}} \exists p^-.\Concept{u}$: 
      Analogous to the previous case.
  \end{enumerate}
  If one or more components $k$ are defined, then
  \[ \{a\} \sqsubseteq {\bigsqcap_{x:A \in\sccqpattern}} A \sqcap {\bigsqcap_{(x,u):p \in\sccqpattern}} \exists p.\Concept{u} \sqcap {\bigsqcap_{(u,x):p \in\sccqpattern}} \exists p^-.\Concept{u}.\]
  At least one component $k$ must be defined, since otherwise the concept $\vconcept{x}$ would not be defined,
  as there must exist either $\atomC{x}{A} \in\sccqpattern$ for some concept name $A$
  or $\atomP{x}{u}{p} \in\sccqpattern$ (or $\atomP{x}{u}{p} \in\sccqpattern$ respectively) 
  for some property $p$, if $x \in \var(P)$.
  Note that this also covers the cases where only generic patterns (\eg, $\atomGenc{x}$) 
  exist for a variable $x$; here,
  the concept $\vconcept{x}$ is not defined.
  (Similarly, the case where $x \not\in \sccqpattern$ and $a$ is fresh is also not possible, 
  since we initially assumed that $x \in \sccqpattern$.)
  Thus, at least one of the components $k$ must be defined; this proves the second direction.

  We thus prove both directions of this case.
\end{proof}

\begin{lemma}[Pert. to \Cref{trfm:def:cwa}.3]
   \label{trfm:lem:cwafour}
   For every $\graphext$ of a \eccqname $q$, it holds that $\shaclvalid{\graphext}{\Sigma_{3}}$.
\end{lemma}

\begin{proof}[Proof. Inspired by~\cite{DBLP:conf/www/Seifer0LS24}] 
  We consider both directions of this proof separately.

  \[\ddot{A} \sqsubseteq \bigsqcup_{\atomC{u}{A} \in H} \Concept{u} \sqcup \bigsqcup_{x\in\var(H)} \vconcept{x}^A\]

  Let $a$ be an arbitrary individual name such that $\{a\} \sqsubseteq \ddot{A}$ is valid in $\graphext$. 
  By construction, this means that $\atomC{a}{A} \in \graphout$.
  This atom can arise for two reasons:
  \begin{enumerate}
    \item There is an atom $\atomC{v}{A} \in H$ such that $\atomC{a}{A}$ is an instance of pattern 
        $\atomC{v}{A}$, and $v$ is either the concept name $a$ or a variable $x$. 
        If $v$ is $a$, then $\{a\} \sqsubseteq \Concept{v}$ is trivially valid in $\graphext$. 
        Otherwise, $v$ is $x$ such that $a$ is a binding of $x$ or freshly generated 
        (in the case where $x \not\in \sccqpattern$)
        and then it holds that $\atomC{a}{\vconcept{x}} \in \graphext$, 
        so $\{a\} \in \Concept{v}$ is valid in $\graphext$. 
        Therefore, $\{a\} \sqsubseteq \bigsqcup_{\atomC{u}{A}\in H} \Concept{u}$ is valid in $\graphext$.

    \item There is an atom $\atomGenc{x}$ in $\sccqpattern$ and $\sccqtemplate$, 
        with $\atomCNot{x}{A} \not\in F$, 
        and $a$ is an instance of $x$ (\ie $\atomC{a}{\vconcept{x}} \in \graphext$). 
        Since $\vconcept{x}^{A} \equiv \vconcept{x} \sqcap A$, 
        as well as $\{a\} \sqsubseteq \vconcept{x}$ and $\{a\} \sqsubseteq A$ and then also
        $\{a\} \sqsubseteq \bigsqcup_{x\in\var(H)} \vconcept{x}^A$ is valid in $\graphext$.
  \end{enumerate}

  Thus, at least one of the operands of the top-level disjunction holds for $\{a\}$ via either case 1 or 2, 
  and since $\{a\}$ was arbitrary, this direction of the equivalency holds in general.
  We consider the inverse direction next.

  \[\ddot{A} \sqsupseteq \bigsqcup_{\atomC{u}{A} \in H} \Concept{u} \sqcup \bigsqcup_{x\in\var(H)} \vconcept{x}^A\]

  Here, we consider the two components of the top-level disjunction separately and show 
  that in either case the property holds for all possible (\ie arbitrary) individuals $a$.

  \begin{enumerate}
    \item $\ddot{A} \sqsupseteq \bigsqcup_{\atomC{u}{A} \in H} \Concept{u}$.
      Let $a$ be an arbitrary individual name such that $\{a\} 
      \sqsubseteq \bigsqcup_{\atomC{u}{A} \in H} \Concept{u}$ is valid in $\graphext$.
      Then there is at least one atom 
      $\atomC{v}{A} \in H$ such that $\{a\} \sqsubseteq \Concept{v}$ is valid in $\graphext$.
      If $v$ is $a$, then $\atomC{a}{A} \in \graphout$, and thus $\atomC{a}{\ddot{A}} \in \graphext$.
      Otherwise, $v$ is a variable $x$, and $\atomC{a}{\vconcept{x}} \in \graphext$.
      By the definition of variable concepts, $\atomC{a}{A} \in \graphout$, 
      so $\atomC{a}{\ddot{A}} \in \graphext$.
      Therefore, $\{a\} \sqsubseteq \ddot{A}$ is valid in $\graphout$.

    \item $\ddot{A} \sqsupseteq \bigsqcup_{x\in\var(H)} \vconcept{x}^A$.
      Let $a$ be an arbitrary individual name with
      $\{a\} \sqsubseteq \bigsqcup_{x\in\var(H)} \vconcept{x}^A$ is valid in $\graphext$.
      Thus, $\{a\} \sqsubseteq \vconcept{x}^A$ for some $x \in \var(P)$.
      Since $\vconcept{x}^{A} \equiv \vconcept{x} \sqcap A$, therefore $\{a\} \sqsubseteq \vconcept{x}$ and $\{a\} \sqsubseteq A$ in $\graphext$ for that $x$,
      and there is an atom $\atomGenc{x}$ in P and H and $\atomCNot{x}{A} \not\in F$.
      By the definition of the semantics of generic patterns $\atomGenc{x}$, therefore
      $\atomC{a}{A} \in \graphout$, so $\atomC{a}{\ddot{A}} \in \graphext$.
  \end{enumerate}
  
  Therefore, the complete subsumption holds as well.
  Since both directions hold, we have thus shown the equivalency.
\end{proof}

\begin{lemma}[Pert. to \Cref{trfm:def:cwa}.4]
   \label{trfm:lem:cwasix}
   For every $\graphext$ of a \eccqname $q$, it holds that $\shaclvalid{\graphext}{\Sigma_{4}}$.
\end{lemma}

\begin{proof}[Proof. Inspired by~\cite{DBLP:conf/www/Seifer0LS24}] 
  Since the proofs for the inverse role cases are extremely similar, we consider only the first three axioms of \Cref{trfm:def:cwa}.4.
  We start by considering:

  \[\textstyle \bigsqcup_{\atomP{u}{v}{p} \in  \eccqtemplate} \Concept{u} \sqsubseteq \exists \ddot{p}.\Concept{v}\]

  Let $\Int$ be the canonical model of $\graphext$.
  By definition of the validation knowledge base of a graph, 
  $a^\Int \in (\exists \ddot{p}.\Concept{v})^\Int$ if and only if there exists an individual name $b$ 
  where $\atomP{a}{b}{\ddot{p}} \in \graphext$ and $(a^\Int,b^\Int) \in p^I$. 
  If there is an atomic pattern $\atomP{u}{v}{p} \in H$ where $u$ is the individual name $a$ or a variable $x$, 
  and $v$ is the individual name $b$ or a variable $y$ (and thus $a^I \in \Concept{u}^\Int$),
  then by construction, $\atomP{a}{b}{\ddot{p}} \in \graphext$.
  Thus, when $a^\Int \in \bigcup_{\atomP{u}{v}{p} \in H} \Concept{u}^\Int$, then
  $a^\Int \in (\exists \ddot{p}.\Concept{v})^\Int$.
  Hence, $\textstyle \bigsqcup_{\atomP{u}{v}{p} \in H}\Concept{u} \sqsubseteq \exists \ddot{p}.\Concept{v}$.

  Next, we consider the following rule:
  \[\exists \ddot{p}.\Concept{v} \sqsubseteq \textstyle\bigsqcup_{\atomP{u}{v}{p} \in  \eccqtemplate} \Concept{u} \sqcup \textstyle\bigsqcup_{\atomGenp{x} \in H} \vconcept{x}\]

  Let $\Int$ be the canonical model of $\graphext$.
  By definition of the validation knowledge base of a graph, 
  $a^\Int \in (\exists \ddot{p}.\Concept{v})^\Int$ if and only if there exists an individual name $b$ 
  where $\atomP{a}{b}{\ddot{p}} \in \graphext$ and $(a^\Int,b^\Int) \in p^I$. 
  By construction, $\atomP{a}{b}{\ddot{p}} \in \graphext$ arises in exactly two cases:

  \begin{enumerate}
    \item If there exists an atom $\atomP{u}{v}{p} \in H$ 
      where $u$ is the individual name $a$ or a variable $x$, 
      and $v$ is the individual name $b$ or a variable $y$ (and thus $a^I \in \Concept{u}^\Int$). 
      Then $a^\Int \in (\exists \ddot{p}.\Concept{v})^\Int$ 
      if $a^\Int \in \bigsqcup_{\atomP{u}{v}{p} \in H} \Concept{u}^\Int$. 
      Hence, $\exists \ddot{p}.\Concept{v} \sqsubseteq \textstyle \bigsqcup_{\atomP{u}{v}{p} \in H}\Concept{u}$.

    \item If there is an atom $\atomGenp{x}$ in H, with $\atomCNot{x}{p} \not\in F$
      (and thus $a^I \in \vconcept{x}^\Int$). 
      Then $a^\Int \in (\exists \ddot{p}.\Concept{v})^\Int$ 
      if $a^\Int \in \bigsqcup_{\atomP{u}{v}{p} \in H} \Concept{u}^\Int$. 
      Hence, $\exists \ddot{p}.\Concept{v} \sqsubseteq \textstyle\bigsqcup_{\atomGenp{x} \in H} \vconcept{x}$.
  \end{enumerate}

  In either case, one of the top-level disjunctions of the initial axiom holds; 
  thus the axiom is satisfied in all cases.
  Finally, we consider the third rule:
  \[\exists \ddot{p}.\top \equiv {\textstyle\bigsqcup_{\atomP{u}{v}{p} \in  \eccqtemplate}} (\Concept{u} \sqcap \exists \ddot{p}.\Concept{v}) \sqcup \textstyle\bigsqcup_{x\in\var(H)} \vconcept{x}^p\]

  By definition of the validation knowledge base of a graph, $a^\Int \in (\exists \ddot{p} . \top)^\Int$ 
  if and only if there exists an individual name $b$ where $\atomP{a}{b}{\ddot{p}} \in \graphext$ 
  and $(a^\Int, b^\Int) \in p^\Int$.
  By construction, the inclusion $\atomP{a}{b}{\ddot{p}} \in \graphext$ occurs in exactly two cases:
  \begin{enumerate}
    \item if there exists an atom $\atomP{u}{v}{p} \in H$ where $u$ is the individual name $a$ 
        or a variable $x$ and $v$ is the individual name $b$ or a variable $y$.
        Then, $a^\Int \in \Concept{v}^\Int$ and $a^\Int \in (\exists \ddot{p}.\Concept{v})^\Int$, 
        so $a^\Int \in (\exists \ddot{p}.\top)^\Int$ in this case
        if and only if $a^\Int \in {\bigsqcup_{\atomP{u}{v}{p} \in H}} \Concept{u}^\Int \sqcap (\exists \ddot{p}.\Concept{v})^\Int$.

    \item If there is an atom $\atomGenp{x}$ in P, with $\atomPNot{x}{p} \not\in \eccqfilter$ and $x$ in $\voc(H)$.
        In this case, $\vconcept{x}^{p} \equiv \vconcept{x} \sqcap \exists p . \vconcept{x}^{p,o}$.
        $a^\Int \in \vconcept{x}^{p\ \Int}$ and $a^\Int \in (\exists \ddot{p}.\Concept{v})^\Int$;
        for all $x,y \in \var(P)$ and $p \in \mathcal{V}$ 
        if $\atomGenp{x} \in \eccqpattern$, $\atomPNot{x}{p} \not\in \eccqfilter$
        and $\atomP{x}{y}{p}\not\in \eccqpattern$,
        otherwise $\vconcept{x}^{p} \equiv \bot$ and $\vconcept{x}^{p,o} \equiv \bot$
        so $a^\Int \in (\exists \ddot{p}.\top)^\Int$ in this case
        if and only if $a^\Int \in \textstyle\bigsqcup_{x\in\var(H)} \vconcept{x}^p$.
  \end{enumerate}
  Since at least one of these cases applies by construction, equivalency follows for their union.
\end{proof}

\begin{lemma}[Pert. to \Cref{trfm:def:cwa}.5]
   \label{trfm:lem:cwanine}
   For every $\graphext$ of a \eccqname $q$, it holds that $\shaclvalid{\graphext}{\Sigma_{5}}$.
\end{lemma}

\begin{proof}[Proof. Adapted from~\cite{DBLP:conf/www/Seifer0LS24}] 
  For every property $p \in \voc(\sccqpattern)$, 
  $\dot{p} \sqsubseteq p$ is true, since \Cref{trfm:def:extended-graph}
  guarantees that $\graphmed \subseteq \graphin$.
\end{proof}

\begin{lemma}[Pert. to \Cref{trfm:def:cwa}.6]
   \label{trfm:lem:cwaten}
   For every $\graphext$ of a \eccqname $q$, it holds that $\shaclvalid{\graphext}{\Sigma_{6}}$.
\end{lemma}

\begin{proof}[Proof. Adapted from~\cite{DBLP:conf/www/Seifer0LS24}] 
  Given $\sccqpattern$ contains $\atomP{x}{y}{p}$
  and neither $x$ nor $y$ occur in any other atomic patterns in $\sccqpattern$
  -- including any generic patterns such as $\atomGenp{x}$ --
  then from any $\atomP{a}{b}{p} \in \graphin$ 
  it follows that $\atomP{a}{b}{\dot{p}} \in \graphmed$,
  and thus $p \sqsubseteq \dot{p}$.
\end{proof}

\begin{lemma}[Pert. to \Cref{trfm:def:cwa}.7]
   \label{trfm:lem:cwaeleven}
   For every $\graphext$ of a \eccqname $q$, it holds that $\shaclvalid{\graphext}{\Sigma_{7}}$.
\end{lemma}

\begin{proof}[Proof. Adapted from~\cite{DBLP:conf/www/Seifer0LS24}] 
  Let 
  $p \in \voc(\sccqpattern)$,
  $r \in \voc(\sccqtemplate)$, and
  $\sccqpattern$ contain only the atomic pattern $\atomP{x}{y}{p}$;
  let $\sccqtemplate$ contain only $\atomP{x}{y}{r}$;
  and let neither $\sccqtemplate$ nor $\sccqpattern$ contain \emph{other} atomic patterns
  involving $x$, $y$, $p$ or $r$.
  If this is the case, for any $\atomP{a}{b}{\dot{p}} \in \graphmed$ matched in the intermediate graph, 
  we construct $\atomP{a}{b}{\ddot{r}} \in \graphout$ in the output graph.
  From this, we can infer that $\dot{p} \sqsubseteq \ddot{r}$.
  Since $r$ does not occur again in $\sccqtemplate$,
  it also follows that $\ddot{r} \sqsubseteq \dot{p}$ holds for all 
  extended graphs $\graphext$ of the query.
\end{proof}

\subsection[Query Extension]{Query Extension (\Cref{trfm:prop:semext})}
\label{trfm:proof:semext}

\begin{proof}
	We prove the property $\eccqeval{q}{G} = \eccqeval{\operatorname{Ext}(q)}{G}$ 
    for any \eccqname $q$ and RDF graph $G$.
	We consider, without loss of generality, a query pattern $\{\atomC{x}{A}, \atomGenc{x}\}$ 
    and the query template $\{\atomGenc{x}\}$,
	and further the two cases where either $\atomCNot{x}{A}$ in $\eccqfilter$, 
    or $\atomCNot{x}{A}$ is \emph{not} in $\eccqfilter$.

	\begin{enumerate}
		\item $\atomCNot{x}{A} \not\in \eccqfilter$. 
            Then, it follows from the semantics of \eccqname queries, that if $\atomC{x}{A}$ 
            is satisfied in a graph G,
			then the graph contains $\atomC{a}{A}$ for some $a$ where $a \in \mu(x)$;
			then $A \in \mu(\bar{x}^{\ConceptNames})$.
			As a result, the output graph also contains $\atomC{a}{A}$, 
            since $\atomGenc{x} \in \eccqtemplate$ (in the assumption above).
			Therefore, we can add the pattern $\atomGenc{x}{A}$ to the template of the query, 
            without changing the semantics of the query, as
			defined in $\operatorname{Ext}(q)$.

		\item $\atomCNot{x}{A} \in \eccqfilter$ in this case, 
            $\operatorname{Ext}(q)$ does not change the query $q$ given above, by definition,
			since the concept name of the potential candidate for extension, 
            $\atomC{x}{A}$, is in $\eccqfilter$.
	\end{enumerate}

	The property case can be shown to hold with the same reasoning.
\end{proof}

\end{document}